\documentclass[preprint,twoside,11pt]{article}

\usepackage{color}
\usepackage{amssymb}
\usepackage{amsbsy}
\usepackage{amsmath}
\usepackage{graphicx,psfrag,epsf}
\usepackage{subcaption}
\usepackage[shortlabels]{enumitem}
\usepackage{natbib}
\usepackage{amsfonts}
\usepackage{xifthen}
\usepackage{url} 
\usepackage{enumitem}
\usepackage{chngcntr}
\usepackage{algorithm,algorithmicx,algpseudocode}
\usepackage{booktabs}
\usepackage{morefloats}
\usepackage{jmlr2e}

\newcommand{\defn}{\ensuremath{: \, =}}

\newcommand{\matrixnorm}[1]{\left|\!\left|\!\left|{#1}\right|\!\right|\!\right|}

\newcommand{\lambdamin}{\lambda_{\tiny{\min}}}
\newcommand{\lambdamax}{\lambda_{\tiny{\max}}}

\newcommand{\R}{\mathbb{R}}

\newcommand{\Ee}{\mathbb{E}}

\newcommand{\bg}{\mathbf{g}}
\newcommand{\bh}{\mathbf{h}}

\newcommand{\cC}{\mathcal{C}}

\newcommand{\cG}{\mathcal{G}}

\newcommand{\cL}{\mathcal{L}}
\newcommand{\cLs}{\mathcal{L}^\ast}

\newcommand{\cN}{\mathcal{N}}
\newcommand{\cM}{\mathcal{M}}

\newcommand{\cZ}{\mathcal{Z}}

\DeclareMathOperator*{\argmin}{\arg\min}
\DeclareMathOperator{\cov}{cov}

\newcommand{\btheta}{\bar \theta}
\newcommand{\thetas} {\theta^\ast}
\newcommand{\htheta}{\widehat \theta}
\newcommand{\ttheta}{\widetilde \theta}

\newcommand{\tT}{\widetilde\Theta}

\renewenvironment{proof}[1][\proofname]{{\noindent\bfseries Proof of #1. }}{\hfill $\blacksquare$}

\makeatletter 
\newcommand{\labitemc}[2]{%
\def\@itemlabel{\textbf{#1}}
\item
\def\@currentlabel{#1}\label{#2}}
\makeatother

\usepackage[dvipsnames]{xcolor}

\jmlrheading{1}{2021}{1-48}{4/00}{10/00}{meila00a}{Yang Yu, Shih-Kang Chao and Guang Cheng}

\ShortHeadings{High-Dimensional Distributed Bootstrap}{Yu, Chao and Cheng}
\firstpageno{1}

\begin{document}

\title{Distributed Bootstrap for Simultaneous Inference Under High Dimensionality}

\author{\name Yang Yu \email yuyang930930@gmail.com \\
       \addr Department of Statistics\\
       Purdue University\\
       West Lafayette, IN 47907, USA
       \AND
       \name Shih-Kang Chao \email skchao74@gmail.com \\
       \addr Department of Statistics\\
       University of Missouri\\
       Columbia, MO 65211, USA
       \AND
       \name Guang Cheng\thanks{Part of this manuscript was completed while Cheng was at Purdue.} \email guangcheng@ucla.edu \\
       \addr Department of Statistics\\
       University of California, Los Angeles \\
       Los Angeles, CA 90095, USA}

\editor{Victor Chernozhukov}

\maketitle

\begin{abstract}
We propose a distributed bootstrap method for simultaneous inference on high-dimensional massive data that are stored and processed with many machines. The method produces an $\ell_\infty$-norm confidence region based on a communication-efficient de-biased lasso, and we propose an efficient cross-validation approach to tune the method at every iteration. We theoretically prove a lower bound on the number of communication rounds $\tau_{\min}$ that warrants the statistical accuracy and efficiency. Furthermore, $\tau_{\min}$ only increases logarithmically with the number of workers and the intrinsic dimensionality, while nearly invariant to the nominal dimensionality. We test our theory by extensive simulation studies, and a variable screening task on a semi-synthetic dataset based on the US Airline On-Time Performance dataset. The code to reproduce the numerical results is available at GitHub: \url{https://github.com/skchao74/Distributed-bootstrap}. 
\end{abstract}


\section{Introduction}

Modern massive datasets with enormous sample size and tremendous dimensionality are usually impossible to be processed with a single machine. For remedy, a master-worker architecture is often adopted, e.g., Hadoop \citep{singh2014hadoop}, which operates on a cluster of nodes for data storage and processing, where the master node also contains a portion of the data; see Figure~\ref{fig:master_slave}. An inherent problem of this architecture is that inter-node communication can be over a thousand times slower than intra-node computation due to the inter-node communication protocol, which unfortunately always comes with significant overhead \citep{lan2018communication,fan2019communication}. Hence, communication efficiency is usually a top concern for algorithm development in distributed learning.
\begin{figure}[ht]
\centering
\includegraphics[width=0.5\textwidth]{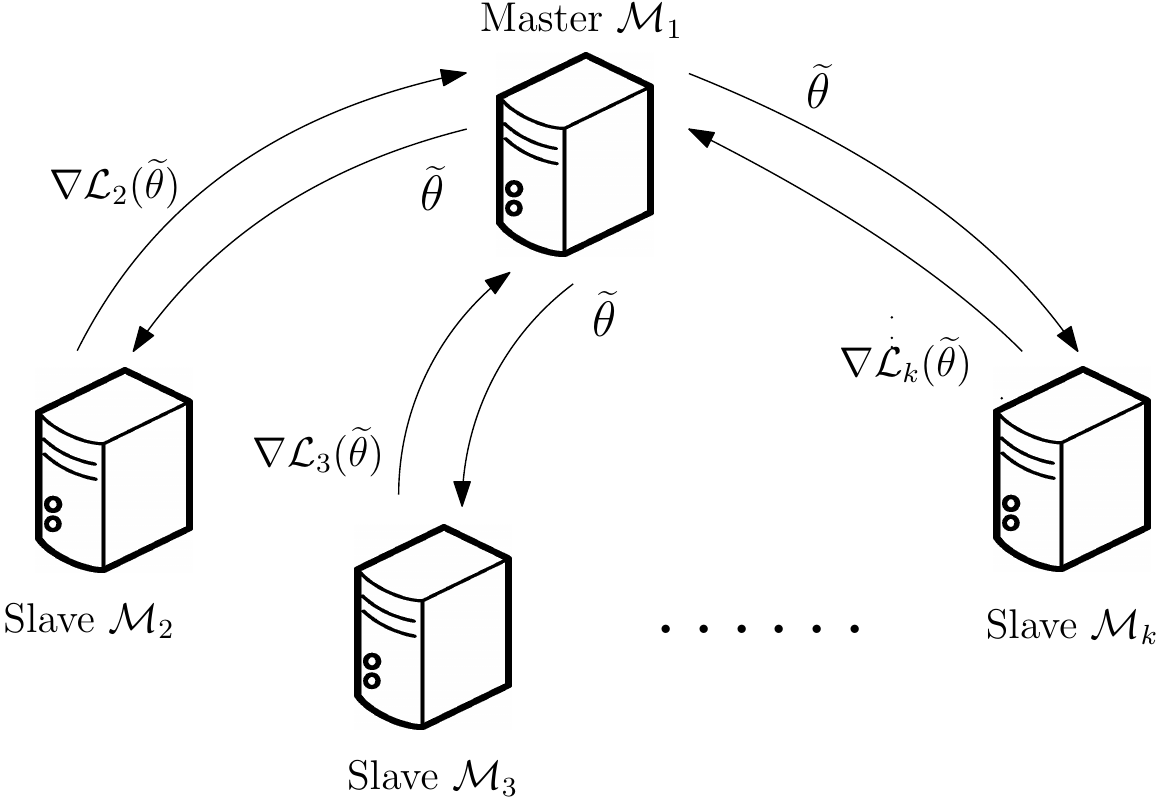}
\caption{Master-worker architecture for storing and processing distributed data.}
\label{fig:master_slave}
\end{figure}

Classical statistical methods are usually not communication-efficient as some of them require hundreds or even thousands passes over the entire dataset. In the last few years, active research has greatly advanced our ability to perform distributed statistical optimization and inference in, e.g., maximum likelihood estimation \citep{zhang2012communication,li2013statistical,chen2014split,battey2015distributed,jordan2019communication,huang2019distributed,chen2018first,ZLW20}, Lasso \citep{lee2017communication,wang2017efficient,wang2017improved}, partially linear models \citep{zhao2016partially}, nonstandard regression \citep{shi2018massive,banerjee2019divide}, quantile regression \citep{volgushev2019distributed,chen2019}, principal component analysis \citep{fan2019distributed,chen2020distributed}, just to name a few. However, solutions for many other problems in the distributed framework, for example the statistical inference for high-dimensional models, are still elusive.
 
Simultaneous inference for high-dimensional statistical models has been widely considered in many applications where datasets can be handled with a standalone computer \citep{cai2017large}, and many recent papers focus on bootstrap as an effective way to implement simultaneous inference \citep{dezeure2017high,zhang2017simultaneous,BCCW18,BCK19,yu2020simultaneous}. These existing methods typically use the well-celebrated de-biased Lasso \citep{van2014asymptotically,zhang2014confidence,javanmard2014confidence,javanmard2014hypothesis}, where the de-biased score results from the KKT condition of the Lasso optimization problem. However, de-biased Lasso is not directly applicable in a distributed computational framework. For one thing, the implementation of de-biased Lasso requires expensive subroutines such as nodewise Lasso \citep{van2014asymptotically}, which has to be replaced by a more communication-efficient method. For another, the quality of the de-biased score, which is essential to the validity of the bootstrap, is generally worse in a distributed computational framework than that in a centralized computational framework. In particular, it is heavily biased so the asymptotic normality fails. However, it can possibly be improved with sufficient rounds of communication between the master and worker nodes. The bootstrap validity therefore critically hinges on the interplay between the dimensionality of the model and the sparsity level, as well as the rounds of communication, the number of worker nodes and the size of local sample that are specific to the distributed computational framework. 

In this paper, we tackle the challenges discussed above and propose a communication-efficient simultaneous inference method for high-dimensional models. The main component at the core of our method is a novel way to improve the quality of the de-biased score with a carefully selected number of rounds of communication while relaxing the constraint on the number of machines. Our method is motivated by \cite{wang2017efficient}, who proposed an iterative procedure for computing the estimator but no statistical inference was provided. Note that the de-biased Lasso has been applied by \cite{lee2017communication} to obtain a communication-efficient $\sqrt{N}$-consistent estimator, but their method restricts the number of worker nodes to be less than the local sample size. Next, we apply communicate-efficient multiplier bootstrap methods \texttt{k-grad} and \texttt{n+k-1-grad}, which are originally proposed in \cite{ycc2020simultaneous} for low dimensional models. These bootstrap methods prevent repeatedly refitting the models and relax the constraint on the number of machines that plague the methods proposed earlier \citep{kleiner2014scalable,sengupta2016subsampled}. A key challenge in implementation is that cross-validation, which is a popular method for selecting tuning parameters, usually requires multiple passes of the entire dataset and is typically inefficient in the distributed computational framework. We propose a new cross-validation that only requires the master node for implementation without needing to communicate with the worker nodes.

Our theoretical study focuses on the explicit lower bounds on the rounds of communication that warrant the validity of the bootstrap method for high-dimensional generalized linear models; see Section \ref{sec:hd_over} for an overview. In short, the greater the number of worker nodes and/or the intrinsic dimensionality, the greater the rounds of communication required for the bootstrap validity. The bootstrap validity and efficiency are corroborated by an extensive simulation study.

We further demonstrate the merit of our method on variable screening with a semi-synthetic dataset, based on the large-scale US Airline On-Time Performance dataset. By performing a pilot study on an independently sampled subset of data, we take four key explanatory variables for flight delay, which correspond to the dummy variables of the four years after the September 11 attacks. On another independently sampled subset of data, we combine the dummy variables of the four years with artificial high-dimensional spurious variables to create a design matrix. We perform our method on this artificial dataset, and find that the relevant variables are correctly identified as the number of iteration increases. In particular, we visualize the effect of these four years by confidence intervals.

We go beyond our previous publication \cite{ycc2020simultaneous} in two major aspects: (1) In this paper we focus on high-dimensional models. In particular, the dimensionality of the model can exceed the sample size in each computing node. We handle high dimensionality using $\ell_1$ penalization, and consider de-biased Lasso under the distributed computational framework. (2) We tune the $\ell_1$ penalized problem with a carefully designed cross-validation method, which can be applied under distributed computational framework.

The rest of the paper is organized as follows. In Section~\ref{sec:method}, we introduce the problem formulation of distributed high-dimensional simultaneous inference and present the main bootstrap algorithm as well as the cross-validation algorithm for hyperparameter tuning. Theoretical guarantees of bootstrap validity for high-dimensional (generalized) linear models are provided in Section~\ref{sec:hd}. Section~\ref{sec:exp} presents simulation results that corroborate our theoretical findings. Section~\ref{sec:real} showcases an application on variable screening for high-dimensional logistic regression with a big real dataset using our new method. Finally, Section~\ref{sec:disc} concludes the paper. Technical details are in Appendices. 
The proofs of the theoretical results are in Supplementary Material. The code to reproduce the numerical results is in GitHub: \url{https://github.com/skchao74/Distributed-bootstrap}.

\noindent{\bf Notations.} We denote the $\ell_p$-norm ($p\geq 1$) of any vector $v=(v_1,\dots,v_n)$ by $\|v\|_p=(\sum_{i=1}^n |v_i|^p)^{1/p}$ and $\|v\|_\infty=\max_{1\leq i\leq n}|v_i|$. The induced $p$-norm and the max-norm of any matrix $M\in\R^{m\times n}$ (with element $M_{ij}$ at $i$-th row and $j$-th column) are denoted by $\matrixnorm{M}_p=\sup_{x\in\R^n;\|x\|_p=1} \|Mx\|_p$ and $\matrixnorm{M}_{\max}=\max_{1\leq i\leq m;1\leq j\leq n}|M_{i,j}|$. We write $a\lesssim b$ if $a=O(b)$, and $a\ll b$ if $a=o(b)$.

\section{Distributed Bootstrap for High-Dimensional Simultaneous Inference} \label{sec:method}

In this section, we introduce the distributed computational framework and present a novel bootstrap algorithm for high-dimensional simultaneous inference under this framework. A communication-efficient cross-validation method is proposed for tuning.

\subsection{Distributed Computation Framework}

Suppose data $\{Z_i\}_{i=1}^N$ are i.i.d., and $\cL(\theta;Z)$ is a twice-differentiable convex loss function arising from a statistical model, where  $\theta=(\theta_1,\dots,\theta_d)\in\R^d$. Suppose that the parameter of interest $\thetas$ is the minimizer of an expected loss: $$\thetas=\argmin_{\theta\in\R^d} \cLs(\theta), \mbox{  where $\cLs(\theta)\defn\Ee_Z[\cL(\theta;Z)]$}.$$ 
We consider a high-dimensional setting where $d>N$ is possible, and $\thetas$ is sparse, i.e., the support of $\thetas$ is small.

We consider a distributed computation framework, in which the entire data are stored distributedly in $k$ machines, and each machine has data size $n$. Denote by $\{Z_{ij}\}_{i=1,\dots,n; j=1,\dots,k}$ the entire data, where $Z_{ij}$ is $i$-th datum on the $j$-th machine $\cM_j$, and $N=nk$. Without loss of generality, assume that the first machine $\cM_1$ is the master node; see Figure \ref{fig:master_slave}. Define the local and global loss functions as 
\begin{align}
\begin{split}
    \mbox{global loss: }\cL_N(\theta)&=\frac1k\sum_{j=1}^k\cL_j(\theta),\quad\mbox{where}\\ \mbox{local loss: }\cL_j(\theta)&=\frac1n\sum_{i=1}^n\cL(\theta;Z_{ij}),\quad j=1,\dots,k. 
\end{split}
\label{eq:loss}
\end{align}
A great computational overhead occurs when the master and worker nodes communicate. 
In order to circumvent the overhead, the rounds of communications between the master and worker nodes should be minimized, and the algorithms with reduced communication overheads are ``communication-efficient''.

\subsection{High-Dimensional Simultaneous Inference}

In this paper, we focus on the simultaneous confidence region for $\thetas$ in a high-dimensional model, which is one of the effective ways for variable selection and inference that are immune to the well-known multiple testing problem. 
In particular, given an estimator $\htheta$ that is $\sqrt{N}$-consistent, simultaneous confidence intervals can be found with confidence $\alpha$, for large $\alpha\in(0,1)$, by finding the quantile
\begin{align}
c(\alpha)&\defn\inf\{t\in\R:P(\widehat T \leq t)\geq\alpha\} \quad\text{where} \label{eqn:c} \\
\widehat T&\defn \big\|\sqrt N\big(\htheta-\thetas\big)\big\|_\infty. \label{eqn:that}
\end{align}
where $\htheta$ may be computed through the de-biased Lasso \citep{van2014asymptotically,zhang2014confidence,javanmard2014confidence,javanmard2014hypothesis}:
\begin{align}
	\htheta = \htheta_{Lasso}-\widehat\Theta\nabla\cL_N(\htheta_{Lasso}), \label{eq:dblasso}
\end{align}
where 
$$\htheta_{Lasso}=\argmin_{\theta\in\R^d} \cL_N(\theta)+\lambda\|\theta\|_1$$
is the Lasso estimator with some hyperparameter $\lambda>0$, $\widehat\Theta$ is a surrogate inverse Hessian matrix and $\cL_N(\theta) = N^{-1}\sum_{i=1}^N \cL(\theta;Z_i)$ is the empirical loss.

Implementing the simultaneous inference based on $\htheta$ and $\widehat T$ in distributed computational framework inevitably faces some computational challenges. Firstly, computing $\htheta$ usually involves some iterative optimization routines that can accumulate a large communication overhead without a careful engineering. Next, some bootstrap methods have been proposed for estimating $c(\alpha)$, e.g., the multiplier bootstrap \citep{zhang2017simultaneous}, but they cannot be straightforwardly implemented within a distributed computational framework due to excessive resampling and communication. Even though some communication-efficient bootstrap methods have been proposed, e.g., \cite{kleiner2014scalable,sengupta2016subsampled,ycc2020simultaneous}, they either require a large number of machines or are inapplicable to high-dimensional models.

Because of the above-mentioned difficulties, inference based on $\widehat T$ is inapplicable in the distributed computational framework and is regarded as an ``oracle'' in this paper. Our goal is to provide a method that is communication-efficient while entertaining the same statistical accuracy as that based on the oracle $\widehat T$.

\subsection{High-Dimensional Distributed Bootstrap} \label{sec:reg_m}

In order to adapt \eqref{eq:dblasso} to the distributed computational setting, we first need to find a good substitute $\ttheta$ for $\htheta_{Lasso}$ that is communication-efficient, while noting that standard algorithms for Lasso are not communication-efficient. Fortunately, $\ttheta$ can be computed by the communication-efficient surrogate likelihood (CSL) algorithm with the $\ell_1$-norm regularization \citep{wang2017efficient,jordan2019communication}, which iteratively generates a sequence of estimators $\ttheta^{(t)}$ with regularization parameters $\lambda^{(t)}$ at each iteration $t=0,\dots,\tau-1$. See Remark \ref{rem:lam} for model tuning and Lines~\ref{line:2}-\ref{line:16} of Algorithm~\ref{alg:hd} for the exact implementation. Under regularity conditions, if $t$ is sufficiently large, it is warranted that $\ttheta$ is close to $\htheta_{Lasso}$.

Typical algorithms for computing $\widehat\Theta$, e.g., the nodewise Lasso \citep{van2014asymptotically}, cannot be extended straightforwardly to the distributed computational framework due to the same issue of communication inefficiency. We overcome this by performing the nodewise Lasso using only $\cM_1$ without accessing the entire dataset. This simple approach does not sacrifice accuracy as long as a sufficient amount of communication brings  $\ttheta$ sufficiently close to $\theta^*$.

Lastly, given the surrogate estimators $\ttheta$ for $\htheta_{Lasso}$ and $\tT$ for $\widehat\Theta$, we estimate the asymptotic quantile $c(\alpha)$ of $\widehat T$ by bootstrapping $\|\tT\sqrt N\nabla\cL_N(\ttheta)\|_\infty$ using the \texttt{k-grad} or \texttt{n+k-1-grad} bootstrap originally proposed by \citet{ycc2020simultaneous} for low-dimensional models. However, the number of communication rounds between master and worker nodes has to be carefully fine-tuned for high-dimensional models. In particular,  
the \texttt{k-grad} algorithm computes
\begin{align}
	\overline W^{(b)} \defn \bigg\|\underbrace{-\tT\frac1{\sqrt{k}}\sum_{j=1}^k\epsilon_j^{(b)}\sqrt n(\bg_j-\bar\bg)}_{=:\overline A}\bigg\|_\infty, \label{eqn:wb}
\end{align}
where $\epsilon_j^{(b)}\overset{\text{i.i.d.}}{\sim}\cN(0,1)$ independent from the data, $\bg_j=\nabla\cL_j(\ttheta)$ and $\bar\bg = k^{-1}\sum_{j=1}^k \bg_j$. However, it is known that \texttt{k-grad} does not perform well when $k$ is small \citep{ycc2020simultaneous}. The improved algorithm \texttt{n+k-1-grad} computes
\begin{align}
\begin{split}
    	\widetilde W^{(b)}\defn \bigg\|&\underbrace{-\tT\frac1{\sqrt{n+k-1}}\bigg(\sum_{i=1}^n\epsilon_{i1}^{(b)}(\bg_{i1}-\bar\bg) +\sum_{j=2}^k\epsilon_j^{(b)}\sqrt n(\bg_j-\bar\bg) \bigg)}_{=:\widetilde A}\bigg\|_\infty,
\end{split}
\label{eqn:wt}
\end{align}
where $\epsilon_{i1}^{(b)}$ and $\epsilon_j^{(b)}$ are i.i.d.\ $\cN(0,1)$ multipliers, and $\bg_{i1}=\nabla\cL(\ttheta;Z_{i1})$ is based on a single datum $Z_{i1}$ in the master. The key advantage of \texttt{k-grad} or \texttt{n+k-1-grad} is that once the master has the gradients $\bg_j$ from the worker nodes, the quantile of $\{{\overline W}^{(b)}\}_{b=1}^B$ can be computed in the master node only, without needing to communicate with worker nodes. See Algorithm~\ref{alg:kgrad} in the Appendix for the pseudocode of \texttt{k-grad} and \texttt{n+k-1-grad}.

\begin{algorithm}[ht!]
\caption{\texttt{k-grad}/\texttt{n+k-1-grad} with de-biased $\ell_1$-CSL estimator}\label{alg:hd}
\begin{algorithmic}[1]
\Statex {\bfseries Require:} $\tau\geq 1$ rounds of communication; hyperparameters $\{\lambda^{(t)}\}_{t=0}^{\tau-1}$
, nodewise Lasso procedure \texttt{Node}$(\cdot,\cdot)$ with hyperparameters $\{\lambda_l\}_{l=1}^d$ (see Section~\ref{sec:node})
\State $\ttheta^{(0)}\gets\argmin_\theta \cL_1(\theta)+\lambda^{(0)}\|\theta\|_1$ at $\cM_1$ \label{line:2}
\State Compute $\tT$ by running \texttt{Node}$(\nabla^2\cL_1(\ttheta^{(0)}),\{\lambda_l\}_{l=1}^d)$ at $\cM_1$ 
\For{$t = 1,\ldots, \tau $} 
\State Transmit $\ttheta^{(t-1)}$ to $\{\cM_j\}_{j=2}^k$
\State Compute $\nabla\cL_1(\ttheta^{(t-1)})$ at $\cM_1$
\For{$j = 2,\ldots, k $}
\State Compute $\nabla\cL_j(\ttheta^{(t-1)})$ at $\cM_j$
\State Transmit $\nabla\cL_j(\ttheta^{(t-1)})$ to $\cM_1$
\EndFor
\State $\nabla\cL_N(\ttheta^{(t-1)})\gets k^{-1}\sum_{j=1}^k\nabla\cL_j(\ttheta^{(t-1)})$ at $\cM_1$
\If{$t < \tau$}
\State $\ttheta^{(t)}\gets\argmin_\theta \cL_1(\theta)-\theta^\top\left(\nabla\cL_1(\ttheta^{(t-1)})-\nabla\cL_N(\ttheta^{(t-1)})\right)+\lambda^{(t)}\|\theta\|_1$ at $\cM_1$ \label{line:srg}
\Else
\State $\ttheta^{(\tau)}\gets\ttheta^{(\tau-1)}-\tT\nabla\cL_N(\ttheta^{(\tau-1)})$ at $\cM_1$ \label{line:db}
\EndIf
\EndFor \label{line:16}
\State Run \texttt{DistBoots}$(\text{`\texttt{k-grad}' or `\texttt{n+k-1-grad}'},\ttheta=\ttheta^{(\tau)},\{\bg_j=\nabla\cL_j(\ttheta^{(\tau-1)})\}_{j=1}^k,$ \\
\hspace{90pt}$\tT=\tT)$ at $\cM_1$
\end{algorithmic}
\end{algorithm}

Algorithm~\ref{alg:hd} presents the complete statistical inference procedure. There are two key innovative steps in Algorithm 1 that facilitate the statistical inference for high dimensional model with a big dataset. First, we introduce de-biased Lasso in distributed inference, which goes beyond high dimensional model estimation considered in \citet{jordan2019communication,wang2017efficient}. Second, we use nodewise Lasso to provide a sparse estimation of the high-dimensional inverse Hessian matrix instead of the empirical Hessian used in \citet{ycc2020simultaneous}.

Algorithm 1 can achieve high computational efficiency due to two reasons. First, we initialize Algorithm 1 with a warm start. Namely, we warm start with the Lasso estimator estimated with dataset in the master node, which provides a good initializer. Second, because the nodewise Lasso is computationally expensive, we perform it only once at the very beginning and freeze it through the iterations of the algorithm without updating it.

The number of iterations $\tau$ in Algorithm~\ref{alg:hd} steers the trade-off between statistical accuracy and communication efficiency. In particular, a larger $\tau$ leads to a more accurate coverage of the simultaneous confidence interval, but it also induces a higher communication cost. Therefore, studying the minimal $\tau$ that warrants the bootstrap accuracy is crucial, which is done in Section~\ref{sec:hd}. 

\begin{remark} \label{rem:lam}
Two groups of hyperparameters need to be chosen in Algorithm~\ref{alg:hd}: $\{\lambda^{(t)}\}_{t=0}^{\tau-1}$ for regularization in CSL estimation, and $\{\lambda_l\}_{l=1}^d$ for regularization in nodewise Lasso (see Algorithm~\ref{alg:node}). In Section~\ref{sec:cv}, we propose a cross-validation method for tuning $\{\lambda^{(t)}\}_{t=0}^{\tau-1}$.  As to $\{\lambda_l\}_{l=1}^d$, while \citet{van2014asymptotically} suggests to choose the same value for all $\lambda_l$ by cross-validation, a potentially better way may be to allow $\lambda_l$ to be different across $l$ and select each $\lambda_l$ via cross-validation for the corresponding nodewise Lasso, which is the approach we take for a distributed variable screening task in Section \ref{sec:real}. 
\end{remark}

\begin{remark}
There exist other options than CSL for $\ttheta$ such as the averaging de-biased estimator \citep{lee2017communication}, but an additional round of communication may be needed to compute the local gradients. More importantly, their method may be inaccurate when $n<k$.
\end{remark}

\subsection{Communication-Efficient Cross-Validation}\label{sec:cv}

We propose a communication-efficient cross-validation method for tuning the hyperparameters $\{\lambda^{(t)}\}_{t=0}^{\tau-1}$ in Algorithm~\ref{alg:hd}. \citet{wang2017efficient} proposes to hold out a validation set on each node for selecting $\lambda^{(t)}$. However, this method requires fitting the model for each candidate value of $\lambda^{(t)}$, which uses the same communication cost as the complete CSL estimation procedure. 

We propose a communication-efficient $K$-fold cross-validation method that chooses $\lambda^{(t)}$ for the CSL estimation at every iteration $t$. At iteration $t$, the master uses the gradients already communicated from the worker nodes at iteration $t-1$. Hence, the cross-validation needs only the master node, which circumvents costly communication between the master and the worker nodes.

Specifically, notice that the surrogate loss (see Line~\ref{line:srg} in Algorithm~\ref{alg:hd}) is constructed using $n$ observations $\cZ=\{Z_{i1}\}_{i=1}^n$ in the master node and $k-1$ gradients $\cG=\{\nabla\cL_j(\ttheta^{(t-1)})\}_{j=2}^k$ from the worker nodes. We then create $K$ (approximately) equal-size partitions to both $\cZ$ and $\cG$. The objective function for training is formed using $K-1$ partitions of $\cZ$ and $\cG$. In terms of the measure of fit, instead of computing the original likelihood or loss, we calculate the unregularized surrogate loss using the last partition of $\cZ$ and $\cG$, still in the master node. See Algorithm~\ref{alg:cv} for the pseudocode.

\begin{algorithm}[ht!]
\caption{Distributed $K$-fold cross-validation for $t$-step CSL}\label{alg:cv}
\begin{algorithmic}[1]
\Statex {\bfseries Require:} $(t-1)$-step CSL estimate $\ttheta^{(t-1)}$, set $\Lambda$ of candidate values for $\lambda^{(t)}$, partition of master data $\cZ=\bigcup_{q=1}^K\cZ_{q}$, partition of worker gradients $\cG=\bigcup_{q=1}^K\cG_q$
\For{$q = 1,\dots, K $}
\State $\cZ_{train}\gets\bigcup_{r\ne q}\cZ_r$; \quad$\cZ_{test}\gets\cZ_q$
\State $\cG_{train}\gets\bigcup_{r\ne q}\cG_r$; \quad$\cG_{test}\gets\cG_q$
\State $g_{1,train}\gets\text{Avg}_{Z\in\cZ_{train}}\Big(\nabla\cL(\ttheta^{(t-1)};Z)\Big)$; \quad$g_{1,test}\gets\text{Avg}_{Z\in\cZ_{test}}\Big(\nabla\cL(\ttheta^{(t-1)};Z)\Big)$
\State $\bar g_{train}\gets\text{Avg}_{g\in\{g_{1,train}\}\cup\cG_{train}}(g)$; \quad$\bar g_{test}\gets\text{Avg}_{g\in\{g_{1,test}\}\cup\cG_{test}}(g)$
\For{$\lambda\in\Lambda_t$}
\State $\beta\gets\argmin_\theta \text{Avg}_{Z\in\cZ_{train}}\big(\cL(\theta;Z)\big)-\theta^\top\left(g_{1,train}-\bar g_{train}\right)+\lambda\|\theta\|_1$
\State $Loss(\lambda,q)\gets\text{Avg}_{Z\in\cZ_{test}}\big(\cL(\beta;Z)\big)-\beta^\top\left(g_{1,test}-\bar g_{test}\right)$
\EndFor
\EndFor
\State Return $\lambda^{(t)}=\argmin_{\lambda\in\Lambda}K^{-1}\sum_{q=1}^K Loss(\lambda,q)$
\end{algorithmic}
\end{algorithm}

\section{Theoretical Analysis} \label{sec:hd}

Section \ref{sec:hd_over} provides an overview of the theoretical results. Sections \ref{sec:hd_lm} and \ref{sec:hd_glm} presents the rigorous statements for linear models and generalized linear models (GLMs) respectively.

\subsection{An Overview of Theoretical Results}\label{sec:hd_over}

As discussed in Section~\ref{sec:reg_m}, $\tau$ has to be large enough to ensure the bootstrap accuracy, yet it also induces a great communication cost. Hence, our main goal is to pin down the minimal number of iterations $\tau_{\min}$ (communication rounds) sufficient for the bootstrap validity in Algorithm~\ref{alg:hd}. An overview of the theoretical results is provided in Figure~\ref{fig:tau_hd}.

As an overall trend in Figure \ref{fig:tau_hd}, $\tau_{\min}$ is increasing logarithmically in $k$ and decreasing in $n$ for both \texttt{k-grad} and \texttt{n+k-1-grad} in (generalized) linear models; in addition, $\tau_{\min}$ is increasing in $\overline s$ logarithmically, where $\overline s$ is the maximum of the sparsity of the true coefficient vector and the inverse population Hessian matrix to be formally defined later. 

By comparing the left and right panels of Figure \ref{fig:tau_hd} under a fixed tuple $(n,k,\overline s)$, the $\tau_{\min}$ for \texttt{k-grad} is always greater or equal to that for \texttt{n+k-1-grad}, which indicates a greater communication efficiency of \texttt{n+k-1-grad}. For very small $k$, \texttt{n+k-1-grad} can still provably work, while \texttt{k-grad} cannot. Particularly, $\tau_{\min}=1$ can work for certain instances of \texttt{n+k-1-grad} but is always too small for \texttt{k-grad}.

Regarding the comparison between high-dimensional sparse linear models (top panels) and GLMs (bottom panels), GLMs typically require a greater $n$ than sparse linear models, which ensures that the error between $\ttheta^{(t)}$ and $\thetas$ decreases in a short transient phase; see Section \ref{rem:wang} in the Appendix for details.
\begin{figure}[ht!]
\begin{center}
\centerline{\includegraphics[width=0.8\columnwidth]{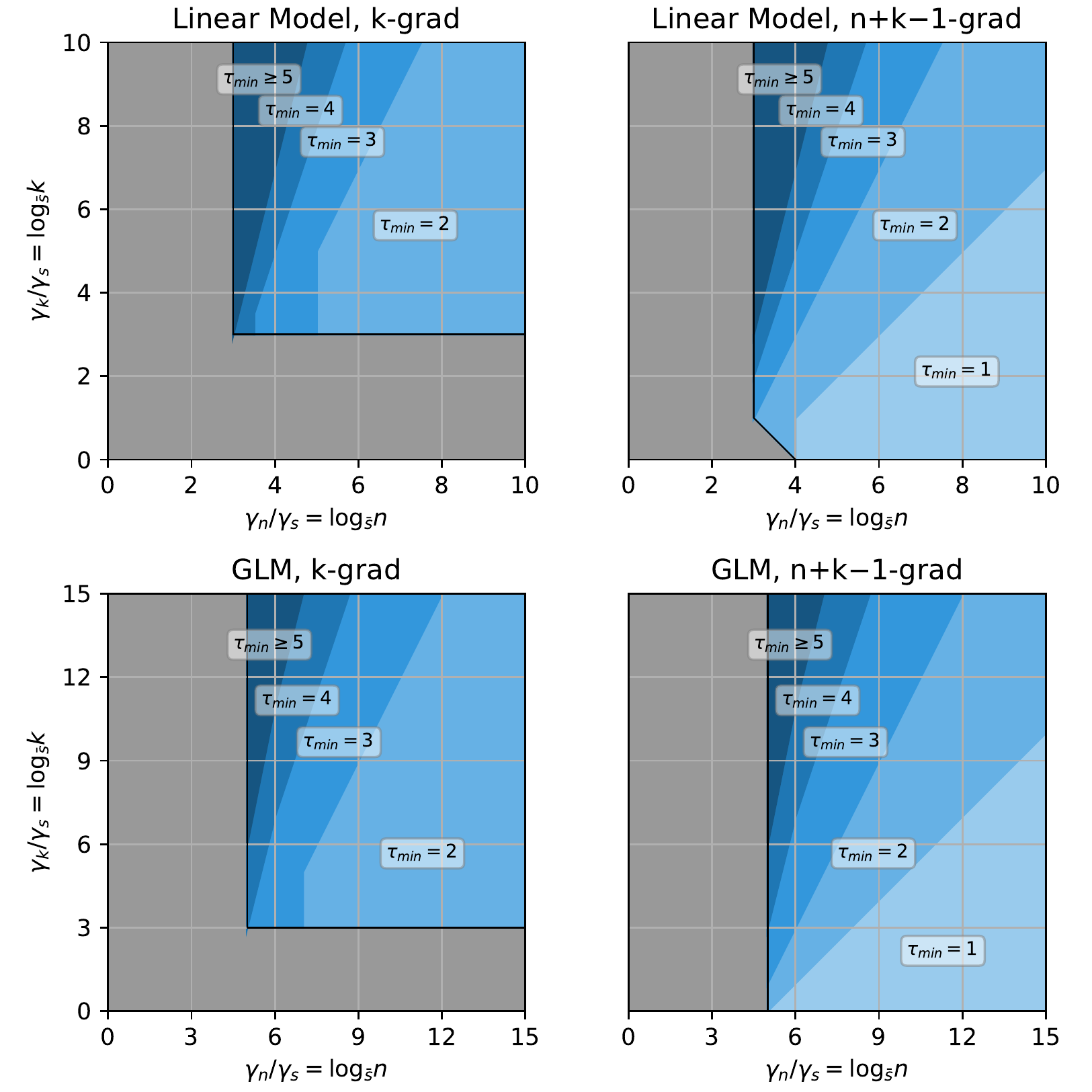}}
\caption{Illustration of Theorems \ref{theo:reg0_csl}-\ref{theo:reg_glm_csl}. Gray region are where the bootstrap validity are not warranted by our theory, and the other area is colored blue with varying lightness according to the lower bound of iteration $\tau$. $\gamma_n = \log_d n$, $\gamma_k = \log_d k$ and $\gamma_{\bar s}=\log_d \bar s$ are the orders of the local sample size $n$, number of machines $k$ and the sparsity $\bar s$.}
\label{fig:tau_hd}
\end{center}
\end{figure}
\subsection{Linear Model} \label{sec:hd_lm}

Suppose that $N$ i.i.d.\ observations are generated by a linear model $y=x^\top\thetas+e$ with an unknown coefficient vector $\thetas\in\R^d$, covariate random vector $x\in\R^d$, and noise $e\in\R$ independent of $x$ with zero mean and variance of $\sigma^2$. We consider the least-squares loss $\cL(\theta;z)=\cL(\theta;x,y)=(y-x^\top\theta)^2/2$.

We impose the following assumptions on the linear model.
\begin{itemize}
	\labitemc{(A1)}{as:design} $x$ is sub-Gaussian, i.e.,
	$$\sup_{\|w\|_2\leq1}\Ee\big[\exp((w^\top x)^2/L^2)\big]=O(1),$$
	for some absolute constant $L>0$.  Moreover, $1/\lambdamin(\Sigma)\leq\mu$ for some absolute constant $\mu>0$, where $\Sigma=\Ee[xx^\top]$.
	
	\labitemc{(A2)}{as:noise} $e$ is sub-Gaussian, i.e.,
	$$\Ee\big[\exp(e^2/L'^2)\big]=O(1),$$
	for some absolute constant $L'>0$.  Moreover, $\sigma>0$ is an absolute constant.
	
	\labitemc{(A3)}{as:sparse} 
	$\thetas$ and $\Theta_{l,\cdot}$ are sparse for $l=1,\cdots,d$, where $\Theta\defn\Sigma^{-1}=\Ee[xx^\top]^{-1}$. Specifically, we denote by $S\defn\{l:\thetas_l\neq0\}$ the active set of covariates and its cardinality by $s_0\defn|S|$.  Also, we define $s_l\defn|\{l'\neq l:\Theta_{l,l'}\neq0\}|$, $s^*\defn\max_l s_l$, and $\overline s=s_0\vee s^*$.
\end{itemize}

Assumption \ref{as:design} ensures a restricted eigenvalue condition when $n \gtrsim \bar s \log d$ by \cite{RZ13}. Under the assumptions, we first investigate the theoretical property of Algorithm~\ref{alg:hd}, where we apply \texttt{k-grad} with the de-biased $\ell_1$-CSL estimator with $\tau$ communications. Define
\begin{align}
    T&\defn \big\|\sqrt N\big(\ttheta^{(\tau)}-\thetas\big)\big\|_\infty,\label{eqn:t}
\end{align}
where $\ttheta^{(\tau)}$ is an output of Algorithm~\ref{alg:hd}. 

\begin{theorem}[\texttt{k-grad}, sparse linear model]\label{theo:reg0_csl}
	Suppose \ref{as:design}-\ref{as:sparse} hold, and that we run Algorithm \ref{alg:hd} with \texttt{k-grad} method in linear models. Let
	\begin{align}
	\lambda_l\asymp\sqrt{\frac{\log d}n} \quad\text{and}\quad \lambda^{(t)}\asymp\sqrt{\frac{\log d}{nk}}+\sqrt{\frac{\log d}n}\bigg(s_0\sqrt{\frac{\log d}n}\bigg)^t, \label{eqn:lam}
	\end{align}
	for $l=1,\dots,d$ and $t=0,\dots,\tau-1$.  Assume $n=d^{\gamma_n}$, $k=d^{\gamma_k}$, $\overline s=d^{\gamma_s}$ for some constants $\gamma_n,\gamma_k,\gamma_s>0$. If $\gamma_n>3\gamma_s$, $\gamma_k>3\gamma_s$, and $\tau\geq\tau_{\min}$, where
	\begin{align*}
	\tau_{\min}=1+\left\lfloor\max\left\{\frac{\gamma_k+\gamma_s}{\gamma_n-2\gamma_s},1+\frac{3\gamma_s}{\gamma_n-2\gamma_s}\right\}\right\rfloor,
	\end{align*}
	then for $T$ defined in \eqref{eqn:t}, we have
	\begin{align}
	\sup_{\alpha\in(0,1)}|P(T\leq c_{\overline W}(\alpha))-\alpha|=o(1). \label{eqn:kgradthm}
	\end{align}
	where $c_{\overline W}(\alpha)\defn\inf\{t\in\R:P_\epsilon(\overline W\leq t)\geq\alpha\}$, in which $\overline W$ is the \texttt{k-grad} bootstrap statistics with the same distribution as $\overline W^{(b)}$ in \eqref{eqn:wb} and $P_\epsilon$ denotes the probability with respect to the randomness from the multipliers.
	
	In addition, \eqref{eqn:kgradthm} also holds if $T$ is replaced by $\widehat T$ defined in \eqref{eqn:that}. 
\end{theorem}

Theorem \ref{theo:reg0_csl} warrants the bootstrap validity for the simultaneous confidence intervals produced by Algorithm \ref{alg:hd} with the \texttt{k-grad}. Furthermore, it also suggests that the bootstrap quantile can approximates the quantile of the oracle statistics $T$; that is, our distributed bootstrap procedure is as statistically efficient as the oracle centralized method.

\bigskip
Next, we show that the same distributed bootstrap validity and the efficiency of the \texttt{k-grad} also hold for the \texttt{n+k-1-grad} in Algorithm \ref{alg:hd}. 

\begin{theorem}[\texttt{n+k-1-grad}, sparse linear model]\label{theo:reg_csl}
	Suppose \ref{as:design}-\ref{as:sparse} hold, and that we run Algorithm \ref{alg:hd} with \texttt{n+k-1-grad} method.  Let $\lambda_l$ and $\lambda^{(t)}$ be as in \eqref{eqn:lam} for $l=1,\dots,d$ and $t=0,\dots,\tau-1$.  Assume $n=d^{\gamma_n}$, $k=d^{\gamma_k}$, $\overline s=d^{\gamma_s}$ for some constants $\gamma_n,\gamma_k,\gamma_s>0$. If $\gamma_n>3\gamma_s$, $\gamma_n+\gamma_k>4\gamma_s$, and $\tau\geq\tau_{\min}$, where
	\begin{align*}
	\tau_{\min}=1+\left\lfloor\frac{(\gamma_k\vee\gamma_s)+\gamma_s}{\gamma_n-2\gamma_s}\right\rfloor,
	\end{align*}
	then for $T$ defined in \eqref{eqn:t}, we have
	\begin{align}
	\sup_{\alpha\in(0,1)}|P(T\leq c_{\widetilde W}(\alpha))-\alpha|=o(1). \label{eqn:nk1gradthm}
	\end{align}
	where $$c_{\widetilde W}(\alpha)\defn\inf\{t\in\R:P_\epsilon(\widetilde W\leq t)\geq\alpha\},$$ in which $\widetilde W$ is the \texttt{n+k-1-grad} bootstrap statistics with the same distribution as $\widetilde W^{(b)}$ in \eqref{eqn:wt} and $P_\epsilon$ denotes the probability with respect to the randomness from the multipliers.
	
	In addition, \eqref{eqn:nk1gradthm} also holds if $T$ is replaced by $\widehat T$ defined in \eqref{eqn:that}.
\end{theorem}

Note by Theorem 2.4 of \citet{van2014asymptotically} that $\widehat T$ is well approximated by $\|\widehat\Theta\sqrt N\nabla\cL_N(\thetas)\|_{\infty}$ , which is further approximated by the $\ell_{\infty}$-norm of the oracle score
$$
A = -\Theta\frac1{\sqrt N}\sum_{i=1}^n\sum_{j=1}^k\nabla\cL(\thetas;Z_{ij}),
$$
given that $\widehat\Theta$ only deviates from $\Theta$ up to order $O_P(s^*(\log d)^{1/2}N^{-1/2})$ in $\ell_\infty$-norm. To gain a deeper look into the efficiency of \texttt{k-grad} and \texttt{n+k-1-grad}, we compare the difference between the covariance of $A$ and the conditional covariance of $\overline A$ (for \texttt{k-grad}, defined in \eqref{eqn:wb}), and $\widetilde A$ (for \texttt{n+k-1-grad}, defined in \eqref{eqn:wt}). In particular, 
conditioning on the data $Z_{ij}$, we have
\begin{align}
\matrixnorm{\cov_\epsilon(\overline A)-\cov(A)}_{\max}\leq &s^*\|\ttheta^{(\tau-1)}-\thetas\|_1
+ns^*\|\ttheta^{(\tau-1)}-\thetas\|_1^2 \notag\\
&+ O_P\bigg(\sqrt{\frac{{s^*}^2}k}+\sqrt{\frac{s^*}n}\bigg), \label{eqn:kgrad_err_hd} \\
\matrixnorm{\cov_\epsilon(\widetilde A)-\cov(A)}_{\max}\leq &s^*\|\ttheta^{(\tau-1)}-\thetas\|_1
+(n\wedge k)s^*\|\ttheta^{(\tau-1)}-\thetas\|_1^2 \notag\\
&+ O_P\bigg(\sqrt{\frac{{s^*}^2}{n+k}}+\sqrt{\frac{s^*}n}\bigg), 
\label{eqn:nk1grad_err_hd}
\end{align}
up to some logarithmic terms in $d$, $n$ or $k$. Overall, \texttt{n+k-1-grad} in \eqref{eqn:nk1grad_err_hd} has a smaller error term than that of \texttt{k-grad} in \eqref{eqn:kgrad_err_hd}. In particular, \texttt{k-grad} requires both $n$ and $k$ to be large, while \texttt{n+k-1-grad} requires a large $n$ but not necessarily a large $k$. In addition, $\tau=1$ could be enough for \texttt{n+k-1-grad}, but not for \texttt{k-grad}. To see it, if $\|\ttheta^{(0)}-\thetas\|_1$ is of order $O_P(s^*/\sqrt n)$, the right-hand side of \eqref{eqn:kgrad_err_hd} can grow with $s^*$, while the error in \eqref{eqn:nk1grad_err_hd} still shrinks to zero as long as $k\ll n$.

\begin{remark}
Note in both Theorems~\ref{theo:reg0_csl}~and~\ref{theo:reg_csl} that the expression of $\tau_{\min}$ does not depend on $d$, because the direct effect of $d$ only enters through an iterative logarithmic term $\log\log d$ which is dominated by $\log\overline s \asymp \log d$.
\end{remark}

\begin{remark}
The rates of $\{\lambda^{(t)}\}_{t=0}^{\tau-1}$ and $\{\lambda_l\}_{l=1}^d$ in Theorems~\ref{theo:reg0_csl}~and~\ref{theo:reg_csl} are motivated by those in \citet{wang2017efficient} and \citet{van2014asymptotically}, which, unfortunately, are not useful in practice. We therefore provide a practically useful cross-validation method in Section \ref{sec:cv}.
\end{remark}

\begin{remark}
The main result (Theorem 2.2) in \citet{zhang2017simultaneous} can be seen as a justification of multiplier bootstrap for high-dimensional linear models with data being processed in a centralized manner. Theorem~\ref{theo:reg_csl} compliments it by justifying a distributed multiplier bootstrap with at least one round of communication ($\tau\ge1$).
\end{remark}

\begin{remark}
A rate of $\sup_{\alpha\in(0,1)}\left|P(T\leq c_{\overline W}(\alpha))-\alpha\right|$ may be shown to be polynomial in $n$ and $k$ with a more careful analysis, which is faster than the order obtained by the extreme value distribution approach \citep{chernozhukov2013gaussian,zhang2017simultaneous} that is at best logarithmic.
\end{remark}

\begin{remark}
We have not addressed the question of whether the conditions for $\tau_{\min}$ in Theorem \ref{theo:reg0_csl} and \ref{theo:reg_csl} can be improved in a minimax sense. This is left for future research. On the other hand, we remark that the total communication cost in our algorithm is of order $\Omega(\tau_{\min} kd)$, because in each iteration we communicate $d$-dimensional vectors between the master node and $k-1$ worker nodes, and $\tau_{\min}$ only grows logarithmically with $k$. Our order matches those in the existing communication-efficient statistical inference literature e.g., \citet{jordan2019communication,wang2017efficient}.
\end{remark}

\subsection{Generalized Linear Model} \label{sec:hd_glm}

In this section, we consider GLMs, which generate i.i.d.\ observations $(x,y)\in\R^d\times\R$.
We assume that the loss function $\cL$ is of the form $\cL(\theta;z)=g(y,x^\top\theta)$ for $\theta,x\in\R^d$ and $y\in\R$ with $g:\R\times\R\to\R$, and $g(a,b)$ is three times differentiable with respect to $b$, and denote $\frac{\partial}{\partial b}g(a,b)$, $\left(\frac{\partial}{\partial b}\right)^2 g(a,b)$, $\left(\frac{\partial}{\partial b}\right)^3 g(a,b)$ by $g'(a,b)$, $g''(a,b)$, $g'''(a,b)$ respectively.  We let $\thetas$ be the unique minimizer of the expected loss $\cLs(\theta)$. 

We let $X_1\in\R^{n\times d}$ be the design matrix in the master node $\cM_1$ and $X_1^*\defn P^*X_1$ be the weighted design matrix with a diagonal $P^*\in\R^{n\times n}$ with elements $\{g''(y_{i1},x_{i1}^\top\thetas)^{1/2}\}_{i=1,\dots,n}$.  We further let $(X_1^*)_{-l}\varphi^*_l$ be the $L_2$ projection of $(X_1^*)_l$ on $(X_1^*)_{-l}$, for $l=1,\dots,d$. Equivalently, for $l=1,\dots,d$, we define $\varphi^*_l\defn\argmin_{\varphi\in\R^{d-1}}\Ee[\|(X_1^*)_l-(X_1^*)_{-l}\varphi\|_2^2]$.

We impose the following assumptions on the GLM.
\begin{itemize}
	\labitemc{(B1)}{as:smth_glm} For some $\Delta>0$, and $\Delta'>0$ such that $|x^\top\thetas|\leq\Delta'$,
	\begin{align*}
	\sup_{|b|\vee|b'|\leq\Delta+\Delta'}&\sup_a\frac{|g''(a,b)-g''(a,b')|}{|b-b'|}\leq1, \\
	\max_{|b_0|\leq\Delta} &\sup_a |g'(a,b_0)|=O(1),\quad\text{and}\quad \max_{|b|\leq\Delta+\Delta'} \sup_a |g''(a,b)|=O(1).
	\end{align*}
	
	\labitemc{(B2)}{as:design_glm} $\|x\|_\infty=O(1)$. Moreover, $x^\top\thetas=O(1)$ and $\max_l\big|g''(y,x^\top\thetas)^{1/2}x_{-l}^\top\varphi^*_l\big|=O(1)$, where $x_{-l}$ consists of all but the $l$-th coordinate of $x$.
	
	\labitemc{(B3)}{as:hes_glm} The least and the greatest eigenvalues of $\nabla^2\cLs(\thetas)$ and $\Ee\left[\nabla\cL(\thetas;Z)\nabla\cL(\thetas;Z)^\top\right]$ are bounded away from zero and infinity respectively.
	
	\labitemc{(B4)}{as:subexp_glm} For some constant $L>0$,
	$$\max_l\max_{q=1,2}\Ee[|\bh_l^{2+q}|/L^q]+\Ee[\exp(|\bh_l|/L)]=O(1),\quad\text{or}$$
	$$\max_l\max_{q=1,2}\Ee[|\bh_l^{2+q}|/L^q]+\Ee[(\max_l|\bh_l|/L)^4]=O(1),$$
	where $\bh=\nabla^2\cLs(\thetas)^{-1}\nabla\cL(\thetas;Z)$ and $\bh_l$ is the $l$-th coordinate.

	\labitemc{(B5)}{as:glmsparse} $\thetas$ and $\Theta_{l,\cdot}$ are sparse, where the inverse population Hessian matrix $\Theta\defn\nabla^2\cLs(\thetas)^{-1}$, i.e., $S\defn\{l:\thetas_l\neq0\}$, $s_0\defn|S|$, $s_l\defn|\{l'\neq l:\Theta_{l,l'}\neq0\}|$, $s^*\defn\max_l s_l$, and $\overline s=s_0\vee s^*$.
\end{itemize}

Assumption \ref{as:smth_glm} imposes smoothness conditions on the loss function, which is satisfied by, for example, the logistic regression. In particular, logistic regression has $g(a,b)=-ab+\log(1+\exp(b))$, and it can be easily seen that $|g'(a,b)|\leq2$, $|g''(a,b)|\leq1$, $|g'''(a,b)|\leq1$. Assumption \ref{as:design_glm} imposes some boundedness conditions required for the validity of the nodewise Lasso (Algorithm~\ref{alg:node}; \citet{van2014asymptotically}) in the master node.  Assumption \ref{as:hes_glm} is a standard assumption in the GLM literature. Assumption \ref{as:subexp_glm} is required for proving the validity of multiplier bootstrap \citep{chernozhukov2013gaussian}.

Analogously to Theorem \ref{theo:reg0_csl} and \ref{theo:reg_csl} that focus on the distributed bootstrap validity and the efficiency of Algorithm~\ref{alg:hd} using \texttt{k-grad}/ \texttt{n+k-1-grad} for linear models, here we extend them to the high-dimensional de-biased GLMs. See Figure~\ref{fig:tau_hd} for a comparison between the results of high-dimensional linear models and GLMs.

\begin{theorem}[\texttt{k-grad}, sparse GLM]\label{theo:reg0_glm_csl}
	Suppose \ref{as:smth_glm}-\ref{as:glmsparse} hold, and that we run Algorithm \ref{alg:hd} with \texttt{k-grad} method in GLMs. Let $\lambda_l\asymp\sqrt{\log d/n}$ for $l=1,\dots,d$, and $\lambda^{(t)}$ be as
	\begin{align}
	\lambda^{(t)} \asymp \begin{cases}
    \sqrt{\frac{\log d}{nk}}+\frac1{s_0^2}\Big(s_0^2\sqrt{\frac{\log d}n}\Big)^{2^t}, & t\leq\tau_0,\\
    \sqrt{\frac{\log d}{nk}}+\frac1{s_0^2}\Big(s_0^2\sqrt{\frac{\log d}n}\Big)^{2^{\tau_0}}\Big(s_0\sqrt{\frac{\log d}n}\Big)^{t-\tau_0}, & t>\tau_0+1,
\end{cases} \label{eqn:lam_glm}
	\end{align}
	for $t=0,\dots,\tau-1$, where
	\begin{align}
	\tau_0=1+\left\lfloor\log_2\frac{\gamma_n-2\gamma_s}{\gamma_n-4\gamma_s}\right\rfloor. \label{eqn:tau0}
	\end{align}
	Assume $n=d^{\gamma_n}$, $k=d^{\gamma_k}$, $\overline s=d^{\gamma_s}$ for some constants $\gamma_n,\gamma_k,\gamma_s>0$. If $\gamma_n>5\gamma_s$, $\gamma_k>3\gamma_s$, and $\tau\geq\tau_{\min}$, where
	\begin{align*}
	\tau_{\min}=\max\left\{\tau_0+\left\lfloor\frac{\gamma_k+\gamma_s}{\gamma_n-2\gamma_s}+\nu_0\right\rfloor, 2+\left\lfloor\log_2\frac{\gamma_n-\gamma_s}{\gamma_n-4\gamma_s}\right\rfloor\right\},
	\end{align*}
	\begin{align}
	\nu_0=2-\frac{2^{\tau_0}(\gamma_n-4\gamma_s)}{\gamma_n-2\gamma_s}\in(0,1], \label{eqn:nu0}
	\end{align}
	then we have \eqref{eqn:kgradthm}.   In addition, \eqref{eqn:kgradthm} also holds if $T$ is replaced by $\widehat T$ defined in \eqref{eqn:that}.
\end{theorem}

The $\tau_0$ in \eqref{eqn:tau0} is the preliminary communication rounds needed for the CSL estimator to go through the regions which are far from $\thetas$. As $\overline s$ grows, the time spent in these regions can increase. However, when $n$ is large, e.g., $n\gg \overline s^6$, the loss function is more well-behaved so that the preliminary communication round can reduce to $\tau_0=1$. See Section \ref{rem:wang} in the Appendix for more details.

\begin{theorem}[\texttt{n+k-1-grad}, sparse GLM]\label{theo:reg_glm_csl}
	Suppose \ref{as:smth_glm}-\ref{as:glmsparse} hold, and that we run Algorithm \ref{alg:hd} with \texttt{n+k-1-grad} method in GLMs. Let $\lambda_l\asymp\sqrt{\log d/n}$ for $l=1,\dots,d$, and $\lambda^{(t)}$ be as in \eqref{eqn:lam_glm} for $t=0,\dots,\tau-1$.  Assume $n=d^{\gamma_n}$, $k=d^{\gamma_k}$, $\overline s=d^{\gamma_s}$ for some constants $\gamma_n,\gamma_k,\gamma_s>0$. If $\gamma_n>5\gamma_s$ and $\tau\geq\tau_{\min}$, where
	\begin{align*}
	\tau_{\min}=\begin{cases}
    \max\left\{2+\left\lfloor\log_2\frac{\gamma_k+\gamma_s}{\gamma_n-4\gamma_s}\right\rfloor,1\right\}, & \text{if}\quad \gamma_k\leq\gamma_n-3\gamma_s,\\
    \tau_0+\left\lfloor\frac{\gamma_k+\gamma_s}{\gamma_n-2\gamma_s}+\nu_0\right\rfloor, & \text{otherwise},
\end{cases}
	\end{align*}
	$\tau_0$ and $\nu_0$ defined as in \eqref{eqn:tau0} and \eqref{eqn:nu0} respectively, then we have \eqref{eqn:nk1gradthm}.   In addition, \eqref{eqn:nk1gradthm} also holds if $T$ is replaced by $\widehat T$ defined in \eqref{eqn:that}.
\end{theorem}

\begin{remark}
The selection of $\{\lambda_l\}_{l=1}^d$ in Theorems~\ref{theo:reg0_glm_csl}~and~\ref{theo:reg_glm_csl} are motivated by those in \citet{van2014asymptotically}, $\{\lambda^{(t)}\}_{t=0}^{\tau-1}$ are motivated by \citet{wang2017efficient} and \citet{jordan2019communication}. We perform a more careful analysis for the two phases of model tuning as in \eqref{eqn:lam_glm}.
\end{remark}

\section{Simulation Studies} \label{sec:exp}

We demonstrate the merits of our methods using synthetic data in this section. The code to reproduce the simulation experiments, results, and plots is available at GitHub: \url{https://github.com/skchao74/Distributed-bootstrap}.

We consider a Gaussian linear model and a logistic regression model. We fix total sample size $N=2^{14}$ and the dimension $d=2^{10}$, and choose the number of machines $k$ from $\{2^2,2^3,\dots,2^6\}$. The true coefficient $\thetas$ is a $d$-dimensional vector in which the first $s_0$ coordinates are 1 and the rest is 0, where $s_0\in\{2^2,2^4\}$ for the linear model and $s_0\in\{2^1,2^3\}$ for the GLM. We generate covariate vector $x$ independently from $\cN(0,\Sigma)$, while considering two different specifications for $\Sigma$: 
\begin{itemize}
	\item Toeplitz: $\Sigma_{l,l'}=0.9^{|l-l'|}$;
	\item Equi-correlation: $\Sigma_{l,l'}=0.8$ for all $l\neq l'$, $\Sigma_{l,l}=1$ for all $l$.
\end{itemize} 
For linear model, we generate the model noise independently from $\cN(0,1)$; for GLM, we obtain i.i.d. responses from $y\sim\text{Ber}(1/(1+\exp[-x^\top\thetas]))$. For each choice of $s_0$ and $k$, we run Algorithm~\ref{alg:hd} with \texttt{k-grad} and \texttt{n+k-1-grad} on $1{,}000$ independently generated datasets, and compute the empirical coverage probability and the average width based on the results from these $1{,}000$ replications. At each replication, we draw $B=500$ bootstrap samples, from which we calculate the $95\%$ empirical quantile to further obtain the $95\%$ simultaneous confidence interval.

For the $\ell_1$-CSL computation, we choose the initial $\lambda^{(0)}$ by a local $K$-fold cross-validation, where $K=10$ for linear regression and $K=5$ for logistic regression. For each iteration $t$, $\lambda^{(t)}$ is selected by 
Algorithm~\ref{alg:cv} in Section~\ref{sec:cv} with $K'$ folds with $K'=\min\{k-1,5\}$, which ensures that each partition of worker gradients is non-empty 
when $k$ is small. For an efficient implementation of the nodewise Lasso, we select a $\hat\lambda$ at every simulation repetition and set $\lambda_l=\bar\lambda$ for all $l$. Specifically, for each simulated dataset, we select $\bar\lambda = 10^{-1}\sum_{l=1}^{10} \hat\lambda_l$, where each $\hat\lambda_l$ is obtained obtained by a cross-validation of nodewise Lasso regression of $l$-th variable on the remaining variables. Since the variables are homogeneous, these $\hat\lambda_l$'s only deviate by some random variations, which can be alleviated by an average.

The computation of the oracle width starts with fixing $(N,d,s_0)$ and generating $500$ independent datasets. For each dataset, we compute the centralized de-biased Lasso estimator $\htheta$ as in \eqref{eq:dblasso}.  The oracle width is defined as two times the $95\%$ empirical quantile of $\|\htheta-\thetas\|_\infty$ of the 500 samples. The average widths are compared against the oracle widths by taking the ratio of the two. 

The empirical coverage probabilities and the average width ratios of \texttt{k-grad} and \texttt{n+k-1-grad} are displayed for the linear model in Figures \ref{fig:hd_lm_tp} (Toeplitz design) and \ref{fig:hd_lm_eq} (equi-correlation design), and for the logistic regression in Figures \ref{fig:hd_glm_tp} (Toeplitz design) and \ref{fig:hd_glm_eq} (equi-correlation design), respectively. Note that increase in $k$ indicates decrease in $n$, given the fixed $N$.

For small $k$, \texttt{k-grad} tends to over-cover, whereas \texttt{n+k-1-grad} has a more accurate coverage. By contrast, the coverage of both algorithms fall when $k$ gets too large (or $n$ gets too small), since the estimator $\ttheta^{(\tau)}$ deviates from $\htheta$ and the deviation of the width from the oracle width, which reflects the discussion of \eqref{eqn:kgrad_err_hd} and \eqref{eqn:nk1grad_err_hd}. Moreover, as $s_0=\|\thetas\|_0$ increases, it becomes harder for both algorithms to achieve the accurate $95\%$ coverage, and both algorithms start to fail at a smaller $k$ (or larger $n$), which stems from the fact that the bootstrap cannot accurately approximate variance of the asymptotic distribution as shown in \eqref{eqn:kgrad_err_hd} and \eqref{eqn:nk1grad_err_hd}. Nevertheless, raising the number of iterations improves the coverage, which verifies our theory.
We also observe an under-coverage of our bootstrap method in both the linear regression and the logistic regression at the early stage of increasing $k$. This is due to the loss of accuracy in estimating the inverse Hessian matrices using only the data in the master node when $k$ increases (or $n$ decreases).

\begin{figure}[ht!]
\vspace{-5mm}
\begin{center}
\centerline{\includegraphics[width=0.9\columnwidth]{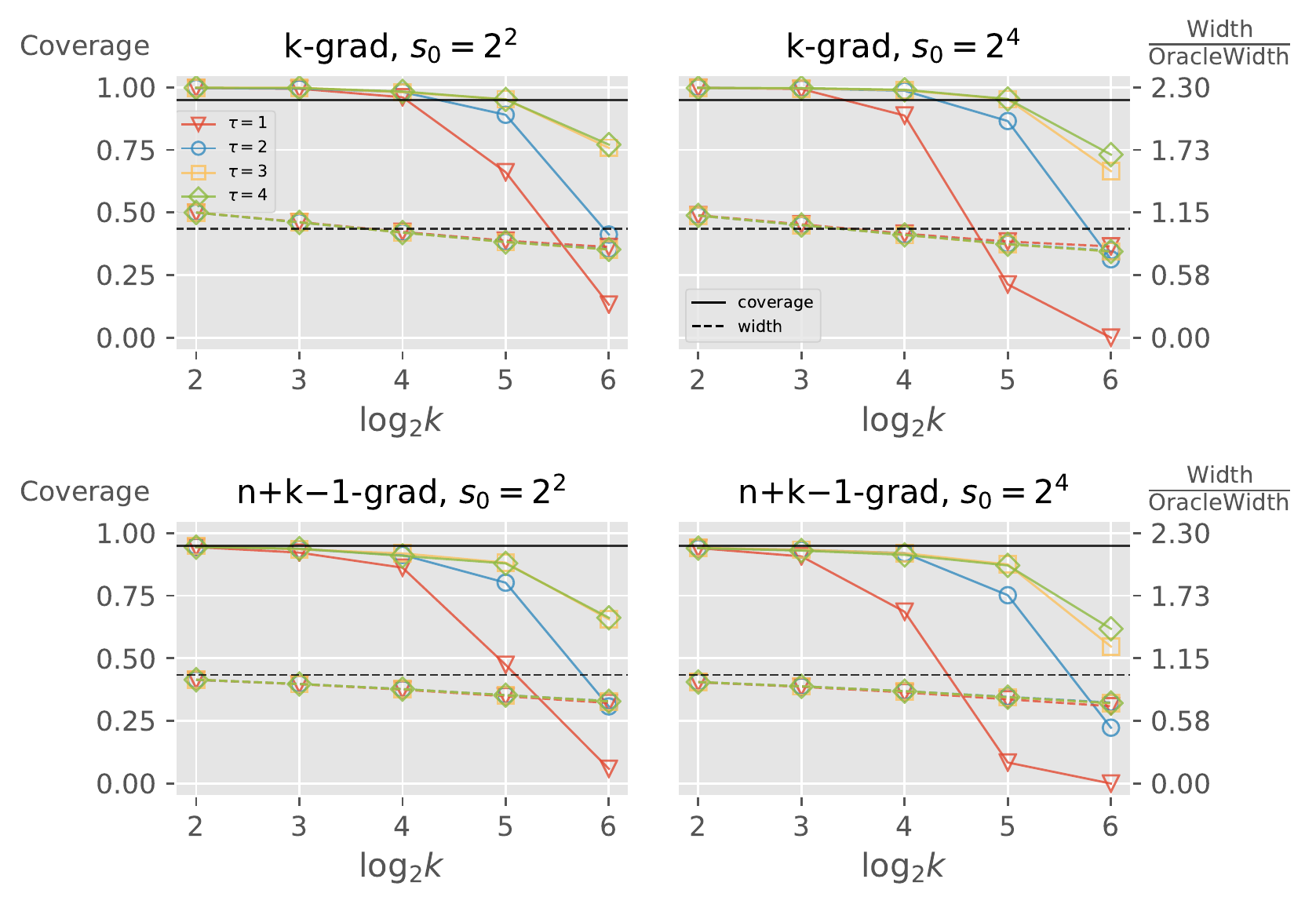}}
\caption{Empirical coverage probability (\textbf{left axis, solid lines}) and average width (\textbf{right axis, dashed lines}) of simultaneous confidence intervals by \texttt{k-grad} and \texttt{n+k-1-grad} in sparse linear regression with Toeplitz design and varying sparsity. Black solid line represents the $95\%$ nominal level and black dashed line represents 1 on the right $y$-axis.}
\label{fig:hd_lm_tp}
\end{center}
\vspace{-7mm}
\end{figure}

\begin{figure}[ht!]
\vspace{-5mm}
\begin{center}
\centerline{\includegraphics[width=0.9\columnwidth]{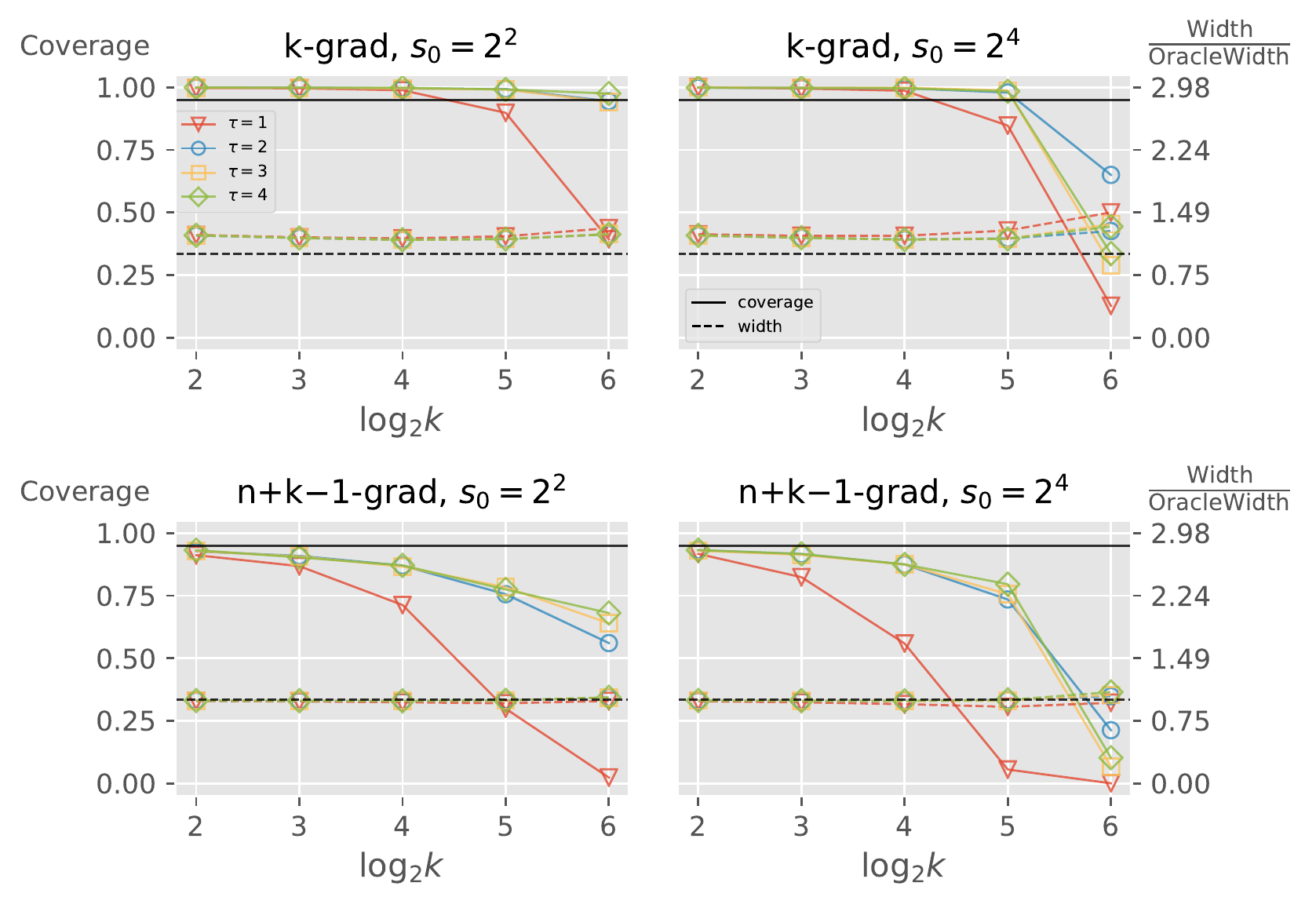}}
\caption{Empirical coverage probability (\textbf{left axis, solid lines}) and average width (\textbf{right axis, dashed lines}) of simultaneous confidence intervals by \texttt{k-grad} and \texttt{n+k-1-grad} in sparse linear regression with equi-correlation design and varying sparsity. Black solid line represents the $95\%$ nominal level and black dashed line represents 1 on the right $y$-axis.}
\label{fig:hd_lm_eq}
\end{center}
\vspace{-7mm}
\end{figure}

\begin{figure}[ht!]
\vspace{-5mm}
\begin{center}
\centerline{\includegraphics[width=0.9\columnwidth]{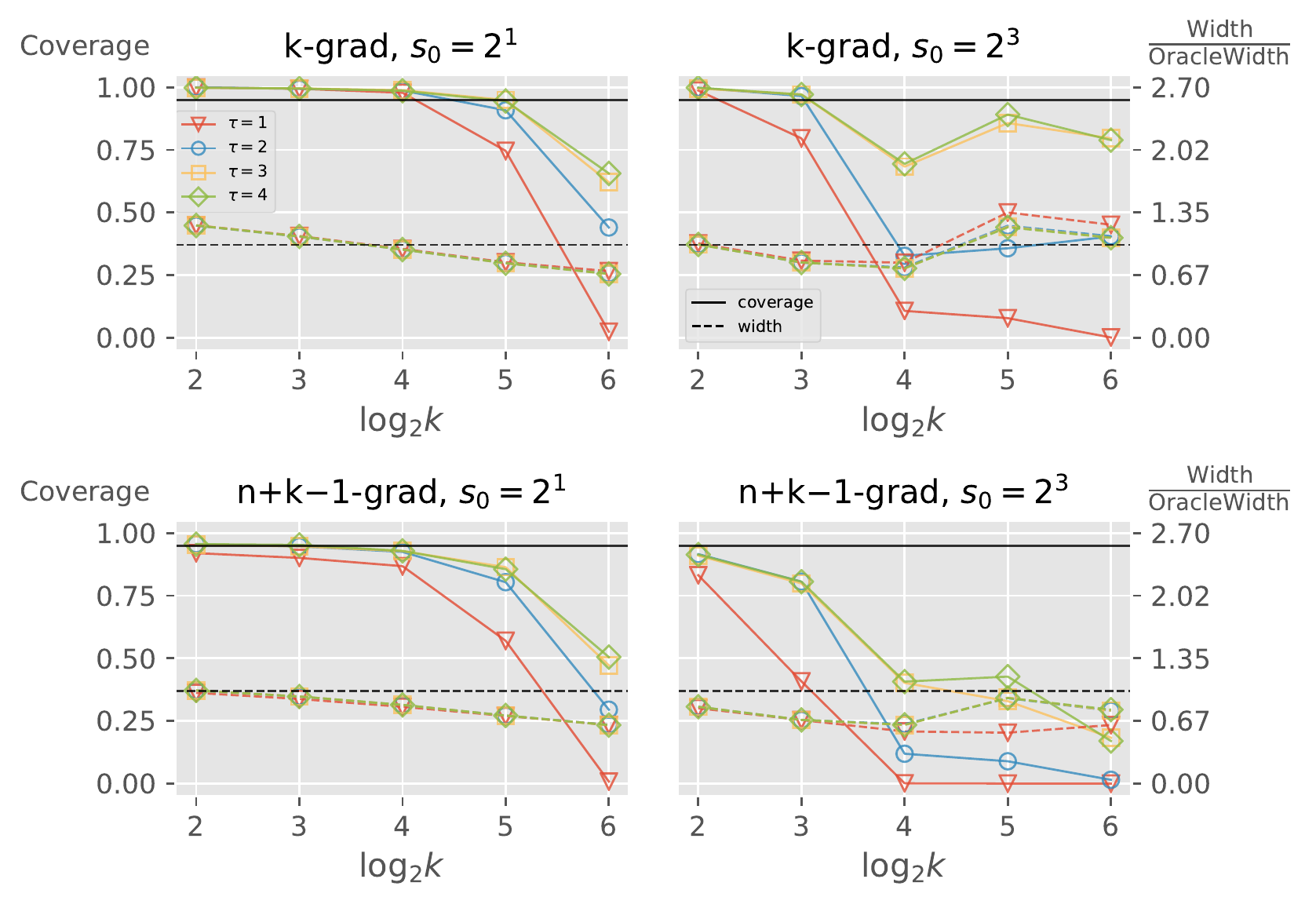}}
\caption{Empirical coverage probability (\textbf{left axis, solid lines}) and average width (\textbf{right axis, dashed lines}) of simultaneous confidence intervals by \texttt{k-grad} and \texttt{n+k-1-grad} in sparse logistic regression with Toeplitz design and varying sparsity.  Black solid line represents the $95\%$ nominal level and black dashed line represents 1 on the right $y$-axis.}
\label{fig:hd_glm_tp}
\end{center}
\vspace{-7mm}
\end{figure}

\begin{figure}[ht!]
\vspace{-5mm}
\begin{center}
\centerline{\includegraphics[width=0.9\columnwidth]{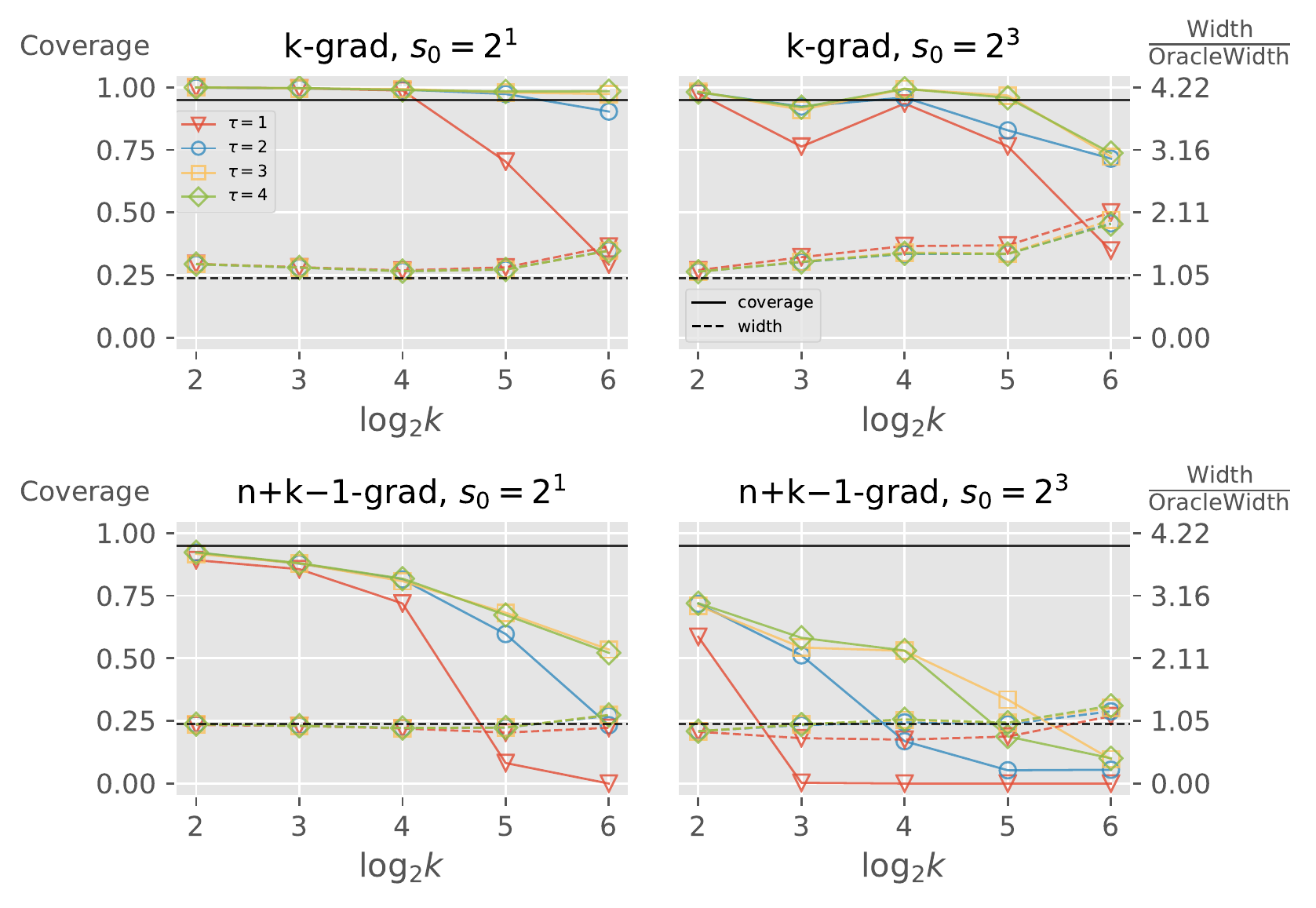}}
\caption{Empirical coverage probability (\textbf{left axis, solid lines}) and average width (\textbf{right axis, dashed lines}) of simultaneous confidence intervals by \texttt{k-grad} and \texttt{n+k-1-grad} in sparse logistic regression with equi-correlation design and varying sparsity.  Black solid line represents the $95\%$ nominal level and black dashed line represents 1 on the right $y$-axis.}
\label{fig:hd_glm_eq}
\end{center}
\vspace{-7mm}
\end{figure}

\section{Variable Screening with Distributed Simultaneous Inference} \label{sec:real}

Having demonstrated the performance of our method on purely synthetic data using sparse models in the last section, in this section, we artificially create spurious variables and mix them with the variables obtained from a real big dataset. We check if our method can successfully select the relevant variables associated with the response variable from the real dataset. The code to retrieve data and reproduce the analyses, results, and plots is available at  GitHub: \url{https://github.com/skchao74/Distributed-bootstrap}.

\subsection{Data}

The US Airline On-Time Performance dataset \citep{DVN/HG7NV7_2008}, available at
\url{http://stat-computing.org/dataexpo/2009}, consists of flight arrival and departure details for all commercial flights within the US from 1987 to 2008. Given the high dimensionality after dummy transformation and the huge sample size of the entire dataset, the most efficient way to process the data is using a distributed computational system, with sample size on each worker node likely to be smaller than the dimension. Our goal here is to uncover statistically significant independent variables associated with flight delay. We use variables \textsf{Year}, \textsf{Month}, \textsf{DayOfWeek}, \textsf{CRSDepTime}, \textsf{CRSArrTime}, \textsf{UniqueCarrier}, \textsf{Origin}, \textsf{Dest}, and \textsf{ArrDelay} in our model; descriptions are deferred to Appendix (Section~\ref{sec:var}).

The response variable is labeled by $1$ to denote a delay if \textsf{ArrDelay} is greater than zero, and by $0$ otherwise. The rest of the variables are treated as categorical explanatory variables and are converted into dummy variables; refer to Appendix (Section~\ref{sec:dummy}) for the details of the dummy variable creation. This results in a total of $203$ predictors. The total sample size is 113.9 million observations. We randomly sample a dataset $\mathcal D_1$ of $N=500{,}000$ observations, and conceptually distribute them across $k=1{,}000$ nodes such that each node receives $n=500$ observations. We randomly sample another dataset $\mathcal D_2$ of $N=500{,}000$ observations for a pilot study to select relevant variables, where $\mathcal D_1 \cap \mathcal D_2 = \emptyset$.

\subsection{An Artificial Design Matrix and Variable Screening}
In the first stage, we perform a preliminary study that informs us some seemingly relevant variables to include in an artificial design matrix, which will be used to demonstrate variable screening performance of our method in the second stage. Note that the purpose of this stage is only to preliminarily discover possibly relevant variables, rather than to select variables in a fully rigorous manner. We perform a logistic regression in a centralized manner with intercept and without regularization using the $N$ observations in $\mathcal D_2$. Standard Wald tests reveal that $144$ out of $203$ slopes are significantly non-zero ($p$-values less than $0.05$).

The four predictors with the least $p$-values correspond to the dummy variables of years 2001–2004, and the coefficients are all negative, which suggests less likelihood of flight delay in these years. This interesting finding matches the results of previous study that the September 11 terrorist attacks have negatively impacted the US airline demand \citep{ito2005assessing}, which led to less flights and congestion. In addition, the \textit{Notice of Market-based Actions to Relieve Airport Congestion and Delay}, (Docket No. OST-2001-9849) issued by Department of Transportation on August 21, 2001, might also alleviate the US airline delay.

To construct the artificial design matrix, we group the $4$ predictors with the least $p$-values mentioned above and the intercept, so the number of the relevant columns is $5$. Given $d$, we artificially create $d-5$ columns of binary and real valued variables 
by first sampling rows from $\cN(0, \cC_{d-5})$, where $\cC_{d-5}$ is a Toeplitz matrix ($(\cC_{d-5})_{l,l'}=0.5^{|l-l'|}$), and then converting half of the columns to either 0 or 1 by their signs. 
Then, we combine these $d-5$ spurious columns with a column of intercept and the $4$ columns in $\mathcal D_1$ that are associated with the selected relevant variables to obtain an artificial design matrix.

In the second stage, using the artificial design matrix with the binary response vector from the \textsf{ArrDelay} in $\mathcal D_1$, we test if our distributed bootstrap \texttt{n+k-1-grad} (Algorithm~\ref{alg:hd}) can screen the artificially created spurious variables. Note that $\mathcal D_1$ and $\mathcal D_2$ are disjoint, where $\mathcal D_2$ is used in the first stage for the preliminary study.
For model tuning, we select $\lambda^{(0)}$ by a local $10$-fold cross-validation; for each $t\ge1$, $\lambda^{(t)}$ is chosen by running a distributed $10$-fold cross-validation in Algorithm~\ref{alg:cv}. We select each $\lambda_l$ by performing a $10$-fold cross-validation for the nodewise Lasso of each variable. The same entire procedure is repeated under each dimensionality $d\in\{200,500,1{,}000\}$.

The left panel of Figure~\ref{fig:real} plots the number of significant variables against the number of iterations $\tau$, which was broken down into the number intersecting with the relevant variables (solid lines) and the number intersecting with the spurious variables (dashed lines). 
First, all of the $4$ relevant variables are tested to be significant at all iterations. 
For the spurious variables, we see that with $\tau=1$, the distributed bootstrap falsely detects one of them. However, as the number of iterations increases, less spurious variables are detected until none of them is detected. We also see that $2$ iterations ($\tau=2$) for $d=500,1{,}000$ and $3$ iterations ($\tau=3$) for $d=200$ are sufficient, which empirically verifies that our method is not very sensitive to the nominal dimension $d$.

As an illustration that is potentially useful in practice, the confidence intervals computed with the simultaneous quantile for the $4$ important slopes under $d=1{,}000$ and $\tau=2$ are plotted in the right panel of Figure~\ref{fig:real}. It can be seen that the flights in years 2002 and 2003 are relatively less likely to delay, which match the decreased air traffic in the aftermath of the September 11 terrorist attacks.

\begin{figure}[ht!]
\centering
    \begin{subfigure}{0.4\textwidth}
      \centering
      \includegraphics[height=0.23\textheight]{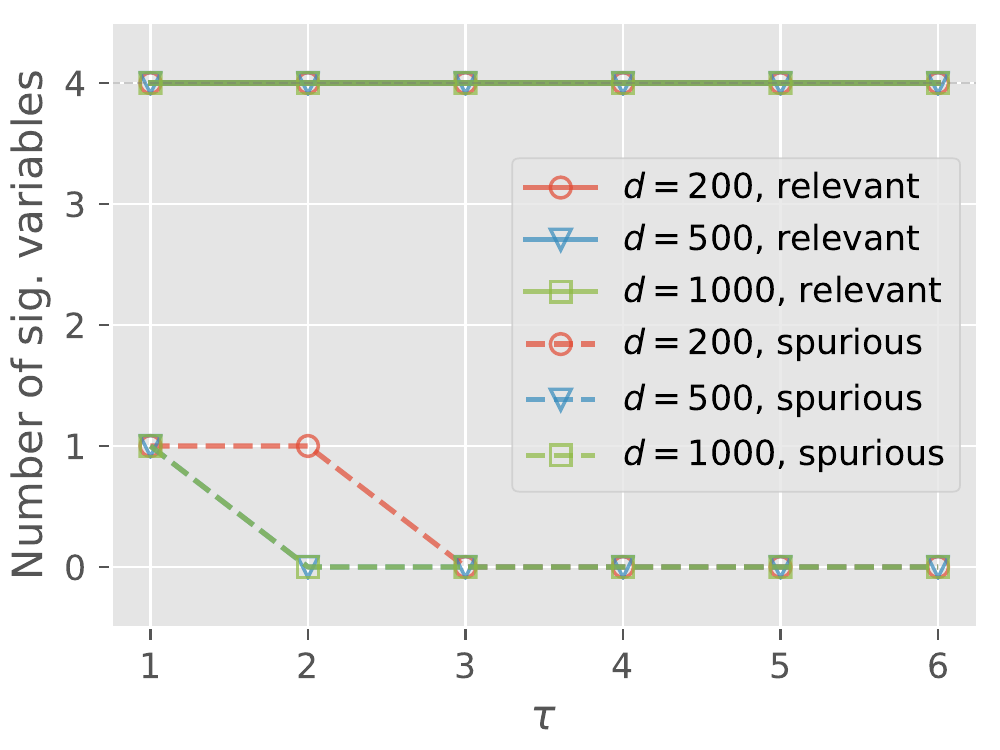}
    \end{subfigure}
    \hfill
    \begin{subfigure}{0.55\textwidth}
      \centering
      \includegraphics[height=0.23\textheight]{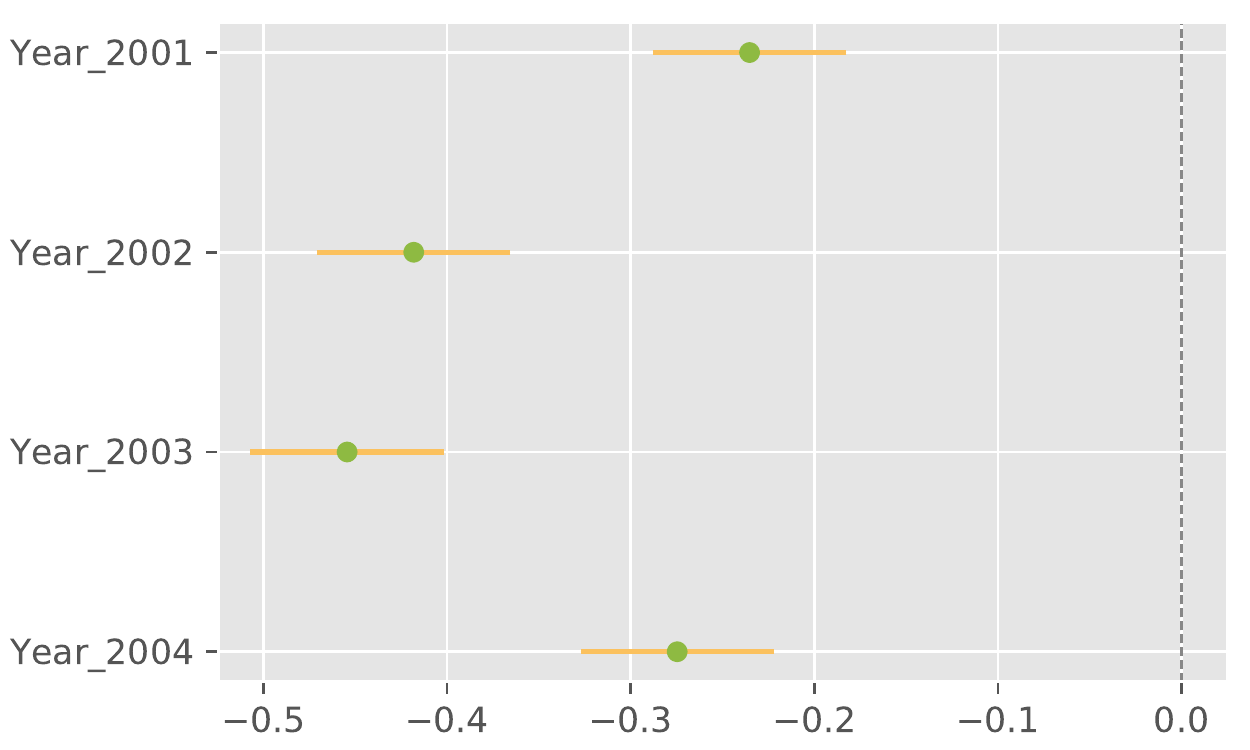}
    \end{subfigure}
\caption{The left panel shows the number of significant variables uncovered by the simultaneous confidence intervals among the $4$ relevant variables and among the $d-5$ spurious variables for $d=200,500,1{,}000$. The right panel shows the simultaneous confidence intervals of the $4$ relevant variables for $d=1{,}000$ and $\tau=2$.}
\label{fig:real}
\end{figure}

\section{Conclusion} \label{sec:disc}

We propose a distributed bootstrap method for high-dimensional simultaneous inference based on the de-biased $\ell_1$-CSL estimator as well as a distributed cross-validation method for hyperparameter tuning. The bootstrap validity and oracle efficiency are rigorously studied, and the merits are further shown via simulation study on coverage probability and efficiency, and a practical example on variable screening.

\acks{Shih-Kang Chao would like to acknowledge the financial support from the Research Council of the University of Missouri. Guang Cheng would like to acknowledge support 
from the National Science Foundation (NSF -- SCALE MoDL (2134209)). }

\newpage

\appendix

\section{Pseudocode for \texttt{k-grad} and \texttt{n+k-1-grad}}

\begin{algorithm}[ht!]
\caption{\texttt{DistBoots$(\text{method},\ttheta,\{\bg_j\}_{j=1}^k,\tT)$}: only need the master node $\cM_1$} \label{alg:kgrad}
\begin{algorithmic}[1]
\Statex {\bfseries Require:} local gradient $\bg_j$ and estimate $\tT$ of inverse Hessian obtained at $\cM_1$
\State $\bar\bg \gets k^{-1}\sum_{j=1}^k \bg_j$
\For{$b = 1,\ldots, B $}
\If{\texttt{method}=`\texttt{k-grad}'}
\State Draw $\epsilon_{1}^{(b)},\ldots,\epsilon_{k}^{(b)}\overset{\text{i.i.d.}}{\sim}\cN(0,1)$ and compute $W^{(b)}$ by \eqref{eqn:wb}
\ElsIf{\texttt{method}=`\texttt{n+k-1-grad}'}
\State Draw $\epsilon_{11}^{(b)},\ldots,\epsilon_{n1}^{(b)},\epsilon_{2}^{(b)},\ldots,\epsilon_{k}^{(b)}\overset{\text{i.i.d.}}{\sim}\cN(0,1)$ and compute $W^{(b)}$ by \eqref{eqn:wt}
\EndIf
\EndFor
\State Compute the quantile $c_{W}(\alpha)$ of $\{W^{(1)},\dots,W^{(B)}\}$ for $\alpha\in(0,1)$ 
\State Return $\ttheta_l\pm N^{-1/2}c_W(\alpha)$, $l=1,\dots,d$
\end{algorithmic}
\end{algorithm}

\begin{remark}
Although in Algorithm~\ref{alg:kgrad} the same $\ttheta$ is used for the center of the confidence interval and for evaluating the gradients $\bg_{ij}$, allowing them to be different (such as in Algorithm~\ref{alg:hd}) can save one round of communication. For example, we can use $\ttheta^{(\tau)}$ for the center of the confidence interval, while the gradients are evaluated with $\ttheta^{(\tau-1)}$.
\end{remark}

\section{Nodewise Lasso} \label{sec:node}

In Algorithm~\ref{alg:node}, we state the nodewise Lasso method for constructing approximate inverse Hessian matrix used in Section~3.1.1 of \citet{van2014asymptotically}, which we apply in Algorithm~\ref{alg:hd}.  We define the components of $\widehat\gamma_l$ as $\widehat\gamma_l=\{\widehat\gamma_{l,l'};l'=1,\dots,d,l'\neq l\}$.  We denote by $\widehat M_{l,-l}$ the $l$-th row of $\widehat M$ without the diagonal element $(l,l)$, and by $\widehat M_{-l,-l}$ the submatrix without the $l$-th row and $l$-th column.

\begin{algorithm}[th]
\caption{\texttt{Node}($\widehat M$)}\label{alg:node}
\begin{algorithmic}[1]
\Statex {\bfseries Require:} sample Hessian matrix $\widehat M\in\R^{d\times d}$, hyperparameters $\{\lambda_l\}_{l=1}^d$
\For{$l = 1,\ldots,d$}
\State Compute $\widehat\gamma_l=\argmin_{\gamma\in\R^{d-1}} \widehat M_{l,l}-2\widehat M_{l,-l}\gamma+\gamma^\top \widehat M_{-l,-l}\gamma+2\lambda_l\|\gamma\|_1$
\State Compute $\widehat\tau_l^2=\widehat M_{l,l}-\widehat M_{l,-l}\widehat\gamma_l$
\EndFor
\State Construct $\widehat{M^{-1}}$ as
$$\widehat{M^{-1}}=
\begin{pmatrix}
    \widehat\tau_1^{-2} & 0 & \dots  & 0 \\
    0 & \widehat\tau_2^{-2} & \dots  & 0 \\
    \vdots & \vdots & \ddots & \vdots \\
    0 & 0 & \dots  & \widehat\tau_d^{-2}
\end{pmatrix}
\begin{pmatrix}
    1& -\widehat\gamma_{1,2} & \dots  & -\widehat\gamma_{1,d} \\
    -\widehat\gamma_{2,1} & 1 & \dots  & -\widehat\gamma_{2,d} \\
    \vdots & \vdots & \ddots & \vdots \\
    -\widehat\gamma_{d,1} & -\widehat\gamma_{d,2} & \dots  & 1
\end{pmatrix}.
$$
\end{algorithmic}
\end{algorithm}

\begin{remark}
Throughout this paper, we fix the choice of nodewise Lasso in Algorithm~\ref{alg:hd} for computing an approximate inverse Hessian matrix. In practice, various approaches (e.g., \citet{zhang2014confidence,javanmard2014confidence}) can be chosen from in consideration of estimation accuracy and computational efficiency.
\end{remark}

\section{CSL Estimator for GLMs}\label{rem:wang}
For the $\ell_1$-penalized CSL estimator of generalized linear models, Theorem 3.3 of \citet{wang2017efficient} states that
	\begin{align}
	\big\|\ttheta^{(t+1)}-\thetas\big\|_1 \lesssim s_0\sqrt{\frac{\log d}N}+s_0\sqrt{\frac{\log d}n}\big\|\ttheta^{(t)}-\thetas\big\|_1+Ms_0\big\|\ttheta^{(t)}-\thetas\big\|_1^2, \label{eqn:csl}
	\end{align}
where $M\geq 0$ is a Lipschitz constant of the $g''$, which exists due to Assumptions \ref{as:smth_glm}. 
In linear models, $g(a,b)=(a-b)^2/2$, $g''$ is a constant, so $M=0$ and CSL estimator has linear convergence to $\thetas$ with rate $s_0(\log d)^{1/2}n^{-1/2}$ until it reaches the upper bound given by the first term, which is also the rate of the centralized (oracle) estimator. For GLMs, however, $M>0$ and the third term can be dominant when $t$ is small. For example, when $t=0$, given that $\|\ttheta^{(0)}-\thetas\|_1\lesssim s_0(\log d)^{1/2}n^{-1/2}$, it is easy to see that the third term is always $s_0$ times larger than the second term (up to a constant), and a larger $n$ is required to ensure third term is less than $\big\|\ttheta^{(t)}-\thetas\big\|_1$ and the error is shrinking. However, when $t$ is sufficiently large, this dominance reverses. The threshold is given by the $\tau_0$ in \eqref{eqn:tau0}, and this implies the three phases of convergence: When $t\leq\tau_0$, the third term dominates and the convergence is quadratic; when $t>\tau_0$, the second term dominates the third and the linear convergence kicks in. Finally, when $t$ is sufficiently large, the first term dominates. Our analysis complements that of \citet{wang2017efficient}, while in their Corollary 3.7 it is simply assumed that the second term dominates the third.

\section{Variable Descriptions} \label{sec:var}
We use the following variables in our model for the semi-synthetic study in Section~\ref{sec:real}:
\begin{itemize}
\itemsep-0.5em 
\item \textsf{Year}: from 1987 to 2008,
\item \textsf{Month}: from 1 to 12,
\item \textsf{DayOfWeek}: from 1 (Monday) to 7 (Sunday),
\item \textsf{CRSDepTime}: scheduled departure time (in four digits, first two representing hour, last two representing minute),
\item \textsf{CRSArrTime}: scheduled arrival time (in the same format as above),
\item \textsf{UniqueCarrier}: unique carrier code,
\item \textsf{Origin}: origin (in IATA airport code),
\item \textsf{Dest}: destination (in IATA airport code),
\item \textsf{ArrDelay}: arrival delay (in minutes). Positive value means there is a delay.
\end{itemize}
The complete variable information can be found at \url{http://stat-computing.org/dataexpo/2009/the-data.html}.

\section{Creation of Dummy Variables} \label{sec:dummy}

We categorize \textsf{CRSDepTime} and \textsf{CRSArrTime} into $24$ one-hour time intervals (e.g., 1420 is converted to 14 to represent the interval [14:00,15:00]), and then treat \textsf{Year}, \textsf{Month}, \textsf{DayOfWeek}, \textsf{CRSDepTime}, \textsf{CRSArrTime}, \textsf{UniqueCarrier}, \textsf{Origin}, and \textsf{Dest} as nominal predictors. The nominal
predictors are encoded by dummies with appropriate dimensions and merging all categories of lower counts into ``others'', and either ``others'' or the smallest ordinal value is treated as the baseline.

To ensure that none of the columns of the design matrix on the master node is completely zero so that the nodewise Lasso can be computed, we create the dummy variables using only the observations in the master node on the dataset $\mathcal D_1$. Specifically, for variables \textsf{UniqueCarrier}, \textsf{Origin}, and \textsf{Dest}, we keep the top categories that make up $90\%$ of the data in the master node on $\mathcal D_1$; the rest categories are merged into ``others'' and are treated as baseline. For \textsf{CRSDepTime} and \textsf{CRSArrTime}, we merge the time intervals 23:00-6:00 and 1:00-7:00 respectively (due to their low counts) and use them as baseline. For \textsf{Year}, \textsf{Month}, and \textsf{DayOfWeek}, we treat year 1987, January, and Monday as baseline respectively.

\section{Extension to Heteroscedastic Error Across Machines}

As suggested by the associated editor, here we consider an extension to a more challenging scenario for linear models where the data across machines have heteroscedastic errors. In this scenario, Algorithm \ref{alg:cv} can no longer apply as it relies on the homogeneity in data across machines. We provide a new Algorithm \ref{alg:hdm} by exploiting the multiplier bootstrap idea underlying the ``High-Dimensional Metrics'' (HDM, \cite{hdm16}).

\begin{algorithm}[ht!]
\caption{Simultaneous inference for distributed data with heteroscedasticity}\label{alg:hdm}
\begin{algorithmic}[1]
\Statex {\bfseries Require:} $\tau\geq 1$ rounds of communication; nodewise Lasso procedure \texttt{Node}$(\cdot,\cdot)$ with hyperparameters $\{\lambda_l\}_{l=1}^d$, theoretical constant $c$
\State $\ttheta^{(0)}\gets\argmin_\theta \cL_1(\theta)+\lambda^{(0)}\|\theta\|_1$ at $\cM_1$, where $\lambda^{(0)}$ is chosen by cross-validation using the data at $\cM_1$
\State Compute $\tT$ by running \texttt{Node}$(\nabla^2\cL_1(\ttheta^{(0)}),\{\lambda_l\}_{l=1}^d)$ at $\cM_1$
\For{$t = 1,\ldots, \tau $} 
\State Transmit $\ttheta^{(t-1)}$ to $\{\cM_j\}_{j=2}^k$
\State Compute $\nabla\cL_1(\ttheta^{(t-1)})$ and $\psi_1^{(t-1)}=n^{-1}\sum_{i=1}^n\nabla\cL((x_{i1}, y_{i1}),\ttheta^{(t-1)})^2$ at $\cM_1$
\For{$j = 2,\ldots, k $}
\State Compute $\nabla\cL_j(\ttheta^{(t-1)})$ and $\psi_j^{(t-1)}=n^{-1}\sum_{i=1}^n\nabla\cL((x_{ij}, y_{ij}),\ttheta^{(t-1)})^2$ at $\cM_j$
\State Transmit $\nabla\cL_j(\ttheta^{(t-1)})$ and $\psi_j^{(t-1)}$ to $\cM_1$
\EndFor
\State $\nabla\cL_N(\ttheta^{(t-1)})\gets k^{-1}\sum_{j=1}^k\nabla\cL_j(\ttheta^{(t-1)})$ at $\cM_1$
\If{$t < \tau$}
\For{$b = 1,\ldots, B $} \label{line:lambda_start}
\State Draw $\epsilon_{1}^{(b)},\ldots,\epsilon_{k}^{(b)}\overset{\text{i.i.d.}}{\sim}\cN(0,1)$
\State $\Lambda_b^{(t)} \gets c k^{-1} \|\sum_{j=1}^k \epsilon_j^{(b)}\nabla\cL_j(\ttheta^{(t-1)}) \|_{\infty}$
\EndFor
\State $\lambda^{(t)}\gets 90\%$ quantile of $\{\Lambda_1^{(t)},\dots,\Lambda_B^{(t)}\}$ \label{line:lambda_end}
\For{$l = 1,\ldots, d$} \label{line:psi_start}
\State $\Psi_l^{(t)}\gets\sqrt{k^{-1}\sum_{j=1}^k(\psi_j^{(t-1)})_l}$
\EndFor
\State $\Psi^{(t)}\gets \text{diag}(\Psi_1^{(t)},\dots,\Psi_d^{(t)})$
\State $\ttheta^{(t)}\gets\argmin_\theta \cL_1(\theta)-\theta^\top\left(\nabla\cL_1(\ttheta^{(t-1)})-\nabla\cL_N(\ttheta^{(t-1)})\right)+\lambda^{(t)}\|\Psi^{(t)}\theta\|_1$ at $\cM_1$ \label{line:psi_end}
\Else
\State $\ttheta^{(\tau)}\gets\ttheta^{(\tau-1)}-\tT\nabla\cL_N(\ttheta^{(\tau-1)})$ at $\cM_1$
\EndIf
\EndFor
\State Run \texttt{DistBoots}$(\text{`\texttt{k-grad}' or `\texttt{n+k-1-grad}'},\ttheta=\ttheta^{(\tau)},\{\bg_j=\nabla\cL_j(\ttheta^{(\tau-1)})\}_{j=1}^k,$ \\
\hspace{90pt}$\tT=\tT)$ at $\cM_1$
\end{algorithmic}
\end{algorithm}

 In Algorithm~\ref{alg:hdm}, we select the regularization parameters $\{\lambda^{(t)}\}_{t=1}^{\tau-1}$ in lines~\ref{line:lambda_start}-\ref{line:lambda_end} by integrating the idea of \citet{spindler2016hdm}. In addition, we handle heteroscedasticity by data-driven regularization loadings $\Psi^{(t)}$ in lines~\ref{line:psi_start}-\ref{line:psi_end}.

Under heteroscedasticity, we expect the \texttt{k-grad} in Algorithm \ref{alg:kgrad} to continue being valid because it treats each machine equally as an independent data point. However, \texttt{n+k-1-grad} may no longer provide an accurate coverage because each single data point in the first machine is treated as equally important as the average of entire data in $j$ machine for $j=2,\cdots,k$, so the variance in the first machine could dominate so that the \texttt{n+k-1-grad} bootstrap could fail to precisely approximate the variance of the target empirical distribution. A careful theoretical study deserves future research.

The empirical performance of Algorithm~\ref{alg:hdm} is verified by a simulation study based on a heteroscedastic Gaussian linear model. We fix total sample size $N=2^{14}$ and the dimension $d=2^{10}$, and choose the number of machines $k$ from $\{2^2,2^3,\dots,2^6\}$. The true coefficient $\thetas$ is a $d$-dimensional vector in which the first $s_0$ coordinates are 1 and the rest is 0, where $s_0\in\{2^2,2^4\}$. We generate covariate vector $x$ independently from $\cN(0,\Sigma)$, where $\Sigma$ is a Toeplitz matrix with $\Sigma_{l,l'}=0.9^{|l-l'|}$. We introduce heteroscedasticity across machines by first independently generating the model noise from $\mathcal N(0,1)$ for all data in the master node $\mathcal M_1$. Next, we generate model noise for each data point $i$ in worker node $\mathcal M_j$ ($j=2,\dots,k$) independently from $\mathcal N(0,\sigma_{j}^2+\omega_{ij})$, where the node level variance $\sigma_j^2$ is generated independently from $\text{Unif}(2,3)$ and the idiosyncratic variance $\omega_{ij}$ is generated independently from $\text{Unif}(-0.2,0.2)$. For each choice of $s_0$ and $k$, we run Algorithm~\ref{alg:hdm} with \texttt{k-grad} and \texttt{n+k-1-grad} on $1{,}000$ independently generated datasets, and compute the empirical coverage probability and the average width based on the results from these $1{,}000$ replications. At each replication, we draw $B=500$ bootstrap samples, from which we calculate the $95\%$ empirical quantile to further obtain the $95\%$ simultaneous confidence interval. For tuning the nodewise Lasso, we use the same approach as in the main text. The computation of the oracle width starts with fixing $(N,d,s_0,k)$ and generating $500$ independent datasets. For each dataset, we compute the centralized de-biased Lasso estimator $\htheta$. The oracle width is defined as two times the $95\%$ empirical quantile of $\|\htheta-\thetas\|_\infty$ of the 500 samples. 

Figure \ref{fig:hdm_c05} shows the coverage probability and efficiency in the form of relative widths of Algorithm~\ref{alg:hdm}. As expected, the coverage of the simultaneous confidence intervals is improved as the iteration goes using the new data-driven parameter tuning and heteroscedasticity-adapted regularization. The \texttt{k-grad} performs much better than the \texttt{n+k-1-grad}, which basically fails as the coverage probability of \texttt{n+k-1-grad} is nearly zero in all cases. The failure of the \texttt{n+k-1-grad} is due to the fact that it over-weigh the data in the master node $\cM_1$ which leads to an under-estimation of the variance in other nodes, whereas in \texttt{k-grad} each node is weighed equally.

By comparing Figure \ref{fig:hdm_c05} and Figure \ref{fig:hdm_c1}, we observe that our algorithm is generally robust to the selection of $c$ as it performs similarly for $c=0.5$ and $c=1$. However, we note that $c=0.5$ could be too small to stabilize the algorithm as the optimization solver in \ref{line:psi_end} fails to converge in about $2\%$ of the replications, and the divergent results are not included in Figure \ref{fig:hdm_c05}. This suggests that the penalty at $c=0.5$ may be so small that leads to an ill-conditioned objective function. After increasing $c$ from $0.5$ to $1$, optimization solvers of all the replications stably converge. 

\begin{figure}[ht]
\centering
\includegraphics[width=0.8\textwidth]{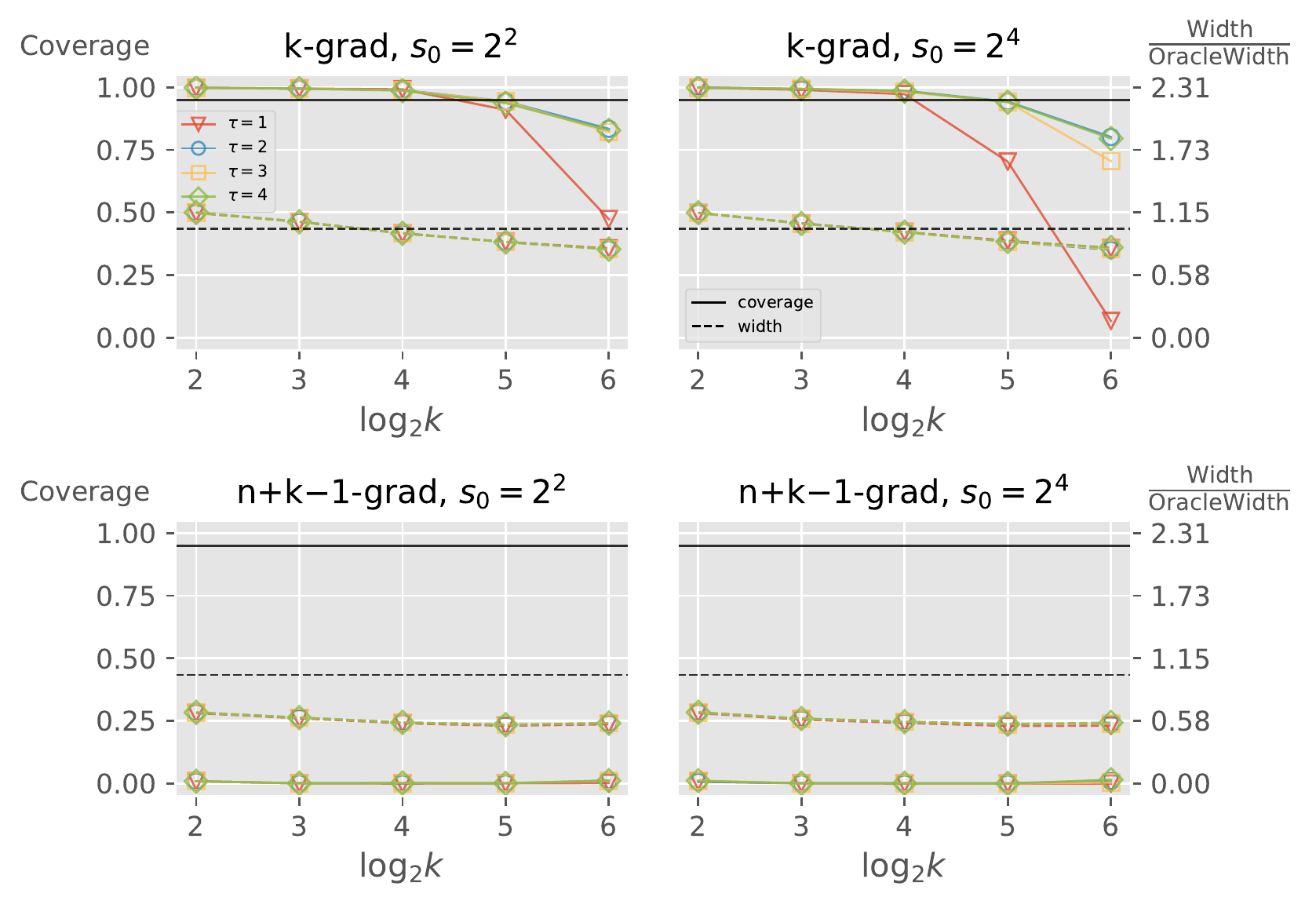}
\caption{Under $c=0.5$, empirical coverage probability (\textbf{left axis, solid lines}) and average relative width (\textbf{right axis, dashed lines}) of simultaneous confidence intervals by \texttt{k-grad} and \texttt{n+k-1-grad} in sparse linear regression with Toeplitz design and varying sparsity. Black solid line represents the $95\%$ nominal level and black dashed line represents 1 on the right $y$-axis.}
\label{fig:hdm_c05}
\end{figure}

\begin{figure}[ht]
\centering
\includegraphics[width=0.8\textwidth]{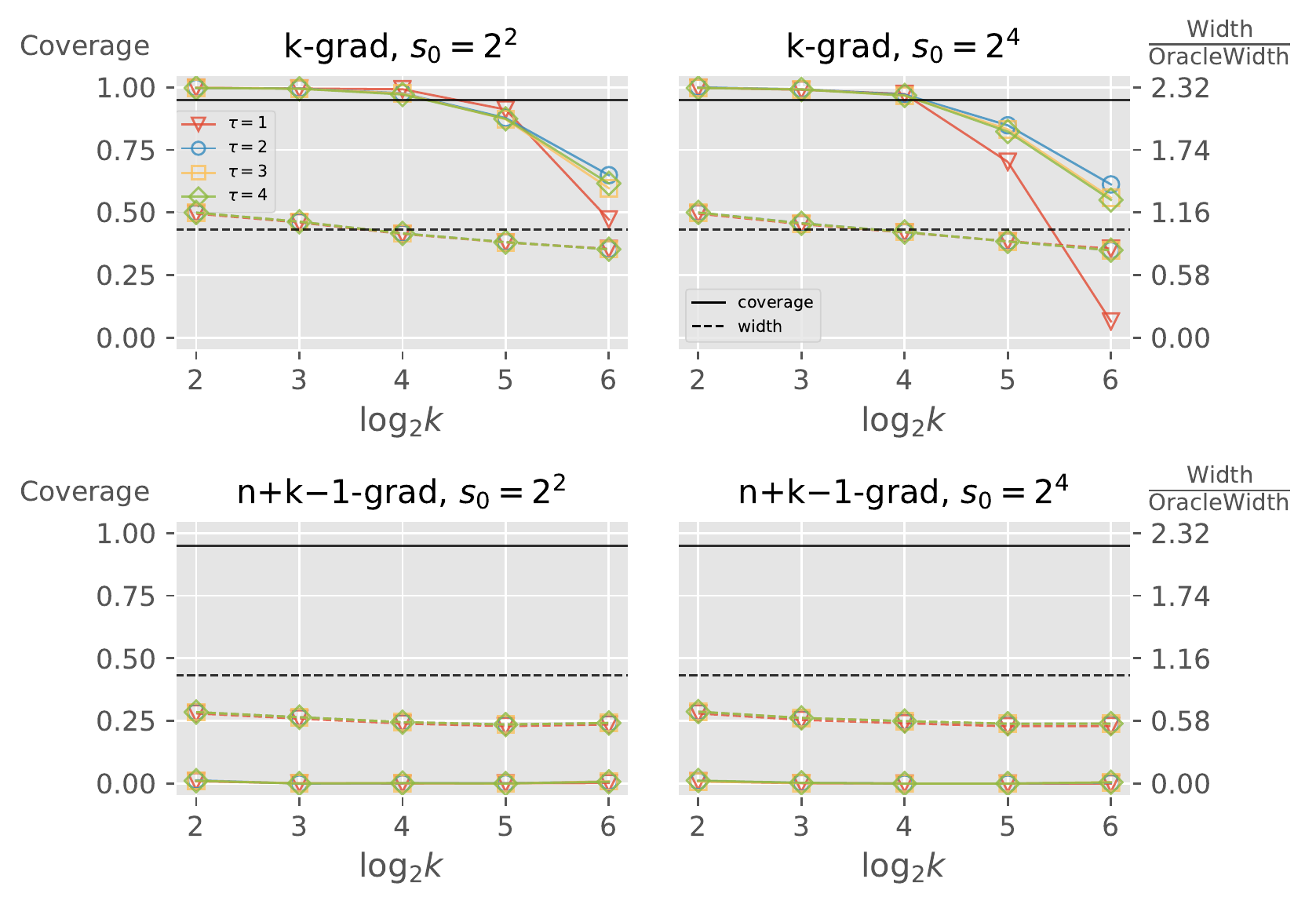}
\caption{Under $c=1$, empirical coverage probability (\textbf{left axis, solid lines}) and average relative width (\textbf{right axis, dashed lines}) of simultaneous confidence intervals by \texttt{k-grad} and \texttt{n+k-1-grad} in sparse linear regression with Toeplitz design and varying sparsity. Black solid line represents the $95\%$ nominal level and black dashed line represents 1 on the right $y$-axis.}
\label{fig:hdm_c1}
\end{figure}

\vskip 0.2in
\bibliography{Reference}

\newpage
\begin{center}
{\large\bf SUPPLEMENTARY MATERIAL}
\end{center}

\section*{1. Proofs of Main Results} \label{sec:mp}

To simplify the notation, in the proof we denote $\btheta=\ttheta^{(\tau-1)}$, where $\ttheta^{(\tau-1)}$ is the $\ell_1$-penalized estimator at $\tau-1$ iterator output by Algorithm~\ref{alg:hd}. Denote $\ttheta=\ttheta^{(\tau)}$ output by Algorithm~\ref{alg:hd}.

\begin{proof}[Theorem \ref{theo:reg0_csl}]
	We apply Theorem~3 of \citet{wang2017efficient}, where their Assumption~2 is inherited from Assumption~\ref{as:design}, and obtain that if $n\gg s_0^2\log d$,
	$$\left\|\btheta-\thetas\right\|_1 = \left\|\ttheta^{(\tau-1)}-\thetas\right\|_1 = O_P\left(s_0\sqrt{\frac{\log d}N}+\left(s_0\sqrt{\frac{\log d}n}\right)^\tau\right).$$
	Then, by Lemma~\ref{lem:reg0}, we have that $\sup_{\alpha\in(0,1)}\left|P(T\leq c_{\overline W}(\alpha))-\alpha\right|=o(1),$ as long as $n\gg{s^*}^2\log^{3+\kappa} d + {s^*}\log^{5+\kappa} d+s_0^2\log d$, $k\gg{s^*}^2\log^{5+\kappa} d$, and
	$$s_0\sqrt{\frac{\log d}N}+\left(s_0\sqrt{\frac{\log d}n}\right)^\tau\ll\min\left\{\frac1{ \sqrt{k{s^*}}\log^{1+\kappa}d},\frac1{ \sqrt{n {s^*}}\log^{1+\kappa}d}\right\}.$$
	
	These conditions hold if $n\gg  ({s^*}^2+{s^*} s_0^2)\log^{3+\kappa} d + {s^*}\log^{5+\kappa} d$, $k\gg  {s^*} s_0^2\log^{3+\kappa} d + {s^*}^2\log^{5+\kappa} d$, and
	$$\tau>\max\left\{\frac{\log k+\log{  {s^*}}+\log(C\log^{2+\kappa}d)}{\log n-\log(s_0^2)-\log\log d},1+\frac{\log{  {s^*}}+\log(s_0^2)+\log(C\log^{3+\kappa}d)}{\log n-\log(s_0^2)-\log\log d}\right\}.$$
	
	If $n=d^{\gamma_n}$, $k=d^{\gamma_k}$, $\overline s = s_0\vee s^*=d^{\gamma_s}$ for some constants $\gamma_n$, $\gamma_k$, and $\gamma_s$, then a sufficient condition is $\gamma_n>3\gamma_s$, $\gamma_k>3\gamma_s$, and
	$$\tau\geq1+\left\lfloor\max\left\{\frac{\gamma_k+\gamma_s}{\gamma_n-2\gamma_s},1+\frac{3\gamma_s}{\gamma_n-2\gamma_s}\right\}\right\rfloor.$$
\end{proof}

\begin{proof}[Theorem \ref{theo:reg_csl}]
	Similarly to the proof of Theorem~\ref{theo:reg0_csl}, applying Theorem~3 of \citet{wang2017efficient} and Lemma~\ref{lem:reg}, we have that $\sup_{\alpha\in(0,1)}\left|P(T\leq c_{\widetilde W}(\alpha))-\alpha\right|=o(1),$ as long as $n\gg{s^*}^2\log^{3+\kappa} d + {s^*}\log^{5+\kappa} d+s_0^2\log d$, $n+k\gg{s^*}^2\log^{5+\kappa} d$, and
	$$s_0\sqrt{\frac{\log d}N}+\left(s_0\sqrt{\frac{\log d}n}\right)^\tau\ll\min\left\{\frac1{ \sqrt{k{s^*}}\log^{1+\kappa}d},\frac1{  {s^*}\sqrt{\log((n+k)d)}\log^{2+\kappa}d}\right\}.$$
	These conditions hold if $n\gg  ({s^*}^2+{s^*} s_0^2)\log^{3+\kappa} d + {s^*}\log^{5+\kappa} d$, $n+k\gg{s^*}^2\log^{5+\kappa} d$, $nk\gg{s^*}^2s_0^2\log^{5+\kappa} d$, and
	$$\tau>\max\left\{\frac{\log k+\log{  {s^*}}+\log(C\log^{2+\kappa}d)}{\log n-\log(s_0^2)-\log\log d},\frac{\log{{s^*}^2}+\log\log((n+k)d)+\log(C\log^{4+\kappa}d)}{\log n-\log(s_0^2)-\log\log d}\right\}.$$
	
	If $n=d^{\gamma_n}$, $k=d^{\gamma_k}$, $\overline s = s_0\vee s^*=d^{\gamma_s}$ for some constants $\gamma_n$, $\gamma_k$, and $\gamma_s$, then a sufficient condition is $\gamma_n>3\gamma_s$, $\gamma_n+\gamma_k>4\gamma_s$, and
	$$\tau\geq1+\left\lfloor\frac{(\gamma_k\vee\gamma_s)+\gamma_s}{\gamma_n-2\gamma_s}\right\rfloor.$$
\end{proof}
	
\begin{proof}[Theorem \ref{theo:reg0_glm_csl}]
	We apply Theorem~6 of \citet{wang2017efficient}, where their Assumption~2 is inherited from Assumption~\ref{as:hes_glm}, and obtain that if $n\gg s_0^4\log d$,
	\begin{align*}
	\left\|\btheta-\thetas\right\|_1 &= \left\|\ttheta^{(\tau-1)}-\thetas\right\|_1 = \begin{cases}
    O_P\left(s_0\sqrt{\frac{\log d}N}+\frac1{s_0}\left(s_0^2\sqrt{\frac{\log d}n}\right)^{2^{\tau-1}}\right), & \tau\leq\tau_0+1,\\
    O_P\left(s_0\sqrt{\frac{\log d}N}+\frac1{s_0}\left(s_0^2\sqrt{\frac{\log d}n}\right)^{2^{\tau_0}}\left(s_0\sqrt{\frac{\log d}n}\right)^{\tau-\tau_0-1}\right), & \tau>\tau_0+1,
\end{cases}
	\end{align*}
	where $\tau_0$ is the smallest integer $t$ such that
	$$\left(s_0^2\sqrt{\frac{\log d}n}\right)^{2^t}\lesssim s_0\sqrt{\frac{\log d}n},$$
	that is,
	$$\tau_0=\left\lceil\log_2\left(\frac{\log n-\log(s_0^2)-\log(C\log d)}{\log n-\log(s_0^4)-\log\log d}\right)\right\rceil.$$
	Then, by Lemma~\ref{lem:reg0_glm}, we have that $\sup_{\alpha\in(0,1)}\left|P(T\leq c_{\overline W}(\alpha))-\alpha\right|=o(1),$ as long as $n\gg  (s_0^2+{s^*}^2)\log^{3+\kappa} d+(s_0+{s^*})\log^{5+\kappa} d+s_0^4\log d$, $k\gg{s^*}^2\log^{5+\kappa} d$, and
	\begin{align*}
	s_0\sqrt{\frac{\log d}N}+\frac1{s_0}\left(s_0^2\sqrt{\frac{\log d}n}\right)^{2^{\tau-1}}\ll\min\left\{\frac1{ \sqrt{k{s^*}}s_0\log^{1+\kappa}d},\frac1{ \sqrt{n {s^*}}\log^{1+\kappa}d}\right\},
	\end{align*}
	if $\tau\leq\tau_0+1$, and
	\begin{align*}
	s_0\sqrt{\frac{\log d}N}+\frac1{s_0}\left(s_0^2\sqrt{\frac{\log d}n}\right)^{2^{\tau_0}}\left(s_0\sqrt{\frac{\log d}n}\right)^{\tau-\tau_0-1}\ll\min\left\{\frac1{ \sqrt{k{s^*}}s_0\log^{1+\kappa}d},\frac1{ \sqrt{n {s^*}}\log^{1+\kappa}d}\right\},
	\end{align*}
	if $\tau>\tau_0+1$.
		
	If $n=d^{\gamma_n}$, $k=d^{\gamma_k}$, $\overline s = s_0\vee s^*=d^{\gamma_s}$ for some constants $\gamma_n$, $\gamma_k$, and $\gamma_s$, then a sufficient condition is $\gamma_n>5\gamma_s$, $\gamma_k>3\gamma_s$, and
	\begin{align*}
	\tau&\geq1+\left\lfloor\max\left\{1+\log_2\frac{\gamma_n-\gamma_s}{\gamma_n-4\gamma_s}, \tau_0+1+\frac{\gamma_k+(4\cdot 2^{\tau_0}+1)\gamma_s-2^{\tau_0}\gamma_n}{\gamma_n-2\gamma_s}\right\}\right\rfloor \\
	&=\left\lfloor\max\left\{2+\log_2\frac{\gamma_n-\gamma_s}{\gamma_n-4\gamma_s}, \tau_0+2+\frac{\gamma_k+(4\cdot 2^{\tau_0}+1)\gamma_s-2^{\tau_0}\gamma_n}{\gamma_n-2}\right\}\right\rfloor \\
	&=\left\lfloor\max\left\{2+\log_2\frac{\gamma_n-\gamma_s}{\gamma_n-4\gamma_s}, \tau_0+\frac{\gamma_k+\gamma_s}{\gamma_n-2\gamma_s}+\nu_0\right\}\right\rfloor \\
	&=\max\left\{\tau_0+\left\lfloor\frac{\gamma_k+\gamma_s}{\gamma_n-2\gamma_s}+\nu_0\right\rfloor, 2+\left\lfloor\log_2\frac{\gamma_n-\gamma_s}{\gamma_n-4\gamma_s}\right\rfloor\right\},
	\end{align*}
	where
	$$\tau_0=1+\left\lfloor\log_2\frac{\gamma_n-2\gamma_s}{\gamma_n-4\gamma_s}\right\rfloor,\quad\nu_0=2-\frac{2^{\tau_0}(\gamma_n-4\gamma_s)}{\gamma_n-2\gamma_s}\in(0,1].$$
\end{proof}

\begin{proof}[Theorem \ref{theo:reg_glm_csl}]
	Similarly to the proof of Theorem~\ref{theo:reg_csl}, applying Theorem~3 of \citet{wang2017efficient} and Lemma~\ref{lem:reg_glm}, we have that $\sup_{\alpha\in(0,1)}\left|P(T\leq c_{\overline W}(\alpha))-\alpha\right|=o(1),$ as long as $n\gg  (s_0+{s^*})\log^{5+\kappa} d+(s_0^2+{s^*}^2)\log^{3+\kappa} d$, $n+k\gg{s^*}^2\log^{5+\kappa} d$, and
	\begin{align*}
	&s_0\sqrt{\frac{\log d}N}+\frac1{s_0}\left(s_0^2\sqrt{\frac{\log d}n}\right)^{2^{\tau-1}} \\
	&\ll\min\left\{\frac{n+k}{  {s^*}\left(n+k\sqrt{\log d}+k^{3/4}\log^{3/4}d\right)\log^{2+\kappa}d},\frac1{ \sqrt{k{s^*}}s_0\log^{1+\kappa}d}, \frac1{\left(nk{s^*}\log^{1+\kappa}d\right)^{1/4}}\right\},
	\end{align*}
	if $\tau\leq\tau_0+1$, and
	\begin{align*}
	&s_0\sqrt{\frac{\log d}N}+\frac1{s_0}\left(s_0^2\sqrt{\frac{\log d}n}\right)^{2^{\tau_0}}\left(s_0\sqrt{\frac{\log d}n}\right)^{\tau-\tau_0-1} \\
	&\ll\min\left\{\frac{n+k}{  {s^*}\left(n+k\sqrt{\log d}+k^{3/4}\log^{3/4}d\right)\log^{2+\kappa}d},\frac1{ \sqrt{k{s^*}}s_0\log^{1+\kappa}d}, \frac1{\left(nk{s^*}\log^{1+\kappa}d\right)^{1/4}}\right\},
	\end{align*}
	if $\tau>\tau_0+1$, where
	$$\tau_0=\left\lceil\log_2\left(\frac{\log n-\log(s_0^2)-\log(C\log d)}{\log n-\log(s_0^4)-\log\log d}\right)\right\rceil.$$
	
	If $n=d^{\gamma_n}$, $k=d^{\gamma_k}$, $\overline s = s_0\vee s^*=d^{\gamma_s}$ for some constants $\gamma_n$, $\gamma_k$, and $\gamma_s$, then a sufficient condition is $\gamma_n>5\gamma_s$, and
	
	Let $\overline s=s_0\vee s^*$. If $n=\overline s^{\gamma_n}$, $k=\overline s^{\gamma_k}$, and $d=\overline s^{\gamma_d}$ for some constants $\gamma_n$, $\gamma_k$, and $\gamma_d$, then a sufficient condition is $\gamma_n>5$, and if $\tau\leq\tau_0+1$,
	$$\tau\geq\max\left\{2+\left\lfloor\log_2\frac{\gamma_k+1}{\gamma_n-4}\right\rfloor,1\right\},$$
	and if $\tau>\tau_0+1$
	\begin{align*}
	\tau&\geq1+\left\lfloor\tau_0+1+\frac{\gamma_k+4\cdot 2^{\tau_0}+1-2^{\tau_0}\gamma_n}{\gamma_n-2}\right\rfloor \\
	&=\left\lfloor\tau_0+2+\frac{\gamma_k+4\cdot 2^{\tau_0}+1-2^{\tau_0}\gamma_n}{\gamma_n-2}\right\rfloor \\
	&=\left\lfloor\tau_0+\frac{\gamma_k+1}{\gamma_n-2}+\nu_0\right\rfloor \\
	&=\tau_0+\left\lfloor\frac{\gamma_k+1}{\gamma_n-2}+\nu_0\right\rfloor,
	\end{align*}
	where
	$$\tau_0=1+\left\lfloor\log_2\frac{\gamma_n-2}{\gamma_n-4}\right\rfloor,\quad\nu_0=2-\frac{2^{\tau_0}(\gamma_n-4)}{\gamma_n-2}\in(0,1].$$
	
\end{proof}

\section*{2. Technical Lemmas}

\begin{lemma}[\texttt{k-grad}]\label{lem:reg0}
	In sparse linear model, under Assumptions~\ref{as:design}~and~\ref{as:noise}, if $n\gg{s^*}^2\log^{3+\kappa} d + {s^*}\log^{5+\kappa} d$, $k\gg{s^*}^2\log^{5+\kappa} d$, and
	\begin{align*}
	\left\|\btheta-\thetas\right\|_1\ll\min\left\{\frac1{ \sqrt{k{s^*}}\log^{1+\kappa}d},\frac1{ \sqrt{n {s^*}}\log^{1+\kappa}d}\right\},
	\end{align*}
	for some $\kappa>0$, then we have that
	\begin{align}
	\sup_{\alpha\in(0,1)}\left|P(T\leq c_{\overline W}(\alpha))-\alpha\right|&=o(1), \quad\text{and} \label{eqn:lem_wb_t} \\
	\sup_{\alpha\in(0,1)}\left|P(\widehat T\leq c_{\overline W}(\alpha))-\alpha\right|&=o(1). \label{eqn:lem_wb_ht}
	\end{align}
\end{lemma}

\begin{proof}[Lemma \ref{lem:reg0}]
	    As noted by \cite{zhang2017simultaneous}, since $\|\sqrt N(\ttheta-\thetas)\|_\infty=\max_l\sqrt N|\ttheta_l-\thetas_l|=\sqrt N\max_l\big((\ttheta_l-\thetas_l)\vee(\thetas_l-\ttheta_l)\big)$, the arguments for the bootstrap consistency result with
    \begin{align}
        T&=\max_l\sqrt N(\ttheta-\thetas)_l\quad\text{and} \label{eqn:tp} \\
         \widehat T&=\max_l\sqrt N(\htheta-\thetas)_l \label{eqn:thp}
    \end{align}
    imply the bootstrap consistency result for $T=\|\sqrt N(\ttheta-\thetas)\|_\infty$ and $\widehat T=\|\sqrt N(\htheta-\thetas)\|_\infty$. Hence, from now on, we redefine $T$ and $\widehat T$ as \eqref{eqn:tp} and \eqref{eqn:thp}. Define an oracle multiplier bootstrap statistic as
	\begin{align}
	W^*\defn\max_{1\leq l\leq d}-\frac1{\sqrt{N}}\sum_{i=1}^n\sum_{j=1}^k \left(\nabla^2\cLs(\thetas)^{-1}\nabla\cL(\thetas;Z_{ij})\right)_l\epsilon_{ij}^*, \label{eqn:wsdef}
	\end{align}
	where $\{\epsilon_{ij}^*\}_{i=1,\dots,n;j=1,\dots,k}$ are $N$ independent standard Gaussian variables, also independent of the entire dataset.  The proof consists of two steps; the first step is to show that $W^*$ achieves bootstrap consistency, i.e., $\sup_{\alpha\in(0,1)}|P(T\leq c_{W^*}(\alpha))-\alpha|$ converges to $0$, where $c_{W^*}(\alpha)=\inf\{t\in\R:P_\epsilon(W^*\leq t)\geq\alpha\},$ and the second step is to show the bootstrap consistency of our proposed bootstrap statistic by showing the quantiles of $W$ and $W^*$ are close.
	
	Note that $\nabla^2\cLs(\thetas)^{-1}\nabla\cL(\thetas;Z)=\Ee[xx^\top]^{-1}x(x^\top\thetas-y)=\Theta xe$ and
	$$\Ee\left[\left(\nabla^2\cLs(\thetas)^{-1}\nabla\cL(\thetas;Z)\right)\left(\nabla^2\cLs(\thetas)^{-1}\nabla\cL(\thetas;Z)\right)^\top\right]=\Theta\Ee\left[xx^\top e^2\right]\Theta=\sigma^2\Theta\Sigma\Theta=\sigma^2\Theta.$$
	Then, under Assumptions~\ref{as:design}~and~\ref{as:noise},
	\begin{align}
	\min_l\Ee\left[\left(\nabla^2\cLs(\thetas)^{-1}\nabla\cL(\thetas;Z)\right)_l^2\right]=\sigma^2\min_l\Theta_{l,l}\geq\sigma^2\lambdamin(\Theta)=\frac{\sigma^2}{\lambdamax(\Sigma)}, \label{eqn:chkvar}
	\end{align}
	is bounded away from zero.  Under Assumption \ref{as:design}, $x$ is sub-Gaussian, that is, $w^\top x$ is sub-Gaussian with uniformly bounded $\psi_2$-norm for all $w\in S^{d-1}$.  To show $w^\top\Theta x$ is also sub-Gaussian with uniformly bounded $\psi_2$-norm, we write it as
	$$w^\top\Theta x = (\Theta w)^\top x = \left\|\Theta w\right\|_2 \left(\frac{\Theta w}{\left\|\Theta w\right\|_2}\right)^\top x.$$
	Since $\Theta w/\left\|\Theta w\right\|_2\in S^{d-1}$, we have that $\left(\Theta w/\left\|\Theta w\right\|_2\right) x$ is sub-Gaussian with $O(1)$ $\psi_2$-norm, and hence, $w^\top\Theta x$ is sub-Gaussian with $O(\left\|\Theta w\right\|_2)=O(\lambdamax(\Theta))=O(\lambdamin(\Sigma)^{-1})=O(1)$ $\psi_2$-norm, under Assumption~\ref{as:design}.  Since $e$ is also sub-Gaussian under Assumption~\ref{as:noise} and is independent of $w^\top\Theta x$, we have that $w^\top\Theta xe$ is sub-exponential with uniformly bounded $\psi_1$-norm for all $w\in S^{d-1}$, and also, all $\left(\nabla^2\cLs(\thetas)^{-1}\nabla\cL(\thetas;Z)\right)_l$ are sub-exponential with uniformly bounded $\psi_1$-norm.  Combining this with \eqref{eqn:chkvar}, we have verified Assumption~(E.1) of \cite{chernozhukov2013gaussian} for $\nabla^2\cLs(\thetas)^{-1}\nabla\cL(\thetas;Z)$.
	
	Define
	\begin{align}
	T_0\defn\max_{1\leq l\leq d}-\sqrt{N}\left(\nabla^2\cLs(\thetas)^{-1}\nabla\cL_N(\thetas)\right)_l, \label{eqn:t0}
	\end{align}
	which is a Bahadur representation of $T$.  Under the condition $\log^7(dN)/N\lesssim N^{-c}$ for some constant $c>0$, which holds if $N\gtrsim\log^{7+\kappa}d$ for some $\kappa>0$, applying Theorem 3.2 and Corollary 2.1 of \cite{chernozhukov2013gaussian}, we obtain that for some constant $c>0$ and for every $v,\zeta>0$, 
	\begin{align}
	\begin{split}
	\sup_{\alpha\in(0,1)}\left|P(T\leq c_{W^*}(\alpha))-\alpha\right| &\lesssim N^{-c}+v^{1/3}\left(1\vee\log\frac dv\right)^{2/3}+P\left(\matrixnorm{\widehat\Omega-\Omega_0}_{\max}>v\right) \\
	&\quad+\zeta\sqrt{1\vee\log\frac d\zeta}+P\left(|T-T_0|>\zeta\right), \label{eqn:ws}
	\end{split}
	\end{align}
	where
	\begin{align}
	\begin{split}
	\widehat\Omega&\defn\cov_\epsilon\left(-\frac1{\sqrt{N}}\sum_{i=1}^n\sum_{j=1}^k\nabla^2\cLs(\thetas)^{-1}\nabla\cL(\thetas;Z_{ij})\epsilon_{ij}^* \right) \\
	&=\nabla^2\cLs(\thetas)^{-1}\left(\frac1N\sum_{i=1}^n\sum_{j=1}^k\nabla\cL(\thetas;Z_{ij})\nabla\cL(\thetas;Z_{ij})^\top\right)\nabla^2\cLs(\thetas)^{-1}, \quad\text{and} \label{eqn:oh}
	\end{split}
	\end{align}
	\begin{align}
	\Omega_0&\defn\cov\left(-\nabla^2\cLs(\thetas)^{-1}\nabla\cL(\thetas;Z)\right)=\nabla^2\cLs(\thetas)^{-1}\Ee\left[\nabla\cL(\thetas;Z)\nabla\cL(\thetas;Z)^\top\right]\nabla^2\cLs(\thetas)^{-1}. \label{eqn:o0}
	\end{align}
	To show the quantiles of $\overline W$ and $W^*$ are close, we first have that for any $\omega$ such that $\alpha+\omega,\alpha-\omega\in(0,1)$,
	\begin{align*}
	&P(\{T\leq c_{\overline W}(\alpha)\}\ominus\{T\leq c_{W^*}(\alpha)\}) \\
	&\leq 2P(c_{W^*}(\alpha-\omega)<T\leq c_{W^*}(\alpha+\omega)) + P(c_{W^*}(\alpha-\omega)>c_{\overline W}(\alpha)) + P(c_{\overline W}(\alpha)>c_{W^*}(\alpha+\omega)),
	\end{align*}
	where $\ominus$ denotes symmetric difference.
	Following the arguments in the proof of Lemma 3.2 of \cite{chernozhukov2013gaussian}, we have that
	$$P(c_{\overline W}(\alpha)>c_{W^*}(\alpha+\pi(u)))\leq P\left(\matrixnorm{\overline\Omega-\widehat\Omega}_{\max}>u\right),\quad\text{and}$$
	$$P(c_{W^*}(\alpha-\pi(u))>c_{\overline W}(\alpha))\leq P\left(\matrixnorm{\overline\Omega-\widehat\Omega}_{\max}>u\right),$$
	where $\pi(u)\defn u^{1/3}\left(1\vee\log(d/u)\right)^{2/3}$ and
	\begin{align}
	\begin{split}
	\overline\Omega&\defn\cov_\epsilon\left(-\frac1{\sqrt{k}}\sum_{j=1}^k \widetilde\Theta \sqrt n\left(\nabla\cL_j(\btheta)-\nabla\cL_N(\btheta)\right)\epsilon_j \right) \\
	&=\widetilde\Theta\left(\frac1k\sum_{j=1}^k n\left(\nabla\cL_j(\btheta)-\nabla\cL_N(\btheta)\right)\left(\nabla\cL_j(\btheta)-\nabla\cL_N(\btheta)\right)^\top\right)\widetilde\Theta^\top. \label{eqn:ob}
	\end{split}
	\end{align}
	By letting $\omega=\pi(u)$, we have that
	\begin{align*}
	&P(\{T\leq c_{\overline W}(\alpha)\}\ominus\{T\leq c_{W^*}(\alpha)\}) \\
	&\leq 2P(c_{W^*}(\alpha-\pi(u))<T\leq c_{W^*}(\alpha+\pi(u))) + P(c_{W^*}(\alpha-\pi(u))>c_{\overline W}(\alpha)) + P(c_{\overline W}(\alpha)>c_{W^*}(\alpha+\pi(u))) \\
	&\leq 2P(c_{W^*}(\alpha-\pi(u))<T\leq c_{W^*}(\alpha+\pi(u))) + 2P\left(\matrixnorm{\overline\Omega-\widehat\Omega}_{\max}>u\right),
	\end{align*}
	where by \eqref{eqn:ws},
	\begin{align*}
	P(c_{W^*}(\alpha-\pi(u))<T\leq c_{W^*}(\alpha+\pi(u))) &=P(T\leq c_{W^*}(\alpha+\pi(u)))-P(T\leq c_{W^*}(\alpha-\pi(u))) \\
	&\lesssim \pi(u)+N^{-c}+\zeta\sqrt{1\vee\log\frac d\zeta}+P\left(|T-T_0|>\zeta\right),
	\end{align*}
	and then,
	\begin{align}
	\sup_{\alpha\in(0,1)}\left|P(T\leq c_{\overline W}(\alpha))-\alpha\right| &\lesssim N^{-c}+v^{1/3}\left(1\vee\log\frac dv\right)^{2/3}+P\left(\matrixnorm{\widehat\Omega-\Omega_0}_{\max}>v\right) \nonumber \\
	&\quad+\zeta\sqrt{1\vee\log\frac d\zeta}+P\left(|T-T_0|>\zeta\right) + u^{1/3}\left(1\vee\log\frac du\right)^{2/3} + P\left(\matrixnorm{\overline\Omega-\widehat\Omega}_{\max}>u\right). \label{eqn:wbbd}
	\end{align}
	
	Applying Lemmas~\ref{lem:tbd_reg},~\ref{lem:gbdo_reg},~and~\ref{lem:gbd0_reg}, we have that there exist some $\zeta,u,v>0$ such that
	\begin{align}
	\zeta\sqrt{1\vee\log\frac d\zeta}&+P\left(|T-T_0|>\zeta\right)=o(1),\quad\text{and} \label{eqn:zeta} \\
	u^{1/3}\left(1\vee\log\frac du\right)^{2/3}&+P\left(\matrixnorm{\overline\Omega-\widehat\Omega}_{\max}>u\right)=o(1),\quad\text{and} \label{eqn:u} \\
	v^{1/3}\left(1\vee\log\frac dv\right)^{2/3}&+P\left(\matrixnorm{\widehat\Omega-\Omega_0}_{\max}>v\right)=o(1), \label{eqn:v}
	\end{align}
	and hence, after simplifying the conditions, obtain the first result in the lemma. To obtain the second result, we use Lemma~\ref{lem:thbd_reg}, which yields
	\begin{align}
	\xi\sqrt{1\vee\log\frac d\xi}+P\left(|\widehat T-T_0|>\xi\right)=o(1). \label{eqn:xi}
	\end{align}
\end{proof}

\begin{lemma}[\texttt{n+k-1-grad}]\label{lem:reg}
	In sparse linear model, under Assumptions~\ref{as:design}~and~\ref{as:noise}, if $n\gg{s^*}^2\log^{3+\kappa} d + {s^*}\log^{5+\kappa} d$, $n+k\gg{s^*}^2\log^{5+\kappa} d$, $nk\gtrsim\log^{7+\kappa}d$, and
	$$\left\|\btheta-\thetas\right\|_1\ll\min\left\{\frac1{ \sqrt{k{s^*}}\log^{1+\kappa}d},\frac1{  {s^*}\sqrt{\log((n+k)d)}\log^{2+\kappa}d}\right\},$$
	for some $\kappa>0$, then we have that
	\begin{align}
	\sup_{\alpha\in(0,1)}\left|P(T\leq c_{\widetilde W}(\alpha))-\alpha\right|&=o(1),\quad\text{and} \label{eqn:lem_wt_t} \\
	\sup_{\alpha\in(0,1)}\left|P(\widehat T\leq c_{\widetilde W}(\alpha))-\alpha\right|&=o(1). \label{eqn:lem_wt_ht}
	\end{align}
\end{lemma}

\begin{proof}[Lemma \ref{lem:reg}]
	By the argument in the proof of Lemma~\ref{lem:reg0}, we have that
	\begin{align}
	\sup_{\alpha\in(0,1)}\left|P(T\leq c_{\widetilde W}(\alpha))-\alpha\right| &\lesssim N^{-c}+v^{1/3}\left(1\vee\log\frac dv\right)^{2/3}+P\left(\matrixnorm{\widehat\Omega-\Omega_0}_{\max}>v\right) \nonumber \\
	&\quad+\zeta\sqrt{1\vee\log\frac d\zeta}+P\left(|T-T_0|>\zeta\right) + u^{1/3}\left(1\vee\log\frac du\right)^{2/3} + P\left(\matrixnorm{\widetilde\Omega-\widehat\Omega}_{\max}>u\right), \label{eqn:wbbd2}
	\end{align}
	where
	\begin{align}
	\begin{split}
	\widetilde\Omega&\defn\cov_\epsilon\left(-\frac1{\sqrt{n+k-1}}\left(\sum_{i=1}^n\widetilde\Theta\left(\nabla\cL(\btheta;Z_{i1})-\nabla\cL_N(\btheta)\right)\epsilon_{i1}+\sum_{j=2}^k\widetilde\Theta\sqrt n\left(\nabla\cL_j(\btheta)-\nabla\cL_N(\btheta)\right)\epsilon_j\right)\right) \\
	&=\widetilde\Theta\frac1{n+k-1}\Bigg(\sum_{i=1}^n\left(\nabla\cL(\theta;Z_{i1})-\nabla\cL_N(\theta)\right)\left(\nabla\cL(\theta;Z_{i1})-\nabla\cL_N(\theta)\right)^\top \\
	&\quad+\sum_{j=2}^k n\left(\nabla\cL_j(\theta)-\nabla\cL_N(\theta)\right)\left(\nabla\cL_j(\theta)-\nabla\cL_N(\theta)\right)^\top\Bigg)\widetilde\Theta^\top, \label{eqn:ot}
	\end{split}
	\end{align}
	if $N\gtrsim\log^{7+\kappa}d$ for some $\kappa>0$.  Applying Lemmas~\ref{lem:tbd_reg},~\ref{lem:gbdo_reg},~and~\ref{lem:gbd_reg}, we have that there exist some $\zeta,u,v>0$ such that \eqref{eqn:zeta},
	\begin{align}
	u^{1/3}\left(1\vee\log\frac du\right)^{2/3}+P\left(\matrixnorm{\widetilde\Omega-\widehat\Omega}_{\max}>u\right)=o(1), \label{eqn:ut}
	\end{align}
	and \eqref{eqn:v} hold, and hence, after simplifying the conditions, obtain the first result in the lemma. To obtain the second result, we use Lemma~\ref{lem:thbd_reg}, which yields \eqref{eqn:xi}.
\end{proof}

\begin{lemma}[\texttt{k-grad}]\label{lem:reg0_glm}
	In sparse GLM, under Assumptions~\ref{as:smth_glm}--\ref{as:subexp_glm}, if $n\gg  (s_0^2+{s^*}^2)\log^{3+\kappa} d+(s_0+{s^*})\log^{5+\kappa} d$, $k\gg{s^*}^2\log^{5+\kappa} d$, and
	\begin{align*}
	\left\|\btheta-\thetas\right\|_1\ll\min\left\{\frac1{ \sqrt{k{s^*}}s_0\log^{1+\kappa}d},\frac1{ \sqrt{n {s^*}}\log^{1+\kappa}d}\right\},
	\end{align*}
	for some $\kappa>0$, then we have that \eqref{eqn:lem_wb_t} and \eqref{eqn:lem_wb_ht} hold.
\end{lemma}

\begin{proof}[Lemma \ref{lem:reg0_glm}]
    We redefine $T$ and $\widehat T$ as \eqref{eqn:tp} and \eqref{eqn:thp}.  We define an oracle multiplier bootstrap statistic as in \eqref{eqn:wsdef}.  Under Assumption~\ref{as:hes_glm},
	\begin{align*}
	\min_l\Ee\left[\left(\nabla^2\cLs(\thetas)^{-1}\nabla\cL(\thetas;Z)\right)_l^2\right]&=\min_l \left(\nabla^2\cLs(\thetas)^{-1}\Ee\left[\nabla\cL(\thetas;Z)\nabla\cL(\thetas;Z)^\top\right]\nabla^2\cLs(\thetas)^{-1}\right)_{l,l} \\
	&\geq\lambdamin\left(\nabla^2\cLs(\thetas)^{-1}\Ee\left[\nabla\cL(\thetas;Z)\nabla\cL(\thetas;Z)^\top\right]\nabla^2\cLs(\thetas)^{-1}\right) \\
	&\geq\lambdamin\left(\nabla^2\cLs(\thetas)^{-1}\right)^2\lambdamin\left(\Ee\left[\nabla\cL(\thetas;Z)\nabla\cL(\thetas;Z)^\top\right]\right) \\
	&=\frac{\lambdamin\left(\Ee\left[\nabla\cL(\thetas;Z)\nabla\cL(\thetas;Z)^\top\right]\right)}{\lambdamax\left(\nabla^2\cLs(\thetas)\right)^2}
	\end{align*}
	is bounded away from zero.  Combining this with Assumption~\ref{as:subexp_glm}, we have verified Assumption~(E.1) of \cite{chernozhukov2013gaussian} for $\nabla^2\cLs(\thetas)^{-1}\nabla\cL(\thetas;Z)$.  Then, we use the same argument as in the proof of Lemma~\ref{lem:reg0}, and obtain \eqref{eqn:wbbd} with
	\begin{align}
	\begin{split}
	\overline\Omega&\defn\tT(\ttheta^{(0)})\left(\frac1k\sum_{j=1}^k n\left(\nabla\cL_j(\btheta)-\nabla\cL_N(\btheta)\right)\left(\nabla\cL_j(\btheta)-\nabla\cL_N(\btheta)\right)^\top\right)\tT(\ttheta^{(0)})^\top, \label{eqn:ob_glm}
	\end{split}
	\end{align}
	under the condition $\log^7(dN)/N\lesssim N^{-c}$ for some constant $c>0$, which holds if $N\gtrsim\log^{7+\kappa}d$ for some $\kappa>0$.  Applying Lemmas~\ref{lem:tbd_reg_glm},~\ref{lem:gbdo_reg_glm},~and~\ref{lem:gbd0_reg_glm}, we have that there exist some $\zeta,u,v>0$ such that \eqref{eqn:zeta}, \eqref{eqn:u}, and \eqref{eqn:v} hold, and hence, after simplifying the conditions, obtain the first result in the lemma. To obtain the second result, we use Lemma~\ref{lem:thbd_reg_glm}, which yields \eqref{eqn:xi}.
\end{proof}

\begin{lemma}[\texttt{n+k-1-grad}]\label{lem:reg_glm}
	In sparse GLM, under Assumptions~\ref{as:smth_glm}--\ref{as:subexp_glm}, if $n\gg  (s_0+{s^*})\log^{5+\kappa} d+(s_0^2+{s^*}^2)\log^{3+\kappa} d$, $n+k\gg{s^*}^2\log^{5+\kappa} d$, $nk\gtrsim\log^{7+\kappa}d$, and
	\begin{align*}
	\left\|\btheta-\thetas\right\|_1&\ll\min\Bigg\{\frac{n+k}{  {s^*}\left(n+k\sqrt{\log d}+k^{3/4}\log^{3/4}d\right)\log^{2+\kappa}d},\frac1{ \sqrt{k{s^*}}s_0\log^{1+\kappa}d}, \frac1{\left(nk{s^*}\log^{1+\kappa}d\right)^{1/4}}\Bigg\},
	\end{align*}
	for some $\kappa>0$, then we have that \eqref{eqn:lem_wt_t} and \eqref{eqn:lem_wt_ht} hold.
\end{lemma}

\begin{proof}[Lemma \ref{lem:reg_glm}]
	By the argument in the proof of Lemma~\ref{lem:reg0_glm}, we obtain \eqref{eqn:wbbd2} with
	\begin{align}
	\begin{split}
	\widetilde\Omega&\defn\tT(\ttheta^{(0)})\frac1{n+k-1}\Bigg(\sum_{i=1}^n\left(\nabla\cL(\theta;Z_{i1})-\nabla\cL_N(\theta)\right)\left(\nabla\cL(\theta;Z_{i1})-\nabla\cL_N(\theta)\right)^\top \\
	&\quad+\sum_{j=2}^k n\left(\nabla\cL_j(\theta)-\nabla\cL_N(\theta)\right)\left(\nabla\cL_j(\theta)-\nabla\cL_N(\theta)\right)^\top\Bigg)\tT(\ttheta^{(0)})^\top, \label{eqn:ot_glm_0}
	\end{split}
	\end{align}
	if $N\gtrsim\log^{7+\kappa}d$ for some $\kappa>0$.  Applying Lemmas~\ref{lem:tbd_reg_glm},~\ref{lem:gbdo_reg_glm},~and~\ref{lem:gbd_reg_glm}, we have that there exist some $\zeta,u,v>0$ such that \eqref{eqn:zeta}, \eqref{eqn:ut}, and \eqref{eqn:v} hold, and hence, after simplifying the conditions, obtain the first result in the lemma. To obtain the second result, we use Lemma~\ref{lem:thbd_reg_glm}, which yields \eqref{eqn:xi}.
\end{proof}

\begin{lemma}\label{lem:tbd_reg}
	$T$ and $T_0$ are defined as in \eqref{eqn:t} and \eqref{eqn:t0} respectively.  In sparse linear model, under Assumptions~\ref{as:design}~and~\ref{as:noise}, provided that $\left\|\btheta-\thetas\right\|_1=O_P(r_{\btheta})$ and $n\gg  {s^*} \log d$, we have that
	$$|T-T_0| = O_P\left(  r_{\btheta} \sqrt{{s^*} k\log d} + \frac{{s^*}\log d}{\sqrt n}\right).$$
	Moreover, if $n\gg{s^*}^2\log^{3+\kappa} d$ and
	$$\left\|\btheta-\thetas\right\|_1\ll\frac1{ \sqrt{k{s^*}}\log^{1+\kappa}d},$$
	for some $\kappa>0$, then there exists some $\zeta>0$ such that \eqref{eqn:zeta} holds.
\end{lemma}

\begin{proof}[Lemma \ref{lem:tbd_reg}]
	First, we note that
	\begin{align*}
	|T-T_0| &\leq \max_{1\leq l\leq d} \left|\sqrt{N}(\ttheta-\thetas)_l+\sqrt{N}\left(\nabla^2\cLs(\thetas)^{-1}\nabla\cL_N(\thetas)\right)_l\right| \\
	&= \sqrt{N}\left\|\ttheta-\thetas+\nabla^2\cLs(\thetas)^{-1}\nabla\cL_N(\thetas)\right\|_\infty,
	\end{align*}
	where we use the fact that $|\max_l a_l - \max_l b_l|\leq\max_l|a_l-b_l|$ for any two vectors $a$ and $b$ of the same dimension.  Next, we bound $\left\|\ttheta-\thetas+\nabla^2\cLs(\thetas)^{-1}\nabla\cL_N(\thetas)\right\|_\infty$.  In linear model, we have that
	$$\ttheta-\thetas+\nabla^2\cLs(\thetas)^{-1}\nabla\cL_N(\thetas)=\btheta+\tT\frac{X_N^\top (y_N-X_N\btheta)}N-\thetas-\Theta\frac{X_N^\top (y_N-X_N\thetas)}N,$$
	and then,
	\begin{align*}
	&\left\|\ttheta-\thetas+\nabla^2\cLs(\thetas)^{-1}\nabla\cL_N(\thetas)\right\|_\infty \\
	&= \left\|\btheta+\tT\frac{X_N^\top (y_N-X_N\btheta)}N-\thetas-\Theta\frac{X_N^\top (y_N-X_N\thetas)}N\right\|_\infty \\
	&= \left\|\btheta+\tT\frac{X_N^\top (y_N-X_N\btheta)}N-\thetas-\tT\frac{X_N^\top (y_N-X_N\thetas)}N+\tT\frac{X_N^\top (y_N-X_N\thetas)}N-\Theta\frac{X_N^\top (y_N-X_N\thetas)}N\right\|_\infty \\
	&\leq \left\|\left(\tT\frac{X_N^\top X_N}N-I_d\right)(\btheta-\thetas)\right\|_\infty + \left\|\left(\tT-\Theta\right)\frac{X_N^\top e_N}N\right\|_\infty \\
	&\leq \matrixnorm{\tT\frac{X_N^\top X_N}N-I_d}_{\max} \left\|\btheta-\thetas\right\|_1 + \matrixnorm{\tT-\Theta}_\infty \left\|\frac{X_N^\top e_N}N\right\|_\infty,
	\end{align*}
	where we use the triangle inequality in the second to last inequality and the fact that for any matrix $A$ and vector $a$ with compatible dimensions, $\|Aa\|_\infty\leq\matrixnorm{A}_{\max}\|a\|_1$ and $\|Aa\|_\infty\leq\matrixnorm{A}_\infty\|a\|_\infty$, in the last inequality.  Further applying the triangle inequality and the fact that for any two matrices $A$ and $B$ with compatible dimensions, $\matrixnorm{AB}_{\max}\leq\matrixnorm{A}_\infty\matrixnorm{B}_{\max}$, we have that
	\begin{align*}
	\matrixnorm{\tT\frac{X_N^\top X_N}N-I_d}_{\max} &=\matrixnorm{\tT\frac{X_N^\top X_N}N-\tT\frac{X_1^\top X_1}n+\tT\frac{X_1^\top X_1}n-I_d}_{\max} \\
	&\leq\matrixnorm{\tT\left(\frac{X_N^\top X_N}N-\frac{X_1^\top X_1}n\right)}_{\max}+\matrixnorm{\tT\frac{X_1^\top X_1}n-I_d}_{\max} \\
	&\leq\matrixnorm{\tT}_\infty \matrixnorm{\frac{X_N^\top X_N}N-\frac{X_1^\top X_1}n}_{\max}+\matrixnorm{\tT\frac{X_1^\top X_1}n-I_d}_{\max}.
	\end{align*}
	Under Assumption~\ref{as:design}, $X_N$ has sub-Gaussian rows.  Then, by Lemma~\ref{lem:hes_hd}, if $n\gg  {s^*} \log d$, we have that
	$$\matrixnorm{\tT}_\infty=\max_l \left\|\tT_l\right\|_1=O_P\left( \sqrt{{s^*}}\right),$$
	$$\matrixnorm{\tT\frac{X_1^\top X_1}n-I_d}_{\max} =O_P\left(\sqrt{\frac{\log d}n}\right),$$
	and
	$$\matrixnorm{\tT-\Theta}_\infty=\max_l\left\|\tT_l-\Theta_l\right\|_1=O_P\left(  {s^*}\sqrt{\frac{\log d}n}\right).$$
	It remains to bound $\matrixnorm{\frac{X_N^\top X_N}N-\frac{X_1^\top X_1}n}_{\max}$ and $\left\|\frac{X_N^\top e_N}N\right\|_\infty$.
	
	Under Assumptions~\ref{as:design}, each $x_{ij,l}$ is sub-Gaussian, and therefore, the product $x_{ij,l}x_{ij,l'}$ of any two is sub-exponential.  By Bernstein's inequality, we have that for any $t>0$,
	$$P\left(\left|\frac{(X_N^\top X_N)_{l,l'}}N-\Sigma_{l,l'}\right| > t\right)\leq 2\exp\left(-cN\left(\frac{t^2}{\Sigma_{l,l'}^2}\wedge\frac t{|\Sigma_{l,l'}|}\right)\right),$$
	or for any $\delta\in(0,1)$,
	$$P\left(\left|\frac{(X_N^\top X_N)_{l,l'}}N-\Sigma_{l,l'}\right| > |\Sigma_{l,l'}|\left(\frac{\log\frac{2d^2}\delta}{cN}\vee\sqrt{\frac{\log\frac{2d^2}\delta}{cN}}\right)\right)\leq \frac\delta{d^2},$$
	for some constant $c>0$.  Then, by the union bound, we have that
	\begin{align}
	P\left(\matrixnorm{\frac{X_N^\top X_N}N-\Sigma}_{\max} > \matrixnorm{\Sigma}_{\max}\left(\frac{\log\frac{2d^2}\delta}{cN}\vee\sqrt{\frac{\log\frac{2d^2}\delta}{cN}}\right)\right)\leq \delta. \label{eqn:bern1}
	\end{align}
	Similarly, we have that
	\begin{align}
	P\left(\matrixnorm{\frac{X_1^\top X_1}n-\Sigma}_{\max} > \matrixnorm{\Sigma}_{\max}\left(\frac{\log\frac{2d^2}\delta}{cn}\vee\sqrt{\frac{\log\frac{2d^2}\delta}{cn}}\right)\right)\leq \delta. \label{eqn:bern3}
	\end{align}
	Then, by the triangle inequality, we have that
	\begin{align*}
	\matrixnorm{\frac{X_N^\top X_N}N-\frac{X_1^\top X_1}n}_{\max} &\leq \matrixnorm{\frac{X_1^\top X_1}n-\Sigma}_{\max} + \matrixnorm{\frac{X_N^\top X_N}N-\Sigma}_{\max} \\
	&\lesssim \matrixnorm{\Sigma}_{\max}\left(\frac{\log\frac{2d^2}\delta}{n}\vee\sqrt{\frac{\log\frac{2d^2}\delta}{n}}\right) \\
	&\lesssim \left(\frac{\log\frac{2d^2}\delta}{n}\vee\sqrt{\frac{\log\frac{2d^2}\delta}{n}}\right),
	\end{align*}
	with probability at least $1-\delta$, where we use $\matrixnorm{\Sigma}_{\max} \leq \matrixnorm{\Sigma}_2 = \lambdamax(\Sigma) = O(1)$ under Assumption~\ref{as:design}.  This implies that
	$$\matrixnorm{\frac{X_N^\top X_N}N-\frac{X_1^\top X_1}n}_{\max} = O_P\left(\sqrt{\frac{\log d}n}\right).$$
	Under Assumptions~\ref{as:design}~and~\ref{as:noise}, each $x_{ij,l}$ and $e_{ij}$ are sub-Gaussian, and therefore, their product $x_{ij,l}e_{ij}$ is sub-exponential.  Applying Bernstein's inequality, we have that for any $\delta\in(0,1)$,
	$$P\left(\left|\frac{(X_N^\top e_N)_l}N\right| > \sqrt{\Sigma_{l,l}}\sigma\left(\frac{\log\frac{2d}\delta}{cN}\vee\sqrt{\frac{\log\frac{2d}\delta}{cN}}\right)\right)\leq \frac\delta d,$$
	for some constant $c>0$.  Then, by the union bound, we have that
	\begin{align}
	P\left(\left\|\frac{X_N^\top e_N}N\right\|_\infty > \max_l \sqrt{\Sigma_{l,l}}\sigma\left(\frac{\log\frac{2d}\delta}{cN}\vee\sqrt{\frac{\log\frac{2d}\delta}{cN}}\right)\right)\leq \delta, \label{eqn:bern2}
	\end{align}
	and then,
	$$\left\|\frac{X_N^\top e_N}N\right\|_\infty = O_P\left(\sqrt{\frac{\log d}N}\right).$$
	Putting all the preceding bounds together, we obtain that
	\begin{align*}
	&\left\|\ttheta-\thetas+\nabla^2\cLs(\thetas)^{-1}\nabla\cL_N(\thetas)\right\|_\infty \\
	&\leq \left(\matrixnorm{\tT}_\infty \matrixnorm{\frac{X_N^\top X_N}N-\frac{X_1^\top X_1}n}_{\max}+\matrixnorm{\tT\frac{X_1^\top X_1}n-I_d}_{\max}\right) \left\|\btheta-\thetas\right\|_1 + \matrixnorm{\tT-\Theta}_\infty \left\|\frac{X_N^\top e_N}N\right\|_\infty \\
	&= \left(O_P\left( \sqrt{{s^*}}\right) O_P\left(\sqrt{\frac{\log d}n}\right) + O_P\left(\sqrt{\frac{\log d}n}\right)\right) O_P(r_{\btheta}) + O_P\left(  {s^*}\sqrt{\frac{\log d}n}\right) O_P\left(\sqrt{\frac{\log d}N}\right) \\
	&= O_P\left( \sqrt{\frac{{s^*}\log d}n}r_{\btheta} + \frac{{s^*}\log d}{n\sqrt k}\right),
	\end{align*}
	where we assume that $\left\|\btheta-\thetas\right\|_1=O_P(r_{\btheta})$, and hence,
	$$|T-T_0| = O_P\left(  r_{\btheta} \sqrt{{s^*} k\log d} + \frac{{s^*}\log d}{\sqrt n}\right).$$
	
	Choosing
	$$\zeta= \left(r_{\btheta} \sqrt{{s^*} k\log d} + \frac{{s^*}\log d}{\sqrt n}\right)^{1-\kappa},$$
	with any $\kappa>0$, we deduce that
	$$P\left(|T-T_0|>\zeta\right)=o(1).$$
	We also have that
	$$\zeta\sqrt{1\vee\log\frac d\zeta}=o(1),$$
	provided that
	$$ \left(r_{\btheta} \sqrt{{s^*} k\log d} + \frac{{s^*}\log d}{\sqrt n}\right) \log^{1/2+\kappa} d =o(1),$$
	which holds if
	$$n\gg{s^*}^2\log^{3+\kappa} d,$$
	and
	$$r_{\btheta}\ll\frac1{ \sqrt{k{s^*}}\log^{1+\kappa}d}.$$
	
\end{proof}

\begin{lemma}\label{lem:thbd_reg}
	$\widehat T$ and $T_0$ are defined as in \eqref{eqn:thp} and \eqref{eqn:t0} respectively.  In sparse linear model, under Assumptions~\ref{as:design}~and~\ref{as:noise}, provided that $n\gg {s^*} \log d$, we have that
	$$|\widehat T-T_0| = O_P\left(\frac{\left(s_0\sqrt{s^*}+{s^*}\right)\log d}{\sqrt n}\right).$$
	Moreover, if $n\gg\left(s_0^2{s^*}+{s^*}^2\right)\log^{3+\kappa} d$ and for some $\kappa>0$, then there exists some $\xi>0$ such that \eqref{eqn:xi} holds.
\end{lemma}

\begin{proof}[Lemma \ref{lem:thbd_reg}]
	By the proof of Lemma~\ref{lem:tbd_reg}, we obtain that
	\begin{align*}
	|\widehat T-T_0| &\leq \max_{1\leq l\leq d} \left|\sqrt{N}(\htheta-\thetas)_l+\sqrt{N}\left(\nabla^2\cLs(\thetas)^{-1}\nabla\cL_N(\thetas)\right)_l\right| \\
	&= \sqrt{N}\left\|\htheta-\thetas+\nabla^2\cLs(\thetas)^{-1}\nabla\cL_N(\thetas)\right\|_\infty \\
	&= \sqrt{N}\left\|\htheta_L+\tT\frac{X_N^\top (y_N-X_N\htheta_L)}N-\thetas-\Theta\frac{X_N^\top (y_N-X_N\thetas)}N\right\|_\infty \\
	&\leq \matrixnorm{\tT\frac{X_N^\top X_N}N-I_d}_{\max} \left\|\htheta_L-\thetas\right\|_1 + \matrixnorm{\tT-\Theta}_\infty \left\|\frac{X_N^\top e_N}N\right\|_\infty \\
	&= O_P\left( \sqrt{{s^*}k\log d}\right) \left\|\htheta_L-\thetas\right\|_1 + O_P\left(\frac{{s^*}\log d}{\sqrt n}\right).
	\end{align*}
	Since
	$$\left\|\htheta_L-\thetas\right\|_1=O_P\left(s_0\sqrt{\frac{\log d}N}\right),$$
	we have that
	\begin{align*}
	|\widehat T-T_0| = O_P\left(\frac{\left(s_0\sqrt{s^*}+{s^*}\right)\log d}{\sqrt n}\right).
	\end{align*}
	Choosing
	$$\xi= \left(\frac{\left(s_0\sqrt{s^*}+{s^*}\right)\log d}{\sqrt n}\right)^{1-\kappa},$$
	with any $\kappa>0$, we deduce that
	$$P\left(|\widehat T-T_0|>\xi\right)=o(1).$$
	We also have that
	$$\xi\sqrt{1\vee\log\frac d\xi}=o(1),$$
	provided that
	$$ \left(\frac{\left(s_0\sqrt{s^*}+{s^*}\right)\log d}{\sqrt n}\right) \log^{1/2+\kappa} d =o(1),$$
	which holds if
	$$n\gg\left(s_0^2{s^*}+{s^*}^2\right)\log^{3+\kappa} d.$$
	
\end{proof}

\begin{lemma}\label{lem:tbd_reg_glm}
	$T$ and $T_0$ are defined as in \eqref{eqn:t} and \eqref{eqn:t0} respectively.  In sparse GLM, under Assumptions~\ref{as:smth_glm}~and~\ref{as:design_glm}, provided that $\left\|\btheta-\thetas\right\|_1=O_P(r_{\btheta})$ and $n\gg s_0^2\log^2 d+{s^*}^2\log d$, we have that
	$$|T-T_0| = O_P\left(  r_{\btheta} \sqrt{{s^*} k\log d} + \frac{{s^*}\log d}{\sqrt n}\right).$$
	Moreover, if $n\gg({s^*}^2+s_0^2) \log^{3+\kappa} d$ and
	$$\left\|\btheta-\thetas\right\|_1\ll\min\left\{\frac1{ \sqrt{k{s^*}}s_0\log^{1+\kappa}d}, \frac1{\left(nk{s^*}\log^{1+\kappa}d\right)^{1/4}}\right\},$$
	for some $\kappa>0$, then there exists some $\zeta>0$ such that \eqref{eqn:zeta} holds.
\end{lemma}

\begin{proof}[Lemma \ref{lem:tbd_reg_glm}]
	Following the argument in the proof of Lemma~\ref{lem:tbd_reg}, we have that
	\begin{align*}
	|T-T_0| &\leq \max_{1\leq l\leq d} \left|\sqrt{N}(\ttheta_l-\thetas_l)+\sqrt{N}\left(\nabla^2\cLs(\thetas)^{-1}\nabla\cL_N(\thetas)\right)_l\right| \\
	&= \sqrt{N}\left\|\ttheta-\thetas+\nabla^2\cLs(\thetas)^{-1}\nabla\cL_N(\thetas)\right\|_\infty,
	\end{align*}
	and
	\begin{align*}
	&\left\|\ttheta-\thetas+\nabla^2\cLs(\thetas)^{-1}\nabla\cL_N(\thetas)\right\|_\infty \\
	&= \left\|\btheta-\tT(\ttheta^{(0)})\nabla\cL_N(\btheta)-\thetas+\Theta\nabla\cL_N(\thetas)\right\|_\infty \\
	&= \left\|\btheta-\tT(\ttheta^{(0)})\nabla\cL_N(\btheta)-\thetas+\tT(\ttheta^{(0)})\nabla\cL_N(\thetas)-\tT(\ttheta^{(0)})\nabla\cL_N(\thetas)+\Theta\nabla\cL_N(\thetas)\right\|_\infty \\
	&\leq \left\|\btheta-\thetas-\tT(\ttheta^{(0)})\left(\nabla\cL_N(\btheta)-\nabla\cL_N(\thetas)\right)\right\|_\infty+\left\|\left(\tT(\ttheta^{(0)})-\Theta\right)\nabla\cL_N(\thetas)\right\|_\infty.
	\end{align*}
	By Taylor's theorem, we have that
	\begin{align}
	\nabla\cL_N(\btheta)-\nabla\cL_N(\thetas)=\int_0^1\nabla^2\cL_N(\thetas+t(\btheta-\thetas))dt(\btheta-\thetas), \label{eqn:graddiff}
	\end{align}
	and then,
	\begin{align*}
	&\left\|\ttheta-\thetas+\nabla^2\cLs(\thetas)^{-1}\nabla\cL_N(\thetas)\right\|_\infty \\
	&\leq \left\|\btheta-\thetas-\tT(\ttheta^{(0)})\int_0^1\nabla^2\cL_N(\thetas+t(\btheta-\thetas))dt(\btheta-\thetas)\right\|_\infty+\left\|\left(\tT(\ttheta^{(0)})-\Theta\right)\nabla\cL_N(\thetas)\right\|_\infty \\
	&= \left\|\int_0^1\left(\tT(\ttheta^{(0)})\nabla^2\cL_N(\thetas+t(\btheta-\thetas))-I_d\right)dt(\btheta-\thetas)\right\|_\infty+\left\|\left(\tT(\ttheta^{(0)})-\Theta\right)\nabla\cL_N(\thetas)\right\|_\infty \\
	&\leq \int_0^1\matrixnorm{\tT(\ttheta^{(0)})\nabla^2\cL_N(\thetas+t(\btheta-\thetas))-I_d}_{\max}dt \left\|\btheta-\thetas\right\|_1 + \matrixnorm{\tT(\ttheta^{(0)})-\Theta}_\infty \left\|\nabla\cL_N(\thetas)\right\|_\infty.
	\end{align*}
	By the triangle inequality, we have that
	\begin{align*}
	&\matrixnorm{\tT(\ttheta^{(0)})\nabla^2\cL_N(\thetas+t(\btheta-\thetas))-I_d}_{\max} \\
	&=\bigg|\!\bigg|\!\bigg|\tT(\ttheta^{(0)})\nabla^2\cL_N(\thetas+t(\btheta-\thetas))-\tT(\ttheta^{(0)})\nabla^2\cL_N(\thetas)+\tT(\ttheta^{(0)})\nabla^2\cL_N(\thetas)-\tT(\ttheta^{(0)})\nabla^2\cL_1(\thetas) \\
	&\quad+\tT(\ttheta^{(0)})\nabla^2\cL_1(\thetas)-\tT(\ttheta^{(0)})\nabla^2\cL_1(\ttheta^{(0)})+\tT(\ttheta^{(0)})\nabla^2\cL_1(\ttheta^{(0)})-I_d\bigg|\!\bigg|\!\bigg|_{\max} \\
	&\leq\matrixnorm{\tT(\ttheta^{(0)})\left(\nabla^2\cL_N(\thetas+t(\btheta-\thetas))-\nabla^2\cL_N(\thetas)\right)}_{\max}+\matrixnorm{\tT(\ttheta^{(0)})\left(\nabla^2\cL_N(\thetas)-\nabla^2\cL_1(\thetas)\right)}_{\max} \\
	&\quad+\matrixnorm{\tT(\ttheta^{(0)})\left(\nabla^2\cL_1(\thetas)-\nabla^2\cL_1(\ttheta^{(0)})\right)}_{\max}+\matrixnorm{\tT(\ttheta^{(0)})\nabla^2\cL_1(\ttheta^{(0)})-I_d}_{\max} \\
	&\leq\matrixnorm{\tT(\ttheta^{(0)})}_\infty \bigg( \matrixnorm{\nabla^2\cL_N(\thetas+t(\btheta-\thetas))-\nabla^2\cL_N(\thetas)}_{\max}+\matrixnorm{\nabla^2\cL_N(\thetas)-\nabla^2\cL_1(\thetas)}_{\max} \\
	&\quad+\matrixnorm{\nabla^2\cL_1(\thetas)-\nabla^2\cL_1(\ttheta^{(0)})}_{\max} \bigg) +\matrixnorm{\tT(\ttheta^{(0)})\nabla^2\cL_1(\ttheta^{(0)})-I_d}_{\max}.
	\end{align*}
	
	Under Assumption~\ref{as:smth_glm}, we have by Taylor's theorem that
	\begin{align*}
	\left|g''(y_{ij},x_{ij}^\top(\thetas+t(\btheta-\thetas)))-g''(y_{ij},x_{ij}^\top\thetas)\right| &= \left|\int_0^1 g'''(y_{ij},x_{ij}^\top(\thetas+st(\btheta-\thetas))) ds \cdot tx_{ij}^\top(\btheta-\thetas)\right| \\
	&\lesssim \left|x_{ij}^\top(\btheta-\thetas)\right|,
	\end{align*}
	and then by the triangle inequality,
	\begin{align}
	\matrixnorm{\nabla^2\cL_N(\thetas+t(\btheta-\thetas))-\nabla^2\cL_N(\thetas)}_{\max} &=  \matrixnorm{\frac1N\sum_{i=1}^n\sum_{j=1}^k x_{ij} x_{ij}^\top \left(g''(y_{ij},x_{ij}^\top(\thetas+t(\btheta-\thetas)))-g''(y_{ij},x_{ij}^\top\thetas)\right)}_{\max} \notag \\
	&\leq \frac1N\sum_{i=1}^n\sum_{j=1}^k \matrixnorm{x_{ij} x_{ij}^\top \left(g''(y_{ij},x_{ij}^\top(\thetas+t(\btheta-\thetas)))-g''(y_{ij},x_{ij}^\top\thetas)\right)}_{\max} \notag \\
	&= \frac1N\sum_{i=1}^n\sum_{j=1}^k \matrixnorm{x_{ij} x_{ij}^\top}_{\max} \left|g''(y_{ij},x_{ij}^\top(\thetas+t(\btheta-\thetas)))-g''(y_{ij},x_{ij}^\top\thetas)\right| \notag \\
	&\lesssim \frac1N\sum_{i=1}^n\sum_{j=1}^k \|x_{ij}\|_\infty^2 \left|x_{ij}^\top(\btheta-\thetas)\right| \notag \\
	&\leq \frac1N\sum_{i=1}^n\sum_{j=1}^k \|x_{ij}\|_\infty^3 \|\btheta-\thetas\|_1 \notag \\
	&\lesssim \left\|\btheta-\thetas\right\|_1, \label{eqn:lip}
	\end{align}
	where we use that $\|x_{ij}\|_\infty=O(1)$ under Assumption~\ref{as:design_glm} in the last inequality.  Similarly, we have that
	$$\matrixnorm{\nabla^2\cL_1(\thetas)-\nabla^2\cL_1(\ttheta^{(0)})}_{\max} \lesssim \|\ttheta^{(0)}-\thetas\|_1=O_P\left(s_0\sqrt{\frac{\log d}n}\right),$$
	by noticing that $\ttheta^{(0)}$ is a local Lasso estimator computed using $n$ observations.  Note that
	\begin{align*}
	\matrixnorm{\nabla^2\cL_N(\thetas)-\nabla^2\cLs(\thetas)}_{\max} = \matrixnorm{\frac1N\sum_{i=1}^n\sum_{j=1}^k g''(y_{ij},x_{ij}^\top\thetas) x_{ij} x_{ij}^\top-\Ee[g''(y,x^\top\thetas)xx^\top]}_{\max},
	\end{align*}
	and $g''(y_{ij},x_{ij}^\top\thetas)=O(1)$ under Assumption~\ref{as:smth_glm}.  Then, we have that by Hoeffding's inequality,
	$$P\left(\frac{\sum_{i=1}^n\sum_{j=1}^k g''(y_{ij},x_{ij}^\top\thetas) x_{ij,l} x_{ij,l'}}N-\Ee[g''(y,x^\top\thetas) x_l x_{l'}]>\sqrt{\frac{2\log(\frac{2d^2}\delta)}N}\right)\leq\frac\delta{d^2},$$
	and by the union bound, for any $\delta\in(0,1)$, with probability at least $1-\delta$,
	$$\matrixnorm{\nabla^2\cL_N(\thetas)-\nabla^2\cLs(\thetas)}_{\max}\leq\sqrt{\frac{2\log(\frac{2d^2}\delta)}N},$$
	which implies that
	\begin{align}
	\matrixnorm{\nabla^2\cL_N(\thetas)-\nabla^2\cLs(\thetas)}_{\max}=O_P\left(\sqrt{\frac{\log d}N}\right). \label{eqn:hoef3}
	\end{align}
	Similarly, we have that
	$$\matrixnorm{\nabla^2\cL_1(\thetas)-\nabla^2\cLs(\thetas)}_{\max}=O_P\left(\sqrt{\frac{\log d}n}\right),$$
	and then, by the triangle inequality,
	\begin{align*}
	\matrixnorm{\nabla^2\cL_N(\thetas)-\nabla^2\cL_1(\thetas)}_{\max} &\leq \matrixnorm{\nabla^2\cL_N(\thetas)-\nabla^2\cLs(\thetas)}_{\max} + \matrixnorm{\nabla^2\cL_1(\thetas)-\nabla^2\cLs(\thetas)}_{\max} \\
	&=O_P\left(\sqrt{\frac{\log d}n}\right).
	\end{align*}
	Note that $\nabla\cL_N(\thetas)=\sum_{i=1}^n\sum_{j=1}^k g'(y_{ij},x_{ij}^\top\thetas) x_{ij}/N$ and $g'(y_{ij},x_{ij}^\top\thetas) x_{ij,l}=O(1)$ for each $l=1,\dots,d$ under Assumptions~\ref{as:smth_glm}~and~\ref{as:design_glm}.  Then, by Hoeffding's inequality, we have that
	\begin{align}
	P\left(\left|\nabla\cL_N(\thetas)_l\right|>t\right)\leq2\exp\left(-\frac{Nt^2}c\right), \label{eqn:hoef1}
	\end{align}
	for any $t>0$, or
	$$P\left(\left|\nabla\cL_N(\thetas)_l\right|>\sqrt{\frac{c\log\frac{2d}\delta}N}\right)\leq\frac\delta d,$$
	for any $\delta\in(0,1)$.  By the union bound, we have with probability at least $1-\delta$ that
	$$\left\|\nabla\cL_N(\thetas)\right\|_\infty\leq\sqrt{\frac{c\log\frac{2d}\delta}N},$$
	which implies that
	\begin{align}
	\left\|\nabla\cL_N(\thetas)\right\|_\infty=O_P\left(\sqrt{\frac{\log d}N}\right). \label{eqn:hoef2}
	\end{align}
	By Lemma~\ref{lem:hes_hd_glm}, provided that $n\gg s_0^2\log^2 d+{s^*}^2\log d$, we have that
	$$\matrixnorm{\tT(\ttheta^{(0)})}_\infty=O_P\left( \sqrt{{s^*}}\right),$$
	$$\matrixnorm{\tT(\ttheta^{(0)})\nabla^2\cL_1(\ttheta^{(0)})-I_d}_{\max}=O_P\left(\sqrt{\frac{\log d}n}\right),$$
	and
	$$\matrixnorm{\tT(\ttheta^{(0)})-\Theta}_\infty=O_P\left(\left(s_0 +   {s^*}\right)\sqrt{\frac{\log d}n}\right).$$
	Putting all the preceding bounds together, we obtain that
	\begin{align*}
	&\matrixnorm{\tT(\ttheta^{(0)})\nabla^2\cL_N(\thetas+t(\btheta-\thetas))-I_d}_{\max} \\
	&= O_P\left( \sqrt{{s^*}}\right) \left( O_P(r_{\btheta})+O_P\left(\sqrt{\frac{\log d}n}\right)+O_P\left(s_0\sqrt{\frac{\log d}n}\right) \right) +O_P\left(\sqrt{\frac{\log d}n}\right) \\
	&= O_P\left( \sqrt{{s^*}} \left(r_{\btheta}+s_0\sqrt{\frac{\log d}n}\right)\right),
	\end{align*}
	and then,
	\begin{align*}
	&\left\|\ttheta-\thetas+\nabla^2\cLs(\thetas)^{-1}\nabla\cL_N(\thetas)\right\|_\infty \\
	&= O_P\left( \sqrt{{s^*}} \left(r_{\btheta}+s_0\sqrt{\frac{\log d}n}\right)\right) O_P(r_{\btheta}) + O_P\left(\left(s_0 +   {s^*}\right)\sqrt{\frac{\log d}n}\right) O_P\left(\sqrt{\frac{\log d}N}\right) \\
	&= O_P\left(  \sqrt{{s^*}} \left(r_{\btheta}+s_0\sqrt{\frac{\log d}n}\right) r_{\btheta} + \left(s_0 + {s^*}\right) \frac{\log d}{n\sqrt k}\right),
	\end{align*}
	where we assume that $\left\|\btheta-\thetas\right\|_1=O_P(r_{\btheta})$, and hence,
	$$|T-T_0| = O_P\left(  \sqrt{{s^*}} \left(\sqrt n r_{\btheta}+s_0\sqrt{\log d}\right) \sqrt  k r_{\btheta} + \left(s_0 + {s^*}\right) \frac{\log d}{\sqrt n}\right).$$
	
	Choosing
	$$\zeta= \left(\sqrt{{s^*}} \left(\sqrt n r_{\btheta}+s_0\sqrt{\log d}\right) \sqrt  k r_{\btheta} + \left(s_0 + {s^*}\right) \frac{\log d}{\sqrt n}\right)^{1-\kappa},$$
	with any $\kappa>0$, we deduce that
	$$P\left(|T-T_0|>\zeta\right)=o(1).$$
	We also have that
	$$\zeta\sqrt{1\vee\log\frac d\zeta}=o(1),$$
	provided that
	$$ \left(\sqrt{{s^*}} \left(\sqrt n r_{\btheta}+s_0\sqrt{\log d}\right) \sqrt  k r_{\btheta} + \left(s_0 + {s^*}\right) \frac{\log d}{\sqrt n}\right) \log^{1/2+\kappa} d =o(1),$$
	which holds if
	$$n\gg\left({s^*}^2+s_0^2\right) \log^{3+\kappa} d,$$
	and
	$$r_{\btheta}\ll\min\left\{\frac1{ \sqrt{k{s^*}}s_0\log^{1+\kappa}d}, \frac1{\left(nk{s^*}\log^{1+\kappa}d\right)^{1/4}}\right\}.$$
	
\end{proof}

\begin{lemma}\label{lem:thbd_reg_glm}
	$\widehat T$ and $T_0$ are defined as in \eqref{eqn:thp} and \eqref{eqn:t0} respectively.  In sparse GLM, under Assumptions~\ref{as:smth_glm}~and~\ref{as:design_glm}, provided that $n\gg s_0^2\log^2 d+{s^*}^2\log d$, we have that
	$$|\widehat T-T_0| = O_P\left(\frac{\left(s_0^2\sqrt{s^*}+{s^*}\right)\log d}{\sqrt n}\right).$$
	Moreover, if $n\gg\left(s^4{s^*}+{s^*}^2\right) \log^{3+\kappa} d$ for some $\kappa>0$, then there exists some $\xi>0$ such that \eqref{eqn:xi} holds.
\end{lemma}

\begin{proof}[Lemma \ref{lem:thbd_reg_glm}]
	By the proof of Lemma~\ref{lem:tbd_reg_glm}, we obtain that
	\begin{align*}
	|\widehat T-T_0| &\leq \max_{1\leq l\leq d} \left|\sqrt{N}(\htheta-\thetas)_l+\sqrt{N}\left(\nabla^2\cLs(\thetas)^{-1}\nabla\cL_N(\thetas)\right)_l\right| \\
	&= \sqrt{N}\left\|\htheta-\thetas+\nabla^2\cLs(\thetas)^{-1}\nabla\cL_N(\thetas)\right\|_\infty \\
	&= \sqrt{N}\left\|\htheta_L-\tT(\ttheta^{(0)})\nabla\cL_N(\htheta_L)-\thetas+\Theta\nabla\cL_N(\thetas)\right\|_\infty \\
	&\leq \int_0^1\matrixnorm{\tT(\ttheta^{(0)})\nabla^2\cL_N(\thetas+t(\htheta_L-\thetas))-I_d}_{\max}dt \left\|\htheta_L-\thetas\right\|_1 + \matrixnorm{\tT(\ttheta^{(0)})-\Theta}_\infty \left\|\nabla\cL_N(\thetas)\right\|_\infty \\
	&= O_P\left(\sqrt{nk{s^*}}\right) \left\|\htheta_L-\thetas\right\|_1^2 + O_P\left( s_0\sqrt{k{s^*}\log d}\right) \left\|\htheta_L-\thetas\right\|_1 + O_P\left(\left(s_0 + {s^*}\right) \frac{\log d}{\sqrt n}\right).
	\end{align*}
	Since
	$$\left\|\htheta_L-\thetas\right\|_1=O_P\left(s_0\sqrt{\frac{\log d}N}\right),$$
	we have that
	\begin{align*}
	|\widehat T-T_0| = O_P\left(\frac{\left(s_0^2\sqrt{s^*}+{s^*}\right)\log d}{\sqrt n}\right).
	\end{align*}
	Choosing
	$$\xi= \left(\frac{\left(s_0^2\sqrt{s^*}+{s^*}\right)\log d}{\sqrt n}\right)^{1-\kappa},$$
	with any $\kappa>0$, we deduce that
	$$P\left(|\widehat T-T_0|>\xi\right)=o(1).$$
	We also have that
	$$\xi\sqrt{1\vee\log\frac d\xi}=o(1),$$
	provided that
	$$ \left(\frac{\left(s_0^2\sqrt{s^*}+{s^*}\right)\log d}{\sqrt n}\right) \log^{1/2+\kappa} d =o(1),$$
	which holds if
	$$n\gg\left(s^4{s^*}+{s^*}^2\right) \log^{3+\kappa} d.$$
	
\end{proof}

\begin{lemma}\label{lem:gbd0_reg}
	$\overline\Omega$ and $\widehat\Omega$ are defined as in \eqref{eqn:ob} and \eqref{eqn:oh} respectively.  In sparse linear model, under Assumptions~\ref{as:design}~and~\ref{as:noise}, provided that $\left\|\btheta-\thetas\right\|_1=O_P(r_{\btheta})$, $r_{\btheta}\sqrt{\log(kd)}\lesssim 1$, $n\gg  {s^*} \log d$, and $k\gtrsim\log^2(dk)\log d$, we have that
	$$\matrixnorm{\overline\Omega-\widehat\Omega}_{\max} = O_P\left(  {s^*}\left(\sqrt{\frac{\log d}k} + \frac{\log^2(dk)\log d}k + \sqrt{\log(kd)}r_{\btheta} + nr_{\btheta}^2\right) + \sqrt{\frac{{s^*}\log d}n}\right).$$
	Moreover, if $n\gg  {s^*}\log^{5+\kappa} d$, $k\gg{s^*}^2\log^{5+\kappa} d$, and
	$$\left\|\btheta-\thetas\right\|_1\ll\min\left\{\frac1{  {s^*}\sqrt{\log(kd)}\log^{2+\kappa}d},\frac1{ \sqrt{n {s^*}}\log^{1+\kappa}d}\right\},$$
	for some $\kappa>0$, then there exists some $u>0$ such that \eqref{eqn:u} holds.
\end{lemma}

\begin{proof}[Lemma \ref{lem:gbd0_reg}]
	Note by the triangle inequality that
	$$\matrixnorm{\overline\Omega-\widehat\Omega}_{\max}\leq\matrixnorm{\overline\Omega-\Omega_0}_{\max} + \matrixnorm{\widehat\Omega-\Omega_0}_{\max},$$
	where $\Omega_0$ is defined as in \eqref{eqn:o0}.  First, we bound $\matrixnorm{\widehat\Omega-\Omega_0}_{\max}$.  With Assumption~(E.1) of \citet{chernozhukov2013gaussian} verified for $\nabla^2\cLs(\thetas)^{-1}\nabla\cL(\thetas;Z)$ in the proof of Lemma~\ref{lem:reg0}, by the proof of Corollary 3.1 of \citet{chernozhukov2013gaussian}, we have that
	\begin{align*}
	\Ee\left[\matrixnorm{\widehat\Omega-\Omega_0}_{\max}\right]\lesssim \sqrt{\frac{\log d}N} + \frac{\log^2(dN)\log d}N,
	\end{align*}
	and then, by Markov's inequality, with probability at least $1-\delta$,
	\begin{align*}
	\matrixnorm{\widehat\Omega-\Omega_0}_{\max}&\lesssim\frac1\delta\left( \sqrt{\frac{\log d}N} + \frac{\log^2(dN)\log d}N\right),
	\end{align*}
	for any $\delta\in(0,1)$, which implies that
	$$\matrixnorm{\widehat\Omega-\Omega_0}_{\max} = O_P\left(\sqrt{\frac{\log d}N} + \frac{\log^2(dN)\log d}N\right).$$

	Next, we bound $\matrixnorm{\overline\Omega-\Omega_0}_{\max}$.  By the triangle inequality, we have that
	\begin{align*}
	&\matrixnorm{\overline\Omega-\Omega_0}_{\max}\\
	&=\matrixnorm{\tT\left(\frac1k\sum_{j=1}^k n\left(\nabla\cL_j(\btheta)-\nabla\cL_N(\btheta)\right)\left(\nabla\cL_j(\btheta)-\nabla\cL_N(\btheta)\right)^\top\right)\tT^\top-\Theta\Ee\left[\nabla\cL(\thetas;Z)\nabla\cL(\thetas;Z)^\top\right]\Theta}_{\max} \\
	&\leq \matrixnorm{\tT\left(\frac1k\sum_{j=1}^k n\left(\nabla\cL_j(\btheta)-\nabla\cL_N(\btheta)\right)\left(\nabla\cL_j(\btheta)-\nabla\cL_N(\btheta)\right)^\top-\Ee\left[\nabla\cL(\thetas;Z)\nabla\cL(\thetas;Z)^\top\right]\right)\tT}_{\max} \\
	&\quad+ \matrixnorm{\tT\Ee\left[\nabla\cL(\thetas;Z)\nabla\cL(\thetas;Z)^\top\right]\tT^\top-\Theta\Ee\left[\nabla\cL(\thetas;Z)\nabla\cL(\thetas;Z)^\top\right]\Theta}_{\max} \\
	&\defn I_1(\btheta) + I_2.
	\end{align*}
	To bound $I_1(\btheta)$, we use the fact that for any two matrices $A$ and $B$ with compatible dimensions, $\matrixnorm{AB}_{\max}\leq\matrixnorm{A}_\infty\matrixnorm{B}_{\max}$ and $\matrixnorm{AB}_{\max}\leq\matrixnorm{A}_{\max}\matrixnorm{B}_1$, and obtain that
	\begin{align*}
	I_1(\btheta) &\leq \matrixnorm{\tT}_\infty \matrixnorm{\frac1k\sum_{j=1}^k n\left(\nabla\cL_j(\btheta)-\nabla\cL_N(\btheta)\right)\left(\nabla\cL_j(\btheta)-\nabla\cL_N(\btheta)\right)^\top-\Ee\left[\nabla\cL(\thetas;Z)\nabla\cL(\thetas;Z)^\top\right]}_{\max} \matrixnorm{\tT^\top}_1 \\
	&= \matrixnorm{\tT}_\infty^2 \matrixnorm{\frac1k\sum_{j=1}^k n\left(\nabla\cL_j(\btheta)-\nabla\cL_N(\btheta)\right)\left(\nabla\cL_j(\btheta)-\nabla\cL_N(\btheta)\right)^\top-\Ee\left[\nabla\cL(\thetas;Z)\nabla\cL(\thetas;Z)^\top\right]}_{\max}.
	\end{align*}
	Under Assumption~\ref{as:design}, by Lemma~\ref{lem:hes_hd}, if $n\gg  {s^*} \log d$, we have that
	$$\matrixnorm{\tT}_\infty=\max_l \left\|\tT_l\right\|_1=O_P\left( \sqrt{{s^*}}\right).$$
	Then, applying Lemma~\ref{lem:vcov0_reg}, we have that
	\begin{align*}
	I_1(\btheta) &= O_P\left(  {s^*}\right) O_P\left(\sqrt{\frac{\log d}k} + \frac{\log^2(dk)\log d}k + \sqrt{\log(kd)}r_{\btheta} + nr_{\btheta}^2\right) \\
	&= O_P\left(  {s^*}\left(\sqrt{\frac{\log d}k} + \frac{\log^2(dk)\log d}k + \sqrt{\log(kd)}r_{\btheta} + nr_{\btheta}^2\right)\right),
	\end{align*}
	under Assumptions~\ref{as:design}~and~\ref{as:noise}, provided that $\left\|\btheta-\thetas\right\|_1=O_P(r_{\btheta})$, $r_{\btheta}\sqrt{\log(kd)}\lesssim 1$, and $k\gtrsim\log^2(dk)\log d$.
	
	It remains to bound $I_2$.  In linear model, we have that
	$$I_2=\matrixnorm{\tT\left(\sigma^2\Sigma\right)\tT^\top-\Theta\left(\sigma^2\Sigma\right)\Theta}_{\max}=\sigma^2\matrixnorm{\tT\Sigma\tT^\top-\Theta}_{\max},$$
	and by the triangle inequality,
	\begin{align*}
	I_2&=\sigma^2\matrixnorm{(\tT-\Theta+\Theta)\Sigma(\tT-\Theta+\Theta)^\top-\Theta}_{\max} \\
	&=\sigma^2\matrixnorm{(\tT-\Theta)\Sigma(\tT-\Theta)^\top+\Theta\Sigma(\tT-\Theta)^\top+(\tT-\Theta)\Sigma\Theta+\Theta\Sigma\Theta-\Theta}_{\max} \\
	&\leq\sigma^2\matrixnorm{(\tT-\Theta)\Sigma(\tT-\Theta)^\top}_{\max} + 2\sigma^2\matrixnorm{\tT-\Theta}_{\max}.
	\end{align*}
	By Lemma~\ref{lem:hes_hd}, we have that
	$$\matrixnorm{\tT-\Theta}_{\max}\leq\max_l\left\|\tT_l-\Theta_l\right\|_2=O_P\left( \sqrt{\frac{{s^*}\log d}n}\right),$$
	and
	\begin{align*}
	\matrixnorm{(\tT-\Theta)\Sigma(\tT-\Theta)^\top}_{\max}&\leq\matrixnorm{\Sigma}_2\max_l\left\|\tT_l-\Theta_l\right\|_2^2=O_P\left( \frac{{s^*}\log d}n\right),
	\end{align*}
	where we use that $\matrixnorm{\Sigma}_{\max}\leq\matrixnorm{\Sigma}_2=O(1)$ under Assumption~\ref{as:design}.  Then, we obtain that
	$$I_2 = O_P\left( \frac{{s^*}\log d}n\right) + O_P\left( \sqrt{\frac{{s^*}\log d}n}\right) = O_P\left( \sqrt{\frac{{s^*}\log d}n}\right).$$
	Putting all the preceding bounds together, we obtain that
	\begin{align*}
        \matrixnorm{\overline\Omega-\Omega_0}_{\max} &= O_P\left(  {s^*}\left(\sqrt{\frac{\log d}k} + \frac{\log^2(dk)\log d}k + \sqrt{\log(kd)}r_{\btheta} + nr_{\btheta}^2\right) + \sqrt{\frac{{s^*}\log d}n}\right),
        \end{align*}
        and
        \begin{align*}
        \matrixnorm{\overline\Omega-\widehat\Omega}_{\max} &= O_P\left(  {s^*}\left(\sqrt{\frac{\log d}k} + \frac{\log^2(dk)\log d}k + \sqrt{\log(kd)}r_{\btheta} + nr_{\btheta}^2\right) + \sqrt{\frac{{s^*}\log d}n}\right).
        \end{align*}
        
        Choosing
	$$u= \left({s^*}\sqrt{\frac{\log d}k} + \frac{{s^*}\log^2(dk)\log d}k + {s^*}\sqrt{\log(kd)}r_{\btheta} + n {s^*} r_{\btheta}^2 + \sqrt{\frac{{s^*}\log d}n}\right)^{1-\kappa},$$
	with any $\kappa>0$, we deduce that
	$$P\left(\matrixnorm{\overline\Omega-\widehat\Omega}_{\max}>u\right)=o(1).$$
	We also have that
	$$u^{1/3}\left(1\vee\log\frac du\right)^{2/3}=o(1),$$
	provided that
	$$ \left({s^*}\sqrt{\frac{\log d}k} + \frac{{s^*}\log^2(dk)\log d}k + {s^*}\sqrt{\log(kd)}r_{\btheta} + n {s^*} r_{\btheta}^2 + \sqrt{\frac{{s^*}\log d}n}\right) \log^{2+\kappa} d =o(1),$$
	which holds if
	$$n\gg  {s^*}\log^{5+\kappa} d,$$
	$$k\gg{s^*}^2\log^{5+\kappa} d,$$
	and
	$$r_{\btheta}\ll\min\left\{\frac1{  {s^*}\sqrt{\log(kd)}\log^{2+\kappa}d},\frac1{ \sqrt{n {s^*}}\log^{1+\kappa}d}\right\}.$$
	
\end{proof}

\begin{lemma}\label{lem:gbdo_reg}
	$\widehat\Omega$ and $\Omega_0$ is defined as in \eqref{eqn:oh} and \eqref{eqn:o0} respectively.  In sparse linear model, under Assumptions~\ref{as:design}~and~\ref{as:noise}, we have that
	$$\matrixnorm{\widehat\Omega-\Omega_0}_{\max} = O_P\left(\sqrt{\frac{\log d}N} + \frac{\log^2(dN)\log d}N\right).$$
	Moreover, if $N\gg\log^{5+\kappa}d$ for some $\kappa>0$, then there exists some $v>0$ such that \eqref{eqn:v} holds.
\end{lemma}

\begin{proof}[Lemma \ref{lem:gbdo_reg}]
	In the proof of Lemma~\ref{lem:gbd0_reg}, we have shown that
	$$\matrixnorm{\widehat\Omega-\Omega_0}_{\max} = O_P\left(\sqrt{\frac{\log d}N} + \frac{\log^2(dN)\log d}N\right).$$
        Choosing
	$$v=\left(\sqrt{\frac{\log d}N} + \frac{\log^2(dN)\log d}N\right)^{1-\kappa},$$
	with any $\kappa>0$, we deduce that
	$$P\left(\matrixnorm{\widehat\Omega-\Omega_0}_{\max}>v\right)=o(1).$$
	We also have that
	$$v^{1/3}\left(1\vee\log\frac dv\right)^{2/3}=o(1),$$
	provided that
	$$\left(\sqrt{\frac{\log d}N} + \frac{\log^2(dN)\log d}N\right) \log^{2+\kappa} d=o(1),$$
	which holds if
	$$N\gg\log^{5+\kappa}d.$$
	The same result applies to the low-dimensional case as well.
\end{proof}

\begin{lemma}\label{lem:gbd_reg}
	$\widetilde\Omega$ and $\widehat\Omega$ are defined as in \eqref{eqn:ot} and \eqref{eqn:oh} respectively.  In sparse linear model, under Assumptions~\ref{as:design}~and~\ref{as:noise}, provided that $\left\|\btheta-\thetas\right\|_1=O_P(r_{\btheta})$, $r_{\btheta}\sqrt{\log((n+k)d)}\lesssim 1$, $n\gg  {s^*} \log d$, and $n+k\gtrsim\log^3 d$, we have that
	$$\matrixnorm{\widetilde\Omega-\widehat\Omega}_{\max} = O_P\left(  {s^*}\left(\sqrt{\frac{\log d}{n+k}} + \frac{\log^2(d(n+k))\log d}{n+k} + \sqrt{\log((n+k)d)} r_{\btheta} + \frac{nk}{n+k}r_{\btheta}^2\right) + \sqrt{\frac{{s^*}\log d}n}\right).$$
	Moreover, if $n\gg  {s^*}\log^{5+\kappa} d$, $n+k\gg{s^*}^2\log^{5+\kappa} d$, and
	$$\left\|\btheta-\thetas\right\|_1\ll\min\left\{\frac1{  {s^*}\sqrt{\log((n+k)d)}\log^{2+\kappa}d},\frac1{ \sqrt{{s^*}}\log^{1+\kappa}d}\sqrt{\frac1n+\frac1k}\right\},$$
	for some $\kappa>0$, then there exists some $u>0$ such that \eqref{eqn:ut} holds.
\end{lemma}

\begin{proof}[Lemma \ref{lem:gbd_reg}]
	Note by the triangle inequality that
	$$\matrixnorm{\widetilde\Omega-\widehat\Omega}_{\max}\leq\matrixnorm{\widetilde\Omega-\Omega_0}_{\max} + \matrixnorm{\widehat\Omega-\Omega_0}_{\max},$$
	where $\Omega_0$ is defined as in \eqref{eqn:o0}.  By the proof of Lemma~\ref{lem:gbd0_reg}, we have that
	$$\matrixnorm{\widehat\Omega-\Omega_0}_{\max} = O_P\left(\sqrt{\frac{\log d}N} + \frac{\log^2(dN)\log d}N\right).$$
	Next, we bound $\matrixnorm{\widetilde\Omega-\Omega_0}_{\max}$ using the same argument as in the proof of Lemma~\ref{lem:gbd0_reg}.  By the triangle inequality, we have that
	\begin{align*}
	&\matrixnorm{\widetilde\Omega-\Omega_0}_{\max}\\
	&=\Bigg|\!\Bigg|\!\Bigg|\tT\frac1{n+k-1}\Bigg(\sum_{i=1}^n\left(\nabla\cL(\theta;Z_{i1})-\nabla\cL_N(\btheta)\right)\left(\nabla\cL(\theta;Z_{i1})-\nabla\cL_N(\btheta)\right)^\top \\
	&\quad+\sum_{j=2}^k n\left(\nabla\cL_j(\btheta)-\nabla\cL_N(\btheta)\right)\left(\nabla\cL_j(\btheta)-\nabla\cL_N(\btheta)\right)^\top\Bigg)\tT^\top-\Theta\Ee\left[\nabla\cL(\thetas;Z)\nabla\cL(\thetas;Z)^\top\right]\Theta\Bigg|\!\Bigg|\!\Bigg|_{\max} \\
	&\leq \Bigg|\!\Bigg|\!\Bigg|\tT\Bigg(\frac1{n+k-1}\Bigg(\sum_{i=1}^n\left(\nabla\cL(\theta;Z_{i1})-\nabla\cL_N(\btheta)\right)\left(\nabla\cL(\theta;Z_{i1})-\nabla\cL_N(\btheta)\right)^\top \\
	&\quad+\sum_{j=2}^k n\left(\nabla\cL_j(\btheta)-\nabla\cL_N(\btheta)\right)\left(\nabla\cL_j(\btheta)-\nabla\cL_N(\btheta)\right)^\top\Bigg)-\Ee\left[\nabla\cL(\thetas;Z)\nabla\cL(\thetas;Z)^\top\right]\Bigg)\tT^\top\Bigg|\!\Bigg|\!\Bigg|_{\max} \\
	&\quad+ \matrixnorm{\tT\Ee\left[\nabla\cL(\thetas;Z)\nabla\cL(\thetas;Z)^\top\right]\tT^\top-\Theta\Ee\left[\nabla\cL(\thetas;Z)\nabla\cL(\thetas;Z)^\top\right]\Theta}_{\max} \\
	&\defn I_1'(\btheta) + I_2.
	\end{align*}
	We have shown in the proof of Lemma~\ref{lem:gbd0_reg} that
	$$I_2 = O_P\left( \sqrt{\frac{{s^*}\log d}n}\right).$$
	To bound $I_1'(\btheta)$, we note that
	\begin{align*}
	I_1'(\btheta) &\leq \matrixnorm{\tT}_\infty^2 \Bigg|\!\Bigg|\!\Bigg|\frac1{n+k-1}\Bigg(\sum_{i=1}^n\left(\nabla\cL(\btheta;Z_{i1})-\nabla\cL_N(\btheta)\right)\left(\nabla\cL(\btheta;Z_{i1})-\nabla\cL_N(\btheta)\right)^\top \\
	&\quad+\sum_{j=2}^k n\left(\nabla\cL_j(\btheta)-\nabla\cL_N(\btheta)\right)\left(\nabla\cL_j(\btheta)-\nabla\cL_N(\btheta)\right)^\top\Bigg) -\Ee\left[\nabla\cL(\thetas;Z)\nabla\cL(\thetas;Z)^\top\right]\Bigg|\!\Bigg|\!\Bigg|_{\max}.
	\end{align*}
	Under Assumption~\ref{as:design}, by Lemma~\ref{lem:hes_hd}, if $n\gg  {s^*} \log d$, we have that
	$$\matrixnorm{\tT}_\infty=\max_l \left\|\tT_l\right\|_1=O_P\left( \sqrt{{s^*}}\right).$$
	Then, applying Lemma~\ref{lem:vcov_reg}, we have that
	\begin{align*}
	I_1'(\btheta) = O_P\left(  {s^*}\left(\sqrt{\frac{\log d}{n+k}} + \frac{\log^2(d(n+k))\log d}{n+k} + \sqrt{\log((n+k)d)} r_{\btheta} + \frac{nk}{n+k}r_{\btheta}^2\right)\right),
	\end{align*}
	under Assumptions~\ref{as:design}~and~\ref{as:noise}, provided that $\left\|\btheta-\thetas\right\|_1=O_P(r_{\btheta})$, $r_{\btheta}\sqrt{\log((n+k)d)}\lesssim 1$, and $n+k\gtrsim\log^2(d(n+k))\log d$.
	Putting all the preceding bounds together, we obtain that
	\begin{align*}
        \matrixnorm{\widetilde\Omega-\Omega_0}_{\max} &= O_P\left(  {s^*}\left(\sqrt{\frac{\log d}{n+k}} + \frac{\log^2(d(n+k))\log d}{n+k} + \sqrt{\log((n+k)d)} r_{\btheta} + \frac{nk}{n+k}r_{\btheta}^2\right) + \sqrt{\frac{{s^*}\log d}n}\right),
        \end{align*}
        and
        \begin{align*}
        \matrixnorm{\widetilde\Omega-\widehat\Omega}_{\max} &= O_P\left(  {s^*}\left(\sqrt{\frac{\log d}{n+k}} + \frac{\log^2(d(n+k))\log d}{n+k} + \sqrt{\log((n+k)d)} r_{\btheta} + \frac{nk}{n+k}r_{\btheta}^2\right) + \sqrt{\frac{{s^*}\log d}n}\right).
        \end{align*}
        
        Choosing
	$$u= \left({s^*}\sqrt{\frac{\log d}{n+k}} + \frac{{s^*}\log^2(d(n+k))\log d}{n+k} + {s^*}\sqrt{\log((n+k)d)} r_{\btheta} + \frac{nk{s^*}}{n+k}r_{\btheta}^2 + \sqrt{\frac{{s^*}\log d}n}\right)^{1-\kappa},$$
	with any $\kappa>0$, we deduce that
	$$P\left(\matrixnorm{\widetilde\Omega-\widehat\Omega}_{\max}>u\right)=o(1).$$
	We also have that
	$$u^{1/3}\left(1\vee\log\frac du\right)^{2/3}=o(1),$$
	provided that
	$$ \left({s^*}\sqrt{\frac{\log d}{n+k}} + \frac{{s^*}\log^2(d(n+k))\log d}{n+k} + {s^*}\sqrt{\log((n+k)d)} r_{\btheta} + \frac{nk{s^*}}{n+k}r_{\btheta}^2 + \sqrt{\frac{{s^*}\log d}n}\right) \log^{2+\kappa} d =o(1),$$
	which holds if
	$$n\gg  {s^*}\log^{5+\kappa} d,$$
	$$n+k\gg{s^*}^2\log^{5+\kappa} d,$$
	and
	$$r_{\btheta}\ll\min\left\{\frac1{  {s^*}\sqrt{\log((n+k)d)}\log^{2+\kappa}d},\frac1{ \sqrt{{s^*}}\log^{1+\kappa}d}\sqrt{\frac1n+\frac1k}\right\}.$$
	
\end{proof}

\begin{lemma}\label{lem:gbd0_reg_glm}
	$\overline\Omega$ and $\widehat\Omega$ are defined as in \eqref{eqn:ob_glm} and \eqref{eqn:oh} respectively.  In sparse GLM, under Assumptions~\ref{as:smth_glm}--\ref{as:subexp_glm}, provided that $\left\|\btheta-\thetas\right\|_1=O_P(r_{\btheta})$, $r_{\btheta}\lesssim 1$, $n\gg s_0^2\log^2 d+{s^*}^2\log d$, and $k\gtrsim\log d$, we have that
	$$\matrixnorm{\overline\Omega-\widehat\Omega}_{\max} = O_P\left(  {s^*}\left(\sqrt{\frac{\log d}k} + \sqrt{\log d} r_{\btheta} + n r_{\btheta}^2\right) + \sqrt{\frac{(s_0+{s^*})\log d}n}+\frac{\log^2(dN)\log d}N\right).$$
	Moreover, if $n\gg  (s_0+{s^*})\log^{5+\kappa} d$, $k\gg{s^*}^2\log^{5+\kappa} d$, and
	$$\left\|\btheta-\thetas\right\|_1\ll\min\left\{\frac1{  {s^*}\log^{5/2+\kappa}d},\frac1{ \sqrt{n {s^*}}\log^{1+\kappa}d}\right\},$$
	for some $\kappa>0$, then there exists some $u>0$ such that \eqref{eqn:u} holds.
\end{lemma}

\begin{proof}[Lemma \ref{lem:gbd0_reg_glm}]
	We use the same argument as in the proof of Lemma~\ref{lem:gbd0_reg}.  Note by the triangle inequality that
	$$\matrixnorm{\overline\Omega-\widehat\Omega}_{\max}\leq\matrixnorm{\overline\Omega-\Omega_0}_{\max} + \matrixnorm{\widehat\Omega-\Omega_0}_{\max},$$
	where $\Omega_0$ is defined as in \eqref{eqn:o0}.  First, we bound $\matrixnorm{\widehat\Omega-\Omega_0}_{\max}$.  With Assumption~(E.1) of \citet{chernozhukov2013gaussian} verified for $\nabla^2\cLs(\thetas)^{-1}\nabla\cL(\thetas;Z)$ in the proof of Lemma~\ref{lem:reg0_glm}, by the proof of Corollary 3.1 of \citet{chernozhukov2013gaussian}, we have that
	$$\matrixnorm{\widehat\Omega-\Omega_0}_{\max} = O_P\left(\sqrt{\frac{\log d}N} + \frac{\log^2(dN)\log d}N\right).$$

	Next, we bound $\matrixnorm{\overline\Omega-\Omega_0}_{\max}$.  By the triangle inequality, we have that
	\begin{align*}
	&\matrixnorm{\overline\Omega-\Omega_0}_{\max}\\
	&=\Bigg|\!\Bigg|\!\Bigg|\tT(\ttheta^{(0)})\left(\frac1k\sum_{j=1}^k n\left(\nabla\cL_j(\btheta)-\nabla\cL_N(\btheta)\right)\left(\nabla\cL_j(\btheta)-\nabla\cL_N(\btheta)\right)^\top\right)\tT(\ttheta^{(0)})^\top \\
	&\quad-\Theta\Ee\left[\nabla\cL(\thetas;Z)\nabla\cL(\thetas;Z)^\top\right]\Theta\Bigg|\!\Bigg|\!\Bigg|_{\max} \\
	&\leq \matrixnorm{\tT(\ttheta^{(0)})\left(\frac1k\sum_{j=1}^k n\left(\nabla\cL_j(\btheta)-\nabla\cL_N(\btheta)\right)\left(\nabla\cL_j(\btheta)-\nabla\cL_N(\btheta)\right)^\top-\Ee\left[\nabla\cL(\thetas;Z)\nabla\cL(\thetas;Z)^\top\right]\right)\tT(\ttheta^{(0)})^\top}_{\max} \\
	&\quad + \matrixnorm{\tT(\ttheta^{(0)})\Ee\left[\nabla\cL(\thetas;Z)\nabla\cL(\thetas;Z)^\top\right]\tT(\ttheta^{(0)})^\top -\Theta\Ee\left[\nabla\cL(\thetas;Z)\nabla\cL(\thetas;Z)^\top\right]\Theta}_{\max} \\
	&\defn I_1(\btheta) + I_2.
	\end{align*}
	Note that
	\begin{align*}
	&\tT(\ttheta^{(0)})\Ee\left[\nabla\cL(\thetas;Z)\nabla\cL(\thetas;Z)^\top\right]\tT(\ttheta^{(0)})^\top \\
	&= \left(\tT(\ttheta^{(0)})-\Theta\right)\Ee\left[\nabla\cL(\thetas;Z)\nabla\cL(\thetas;Z)^\top\right]\left(\tT(\ttheta^{(0)})-\Theta\right)^\top + \Theta\Ee\left[\nabla\cL(\thetas;Z)\nabla\cL(\thetas;Z)^\top\right]\left(\tT(\ttheta^{(0)})-\Theta\right)^\top \\
	&\quad + \left(\tT(\ttheta^{(0)})-\Theta\right)\Ee\left[\nabla\cL(\thetas;Z)\nabla\cL(\thetas;Z)^\top\right]\Theta + \Theta\Ee\left[\nabla\cL(\thetas;Z)\nabla\cL(\thetas;Z)^\top\right]\Theta.
	\end{align*}
	By the triangle inequality, we have that
	\begin{align*}
	I_2 &\leq\matrixnorm{\left(\tT(\ttheta^{(0)})-\Theta\right)\Ee\left[\nabla\cL(\thetas;Z)\nabla\cL(\thetas;Z)^\top\right]\left(\tT(\ttheta^{(0)})-\Theta\right)^\top}_{\max} \\
	&\quad + 2\matrixnorm{\Theta\Ee\left[\nabla\cL(\thetas;Z)\nabla\cL(\thetas;Z)^\top\right]\left(\tT(\ttheta^{(0)})-\Theta\right)^\top}_{\max} \\
	&\leq\matrixnorm{\Ee\left[\nabla\cL(\thetas;Z)\nabla\cL(\thetas;Z)^\top\right]}_2 \max_l \left\|\tT(\ttheta^{(0)})_l-\Theta_l\right\|_2^2 \\
	&\quad + 2\matrixnorm{\Ee\left[\nabla\cL(\thetas;Z)\nabla\cL(\thetas;Z)^\top\right]}_2 \max_l \left\|\Theta_l\right\|_2 \max_l \left\|\tT(\ttheta^{(0)})_l-\Theta_l\right\|_2.
	\end{align*}
	Note that $\max_l \left\|\Theta_l\right\|_2\leq\matrixnorm{\Theta}_2=O(1)$ under Assumption~\ref{as:hes_glm}.  By Lemma~\ref{lem:hes_hd_glm}, provided that $n\gg s_0^2\log^2 d+{s^*}^2\log d$, we have that
	\begin{align*}
	I_2&=O_P\left((s_0+  {s^*})\frac{\log d}n + \sqrt{s_0+{s^*}}\sqrt{\frac{\log d}n}\right) =O_P\left( \sqrt{\frac{(s_0+{s^*})\log d}n}\right).
	\end{align*}
	To bound $I_1(\btheta)$, we note that
	\begin{align*}
	I_1(\btheta) &\leq \matrixnorm{\tT(\ttheta^{(0)})}_\infty^2 \matrixnorm{\frac1k\sum_{j=1}^k n\left(\nabla\cL_j(\btheta)-\nabla\cL_N(\btheta)\right)\left(\nabla\cL_j(\btheta)-\nabla\cL_N(\btheta)\right)^\top-\Ee\left[\nabla\cL(\thetas;Z)\nabla\cL(\thetas;Z)^\top\right]}_{\max}.
	\end{align*}
	By Lemma~\ref{lem:hes_hd_glm}, we have that
	$$\matrixnorm{\tT(\ttheta^{(0)})}_\infty=O_P\left( \sqrt{{s^*}}\right).$$
	Then, applying Lemma~\ref{lem:vcov0_reg_glm}, we obtain that
	$$I_1(\btheta) = O_P\left(  {s^*}\left(\sqrt{\frac{\log d}k} + \sqrt{\log d} r_{\btheta} + n r_{\btheta}^2\right)\right),$$
	provided that $\left\|\btheta-\thetas\right\|_1=O_P(r_{\btheta})$, $r_{\btheta}\lesssim 1$, $n\gtrsim\log d$, and $k\gtrsim\log d$.
	Putting all the preceding bounds together, we obtain that
        \begin{align*}
        \matrixnorm{\overline\Omega-\Omega_0}_{\max} &= O_P\left(  {s^*}\left(\sqrt{\frac{\log d}k} + \sqrt{\log d} r_{\btheta} + n r_{\btheta}^2\right) + \sqrt{\frac{(s_0+{s^*})\log d}n}\right),
        \end{align*}
        and
        \begin{align*}
        \matrixnorm{\overline\Omega-\widehat\Omega}_{\max} &= O_P\left(  {s^*}\left(\sqrt{\frac{\log d}k} + \sqrt{\log d} r_{\btheta} + n r_{\btheta}^2\right) + \sqrt{\frac{(s_0+{s^*})\log d}n}+\frac{\log^2(dN)\log d}N\right).
        \end{align*}
        
        Choosing
	$$u= \left({s^*}\sqrt{\frac{\log d}k} + {s^*}\sqrt{\log d} r_{\btheta} + n {s^*} r_{\btheta}^2 + \sqrt{\frac{(s_0+{s^*})\log d}n}+\frac{\log^2(dN)\log d}N\right)^{1-\kappa},$$
	with any $\kappa>0$, we deduce that
	$$P\left(\matrixnorm{\overline\Omega-\widehat\Omega}_{\max}>u\right)=o(1).$$
	We also have that
	$$u^{1/3}\left(1\vee\log\frac du\right)^{2/3}=o(1),$$
	provided that
	$$ \left({s^*}\sqrt{\frac{\log d}k} + {s^*}\sqrt{\log d} r_{\btheta} + n {s^*} r_{\btheta}^2 + \sqrt{\frac{(s_0+{s^*})\log d}n}+\frac{\log^2(dN)\log d}N\right) \log^{2+\kappa} d =o(1),$$
	which holds if
	$$n\gg  (s_0+{s^*})\log^{5+\kappa} d,$$
	$$k\gg{s^*}^2\log^{5+\kappa} d,$$
	and
	$$r_{\btheta}\ll\min\left\{\frac1{  {s^*}\log^{5/2+\kappa}d},\frac1{ \sqrt{n {s^*}}\log^{1+\kappa}d}\right\}.$$
	
\end{proof}

\begin{lemma}\label{lem:gbdo_reg_glm}
	$\widehat\Omega$ and $\Omega_0$ is defined as in \eqref{eqn:oh} and \eqref{eqn:o0} respectively.  In sparse GLM, under Assumptions~\ref{as:hes_glm}--\ref{as:subexp_glm}, we have that
	$$\matrixnorm{\widehat\Omega-\Omega_0}_{\max} = O_P\left(\sqrt{\frac{\log d}N} + \frac{\log^2(dN)\log d}N\right).$$
	Moreover, if $N\gg\log^{5+\kappa}d$ for some $\kappa>0$, then there exists some $v>0$ such that \eqref{eqn:v} holds.
\end{lemma}

\begin{proof}[Lemma \ref{lem:gbdo_reg_glm}]
	In the proof of Lemma~\ref{lem:gbd0_reg_glm}, we have shown that
	$$\matrixnorm{\widehat\Omega-\Omega_0}_{\max} = O_P\left(\sqrt{\frac{\log d}N} + \frac{\log^2(dN)\log d}N\right).$$
        Choosing
	$$v=\left(\sqrt{\frac{\log d}N} + \frac{\log^2(dN)\log d}N\right)^{1-\kappa},$$
	with any $\kappa>0$, we deduce that
	$$P\left(\matrixnorm{\widehat\Omega-\Omega_0}_{\max}>v\right)=o(1).$$
	We also have that
	$$v^{1/3}\left(1\vee\log\frac dv\right)^{2/3}=o(1),$$
	provided that
	$$\left(\sqrt{\frac{\log d}N} + \frac{\log^2(dN)\log d}N\right) \log^{2+\kappa} d=o(1),$$
	which holds if
	$$N\gg\log^{5+\kappa}d.$$
	The same result applies to the low-dimensional case as well.
\end{proof}

\begin{lemma}\label{lem:gbd_reg_glm}
	$\widetilde\Omega$ and $\widehat\Omega$ are defined as in \eqref{eqn:ot_glm_0} and \eqref{eqn:oh} respectively.  In sparse GLM, under Assumptions~\ref{as:smth_glm}--\ref{as:subexp_glm}, provided that $\left\|\btheta-\thetas\right\|_1=O_P(r_{\btheta})$, $r_{\btheta}\lesssim 1$, and $n\gg s_0^2\log^2 d+{s^*}^2\log d$, we have that
	\begin{align*}
	\matrixnorm{\widetilde\Omega-\widehat\Omega}_{\max} &= O_P\Bigg(  {s^*}\left(\sqrt{\frac{\log d}{n+k}} +\frac{n+k\sqrt{\log d}+k^{3/4}\log^{3/4}d}{n+k} r_{\btheta} + \frac{nk}{n+k}r_{\btheta}^2\right) + \sqrt{\frac{(s_0+{s^*})\log d}n} \\
        &\quad+\frac{\log^2(dN)\log d}N\Bigg).
        \end{align*}
	Moreover, if $n\gg  (s_0+{s^*})\log^{5+\kappa} d+s_0^2\log^2 d+{s^*}^2\log d$, $n+k\gg{s^*}^2\log^{5+\kappa} d$, and
	$$\left\|\btheta-\thetas\right\|_1\ll\min\left\{\frac{n+k}{  {s^*}\left(n+k\sqrt{\log d}+k^{3/4}\log^{3/4}d\right)\log^{2+\kappa}d},\frac1{ \sqrt{{s^*}}\log^{1+\kappa}d}\sqrt{\frac1n+\frac1k}\right\},$$
	for some $\kappa>0$, then there exists some $u>0$ such that \eqref{eqn:ut} holds.
\end{lemma}

\begin{proof}[Lemma \ref{lem:gbd_reg_glm}]
	Note by the triangle inequality that
	$$\matrixnorm{\widetilde\Omega-\widehat\Omega}_{\max}\leq\matrixnorm{\widetilde\Omega-\Omega_0}_{\max} + \matrixnorm{\widehat\Omega-\Omega_0}_{\max},$$
	where $\Omega_0$ is defined as in \eqref{eqn:o0}.  By the proof of Lemma~\ref{lem:gbd0_reg_glm}, we have that
	$$\matrixnorm{\widehat\Omega-\Omega_0}_{\max} = O_P\left(\sqrt{\frac{\log d}N} + \frac{\log^2(dN)\log d}N\right).$$
	Next, we bound $\matrixnorm{\widetilde\Omega-\Omega_0}_{\max}$ using the same argument as in the proof of Lemma~\ref{lem:gbd0_reg_glm}.  By the triangle inequality, we have that
	\begin{align*}
	&\matrixnorm{\widetilde\Omega-\Omega_0}_{\max}\\
	&=\Bigg|\!\Bigg|\!\Bigg|\tT(\ttheta^{(0)})\frac1{n+k-1}\Bigg(\sum_{i=1}^n\left(\nabla\cL(\theta;Z_{i1})-\nabla\cL_N(\btheta)\right)\left(\nabla\cL(\theta;Z_{i1})-\nabla\cL_N(\btheta)\right)^\top \\
	&\quad+\sum_{j=2}^k n\left(\nabla\cL_j(\btheta)-\nabla\cL_N(\btheta)\right)\left(\nabla\cL_j(\btheta)-\nabla\cL_N(\btheta)\right)^\top\Bigg)\tT(\ttheta^{(0)})^\top-\Theta\Ee\left[\nabla\cL(\thetas;Z)\nabla\cL(\thetas;Z)^\top\right]\Theta\Bigg|\!\Bigg|\!\Bigg|_{\max} \\
	&\leq \Bigg|\!\Bigg|\!\Bigg|\tT(\ttheta^{(0)})\Bigg(\frac1{n+k-1}\Bigg(\sum_{i=1}^n\left(\nabla\cL(\theta;Z_{i1})-\nabla\cL_N(\btheta)\right)\left(\nabla\cL(\theta;Z_{i1})-\nabla\cL_N(\btheta)\right)^\top \\
	&\quad+\sum_{j=2}^k n\left(\nabla\cL_j(\btheta)-\nabla\cL_N(\btheta)\right)\left(\nabla\cL_j(\btheta)-\nabla\cL_N(\btheta)\right)^\top\Bigg)-\Ee\left[\nabla\cL(\thetas;Z)\nabla\cL(\thetas;Z)^\top\right]\Bigg)\tT(\ttheta^{(0)})^\top\Bigg|\!\Bigg|\!\Bigg|_{\max} \\
	&\quad+ \matrixnorm{\tT(\ttheta^{(0)})\Ee\left[\nabla\cL(\thetas;Z)\nabla\cL(\thetas;Z)^\top\right]\tT(\ttheta^{(0)})^\top-\Theta\Ee\left[\nabla\cL(\thetas;Z)\nabla\cL(\thetas;Z)^\top\right]\Theta}_{\max} \\
	&\defn I_1'(\btheta) + I_2.
	\end{align*}
	We have shown in the proof of Lemma~\ref{lem:gbd0_reg_glm} that
	$$I_2 = O_P\left( \sqrt{\frac{(s_0+{s^*})\log d}n}\right).$$
	To bound $I_1'(\btheta)$, we note that
	\begin{align*}
	I_1'(\btheta) &\leq \matrixnorm{\tT(\ttheta^{(0)})}_\infty^2 \Bigg|\!\Bigg|\!\Bigg|\frac1{n+k-1}\Bigg(\sum_{i=1}^n\left(\nabla\cL(\btheta;Z_{i1})-\nabla\cL_N(\btheta)\right)\left(\nabla\cL(\btheta;Z_{i1})-\nabla\cL_N(\btheta)\right)^\top \\
	&\quad+\sum_{j=2}^k n\left(\nabla\cL_j(\btheta)-\nabla\cL_N(\btheta)\right)\left(\nabla\cL_j(\btheta)-\nabla\cL_N(\btheta)\right)^\top\Bigg) -\Ee\left[\nabla\cL(\thetas;Z)\nabla\cL(\thetas;Z)^\top\right]\Bigg|\!\Bigg|\!\Bigg|_{\max}.
	\end{align*}
	By Lemma~\ref{lem:hes_hd_glm}, provided that $n\gg s_0^2\log^2 d+{s^*}^2\log d$, we have that
	$$\matrixnorm{\tT(\ttheta^{(0)})}_\infty=O_P\left( \sqrt{{s^*}}\right).$$
	Then, applying Lemma~\ref{lem:vcov_reg_glm}, we have that
	\begin{align*}
	I_1'(\btheta) = O_P\left(  {s^*}\left(\sqrt{\frac{\log d}{n+k}} + \frac{n+k\sqrt{\log d}+k^{3/4}\log^{3/4}d}{n+k} r_{\btheta} + \frac{nk}{n+k}r_{\btheta}^2\right)\right),
	\end{align*}
	under Assumptions~\ref{as:smth_glm}--\ref{as:hes_glm}, provided that $\left\|\btheta-\thetas\right\|_1=O_P(r_{\btheta})$, $r_{\btheta}\lesssim 1$, and $n+k\gtrsim\log d$.
	
	Putting all the preceding bounds together, we obtain that
	\begin{align*}
        \matrixnorm{\widetilde\Omega-\Omega_0}_{\max} &= O_P\left(  {s^*}\left(\sqrt{\frac{\log d}{n+k}} +\frac{n+k\sqrt{\log d}+k^{3/4}\log^{3/4}d}{n+k} r_{\btheta} + \frac{nk}{n+k}r_{\btheta}^2\right) + \sqrt{\frac{(s_0+{s^*})\log d}n}\right),
        \end{align*}
        and
        \begin{align*}
        \matrixnorm{\widetilde\Omega-\widehat\Omega}_{\max} &= O_P\Bigg(  {s^*}\left(\sqrt{\frac{\log d}{n+k}} +\frac{n+k\sqrt{\log d}+k^{3/4}\log^{3/4}d}{n+k} r_{\btheta} + \frac{nk}{n+k}r_{\btheta}^2\right) + \sqrt{\frac{(s_0+{s^*})\log d}n} \\
        &\quad+\frac{\log^2(dN)\log d}N\Bigg).
        \end{align*}
        
        Choosing
	$$u= \left({s^*}\sqrt{\frac{\log d}{n+k}} + \frac{n+k\sqrt{\log d}+k^{3/4}\log^{3/4}d}{n+k} {s^*} r_{\btheta} + \frac{nk{s^*}}{n+k}r_{\btheta}^2 + \sqrt{\frac{(s_0+{s^*})\log d}n}+\frac{\log^2(dN)\log d}N\right)^{1-\kappa},$$
	with any $\kappa>0$, we deduce that
	$$P\left(\matrixnorm{\widetilde\Omega-\widehat\Omega}_{\max}>u\right)=o(1).$$
	We also have that
	$$u^{1/3}\left(1\vee\log\frac du\right)^{2/3}=o(1),$$
	provided that
	\begin{align*}
	& \left({s^*}\sqrt{\frac{\log d}{n+k}} + \frac{n+k\sqrt{\log d}+k^{3/4}\log^{3/4}d}{n+k} {s^*} r_{\btheta} + \frac{nk{s^*}}{n+k}r_{\btheta}^2 + \sqrt{\frac{(s_0+{s^*})\log d}n}+\frac{\log^2(dN)\log d}N\right) \log^{2+\kappa} d \\
	&=o(1),
	\end{align*}
	which holds if
	$$n\gg  (s_0+{s^*})\log^{5+\kappa} d+s_0^2\log^2 d+{s^*}^2\log d,$$
	$$n+k\gg{s^*}^2\log^{5+\kappa} d,$$
	and
	$$r_{\btheta}\ll\min\left\{\frac{n+k}{  {s^*}\left(n+k\sqrt{\log d}+k^{3/4}\log^{3/4}d\right)\log^{2+\kappa}d},\frac1{ \sqrt{{s^*}}\log^{1+\kappa}d}\sqrt{\frac1n+\frac1k}\right\}.$$
	
\end{proof}

\begin{lemma}\label{lem:kgrad_decomp}
	For any $\theta$, we have that
	\begin{align*}
	&\matrixnorm{\frac1k\sum_{j=1}^k n\left(\nabla\cL_j(\theta)-\nabla\cL_N(\theta)\right)\left(\nabla\cL_j(\theta)-\nabla\cL_N(\theta)\right)^\top-\Ee\left[\nabla\cL(\thetas;Z)\nabla\cL(\thetas;Z)^\top\right]}_{\max} \leq U_1(\theta) + U_2 + U_3(\theta),
	\end{align*}
	where
	$$U_1(\theta)\defn\matrixnorm{\frac1k\sum_{j=1}^k n\left(\nabla\cL_j(\theta)-\nabla\cLs(\theta)\right)\left(\nabla\cL_j(\theta)-\nabla\cLs(\theta)\right)^\top-n\nabla\cL_j(\thetas)\nabla\cL_j(\thetas)^\top}_{\max},$$
	$$U_2\defn\matrixnorm{\frac1k\sum_{j=1}^k n\nabla\cL_j(\thetas)\nabla\cL_j(\thetas)^\top-\Ee\left[\nabla\cL(\thetas;Z)\nabla\cL(\thetas;Z)^\top\right]}_{\max},$$
	and
	$$U_3(\theta)\defn n\left\|\nabla\cL_N(\theta)-\nabla\cLs(\theta)\right\|_\infty^2.$$
\end{lemma}

Lemma~\ref{lem:kgrad_decomp} is the same as Lemma~F.1 of \citet{ycc2020simultaneous}. We omit the proof.

\begin{lemma}\label{lem:vcov0_reg}
	In sparse linear model, under Assumptions~\ref{as:design}~and~\ref{as:noise}, provided that $\left\|\btheta-\thetas\right\|_1=O_P(r_{\btheta})$, we have that
	\begin{align*}
	&\matrixnorm{\frac1k\sum_{j=1}^k n\left(\nabla\cL_j(\btheta)-\nabla\cL_N(\btheta)\right)\left(\nabla\cL_j(\btheta)-\nabla\cL_N(\btheta)\right)^\top-\Ee\left[\nabla\cL(\thetas;Z)\nabla\cL(\thetas;Z)^\top\right]}_{\max} \\
	&= O_P\Bigg(\sqrt{\frac{\log d}k} + \frac{\log^2(dk)\log d}k + \Bigg(1+\left(\frac{\log d}k\right)^{1/4} + \sqrt{\frac{\log^2(dk)\log d}k}\Bigg)\sqrt{\log(kd)} r_{\btheta} \\
	&\quad + \left(n + \sqrt{\frac{n\log d}k} + \log(kd)\right) r_{\btheta}^2\Bigg).
	\end{align*}
\end{lemma}

\begin{proof}[Lemma \ref{lem:vcov0_reg}]
	By Lemma~\ref{lem:kgrad_decomp}, it suffices to bound $U_1(\btheta)$, $U_2$, and $U_3(\btheta)$.  We begin by bounding $U_2$.  In linear model, we have that
	\begin{align*}
	U_2&=\matrixnorm{\frac1k\sum_{j=1}^k n \left(\frac{X_j^\top e_j}n\right)\left(\frac{X_j^\top e_j}n\right)^\top-\sigma^2\Sigma}_{\max}.
	\end{align*}
	Note that
	$$\Ee\left[\left(\frac{(X_j^\top e_j)_l}{\sqrt n}\right)^2\right]=\Ee\left[\frac{\sum_{i=1}^n X_{ij,l}^2 e_{ij}^2}n\right]=\sigma^2\Sigma_{l,l}$$
	is bounded away from zero, under Assumptions~\ref{as:design}~and~\ref{as:noise}.  Also, using same argument for obtaining \eqref{eqn:bern2}, we have that for any $t>0$,
	$$P\left(\left|\frac{(X_j^\top e_j)_l}n\right| > t\right)\leq 2\exp\left(-cn\left(\frac{t^2}{\Sigma_{l,l}\sigma^2}\wedge\frac t{\sqrt{\Sigma_{l,l}}\sigma}\right)\right),$$
	and then,
	$$P\left(\left|\frac{(X_j^\top e_j)_l}{\sqrt n}\right| > t\right)\leq 2\exp\left(-c\left(\frac{t^2}{\Sigma_{l,l}\sigma^2}\wedge\frac{t\sqrt n}{\sqrt{\Sigma_{l,l}}\sigma}\right)\right)\leq C\exp\left(-c't\right),$$
	for some positive constants $c$, $c'$, and $C$, that is, $(X_j^\top e_j)_l/\sqrt n$ is sub-exponential with O(1) $\psi_1$-norm for each $(j,l)$.  Then, by the proof of Corollary 3.1 of \citet{chernozhukov2013gaussian}, we have that
	\begin{align*}
	\Ee[U_2]&=\Ee\left[\matrixnorm{\frac1k\sum_{j=1}^k \left(\frac{X_j^\top e_j}{\sqrt n}\right)\left(\frac{X_j^\top e_j}{\sqrt n}\right)^\top-\sigma^2\Sigma}_{\max}\right] \lesssim \sqrt{\frac{\log d}k} + \frac{\log^2(dk)\log d}k,
	\end{align*}
	and then, for any $\delta\in(0,1)$, with probability at least $1-\delta$,
	\begin{align*}
	U_2&\lesssim\frac1\delta\left( \sqrt{\frac{\log d}k} + \frac{\log^2(dk)\log d}k\right),
	\end{align*}
	by Markov's inequality, which implies that
	$$U_2 = O_P\left(\sqrt{\frac{\log d}k} + \frac{\log^2(dk)\log d}k\right).$$
	
	Next, we bound $U_3(\btheta)$.  By the triangle inequality and the fact that for any matrix $A$ and vector $a$ with compatible dimensions, $\|Aa\|_\infty\leq\matrixnorm{A}_{\max}\|a\|_1$, we have that
	\begin{align*}
	\left\|\nabla\cL_N(\btheta)-\nabla\cLs(\btheta)\right\|_\infty
	&\leq \left\|\nabla\cL_N(\btheta)-\nabla\cL_N(\thetas)\right\|_\infty + \left\|\nabla\cL_N(\thetas)\right\|_\infty + \left\|\nabla\cLs(\btheta)\right\|_\infty \\
	&= \left\|\frac{X_N^\top(X_N\btheta-y_N)}N-\frac{X_N^\top(X_N\thetas-y_N)}N\right\|_\infty + \left\|\frac{X_N^\top(X_N\thetas-y_N)}N\right\|_\infty + \left\|\Sigma(\btheta-\thetas)\right\|_\infty \\
	&= \left\|\frac{X_N^\top X_N}N(\btheta-\thetas)\right\|_\infty + \left\|\frac{X_N^\top e_N}N\right\|_\infty + \left\|\Sigma(\btheta-\thetas)\right\|_\infty \\
	&\leq \matrixnorm{\frac{X_N^\top X_N}N}_{\max} \left\|\btheta-\thetas\right\|_1 + \left\|\frac{X_N^\top e_N}N\right\|_\infty + \matrixnorm{\Sigma}_{\max} \left\|\btheta-\thetas\right\|_1 \\
	&\lesssim \matrixnorm{\frac{X_N^\top X_N}N-\Sigma}_{\max} \left\|\btheta-\thetas\right\|_1 + \left\|\frac{X_N^\top e_N}N\right\|_\infty + \matrixnorm{\Sigma}_{\max} \left\|\btheta-\thetas\right\|_1.
	\end{align*}
	By \eqref{eqn:bern1} and \eqref{eqn:bern2}, we have that
	$$\matrixnorm{\frac{X_N^\top X_N}N-\Sigma}_{\max} \leq \matrixnorm{\Sigma}_{\max}\left(\frac{\log\frac{2d^2}\delta}{cN}\vee\sqrt{\frac{\log\frac{2d^2}\delta}{cN}}\right)=O_P\left(\sqrt{\frac{\log d}N}\right),$$
	and
	$$\left\|\frac{X_N^\top e_N}N\right\|_\infty \leq \max_l \sqrt{\Sigma_{l,l}}\sigma\left(\frac{\log\frac{2d}\delta}{cN}\vee\sqrt{\frac{\log\frac{2d}\delta}{cN}}\right)=O_P\left(\sqrt{\frac{\log d}N}\right),$$
	where $\max_l \sqrt{\Sigma_{l,l}}\leq\matrixnorm{\Sigma}_{\max}=O(1)$ under Assumption~\ref{as:design}.  Then, assuming that $\left\|\btheta-\thetas\right\|_1=O_P(r_{\btheta})$, we have that
	\begin{align*}
	\left\|\nabla\cL_N(\btheta)-\nabla\cLs(\btheta)\right\|_\infty
	&= \left(O(1)+O_P\left(\sqrt{\frac{\log d}N}\right)\right) O_P(r_{\btheta}) + O_P\left(\sqrt{\frac{\log d}N}\right) \\
	&= O_P\left(\left(1+\sqrt{\frac{\log d}N}\right) r_{\btheta}+\sqrt{\frac{\log d}N}\right),
	\end{align*}
	and then,
	$$U_3(\btheta) =  O_P\left(\left(1+\sqrt{\frac{\log d}N}\right) n r_{\btheta}^2 + \frac{\log d}k\right).$$
	
	Lastly, we bound $U_1(\btheta)$.  We write $\nabla\cL_j(\btheta)-\nabla\cLs(\btheta)$ as $\left(\nabla\cL_j(\btheta)-\nabla\cLs(\btheta)-\nabla\cL_j(\thetas)\right)+\nabla\cL_j(\thetas)$, and obtain by the triangle inequality that
        \begin{align*}
        U_1(\btheta) &\leq \matrixnorm{\frac1k\sum_{j=1}^k n\left(\nabla\cL_j(\btheta)-\nabla\cLs(\btheta)-\nabla\cL_j(\thetas)\right)\left(\nabla\cL_j(\btheta)-\nabla\cLs(\btheta)-\nabla\cL_j(\thetas)\right)^\top}_{\max} \\
        &\quad+ \matrixnorm{\frac1k\sum_{j=1}^k n\nabla\cL_j(\thetas)\left(\nabla\cL_j(\btheta)-\nabla\cLs(\btheta)-\nabla\cL_j(\thetas)\right)^\top}_{\max} \\
        &\quad+ \matrixnorm{\frac1k\sum_{j=1}^k n\left(\nabla\cL_j(\btheta)-\nabla\cLs(\btheta)-\nabla\cL_j(\thetas)\right)\nabla\cL_j(\thetas)^\top}_{\max} \\
        &= \matrixnorm{\frac1k\sum_{j=1}^k n\left(\nabla\cL_j(\btheta)-\nabla\cLs(\btheta)-\nabla\cL_j(\thetas)\right)\left(\nabla\cL_j(\btheta)-\nabla\cLs(\btheta)-\nabla\cL_j(\thetas)\right)^\top}_{\max} \\
        &\quad+ 2\matrixnorm{\frac1k\sum_{j=1}^k n\nabla\cL_j(\thetas)\left(\nabla\cL_j(\btheta)-\nabla\cLs(\btheta)-\nabla\cL_j(\thetas)\right)^\top}_{\max} \\
        &\defn U_{11}(\btheta) + 2 U_{12}(\btheta).
        \end{align*}
        To bound $U_{12}(\btheta)$, we first define an inner product $\langle A,B\rangle=\matrixnorm{AB^\top}_{\max}$ for any $A,B\in\R^{d\times k}$, the validity of which is easy to check.  We then apply Cauchy-Schwarz inequality on $\langle A,B\rangle$ with
        $$A=\sqrt{\frac nk}
        \begin{bmatrix}
	\nabla\cL_1(\thetas) & \dots & \nabla\cL_k(\thetas)
	\end{bmatrix}$$
        and
        $$B=\sqrt{\frac nk}
        \begin{bmatrix}
	\nabla\cL_1(\btheta)-\nabla\cLs(\btheta)-\nabla\cL_1(\thetas) & \dots & \nabla\cL_k(\btheta)-\nabla\cLs(\btheta)-\nabla\cL_k(\thetas))
	\end{bmatrix}$$
        and obtain that
         \begin{align*}
        U_{12}(\btheta) &\leq \matrixnorm{\frac1k\sum_{j=1}^k n\nabla\cL_j(\thetas)\nabla\cL_j(\thetas)^\top}_{\max}^{1/2} \\
        &\quad \cdot\matrixnorm{\frac1k\sum_{j=1}^k n\left(\nabla\cL_j(\btheta)-\nabla\cLs(\btheta)-\nabla\cL_j(\thetas)\right)\left(\nabla\cL_j(\btheta)-\nabla\cLs(\btheta)-\nabla\cL_j(\thetas)\right)^\top}_{\max}^{1/2} \\
        &= \matrixnorm{\frac1k\sum_{j=1}^k n\nabla\cL_j(\thetas)\nabla\cL_j(\thetas)^\top}_{\max}^{1/2} U_{11}(\btheta)^{1/2}.
        \end{align*}
        By the triangle inequality, we have that
        \begin{align*}
        &\matrixnorm{\frac1k\sum_{j=1}^k n\nabla\cL_j(\thetas)\nabla\cL_j(\thetas)^\top}_{\max} \\
        &\leq \matrixnorm{\frac1k\sum_{j=1}^k n\nabla\cL_j(\thetas)\nabla\cL_j(\thetas)^\top-\Ee\left[\nabla\cL(\thetas;Z)\nabla\cL(\thetas;Z)^\top\right]}_{\max} + \matrixnorm{\Ee\left[\nabla\cL(\thetas;Z)\nabla\cL(\thetas;Z)^\top\right]}_{\max} \\
        &= U_2 + \sigma^2\matrixnorm{\Sigma}_{\max} \\
        &= O_P\left(1+\sqrt{\frac{\log d}k} + \frac{\log^2(dk)\log d}k\right).
        \end{align*}
        It remains to bound $U_{11}(\btheta)$.  Note that
        \begin{align*}
        \nabla\cL_j(\btheta)-\nabla\cLs(\btheta)-\nabla\cL_j(\thetas)
        &= \frac{X_j^\top(X_j\btheta-y_j)}n - \Sigma(\btheta-\thetas) + \frac{X_j^\top(X_j\thetas-y_j)}n \\
        &= \left(\frac{X_j^\top X_j}n-\Sigma\right)(\btheta-\thetas).
        \end{align*}
        Then, we have that
        \begin{align*}
        U_{11}(\btheta) &= \matrixnorm{\frac1k\sum_{j=1}^k n\left(\frac{X_j^\top X_j}n-\Sigma\right)(\btheta-\thetas)(\btheta-\thetas)^\top\left(\frac{X_j^\top X_j}n-\Sigma\right)}_{\max} \\
        &\leq \frac1k\sum_{j=1}^k n\matrixnorm{\left(\frac{X_j^\top X_j}n-\Sigma\right)(\btheta-\thetas)(\btheta-\thetas)^\top\left(\frac{X_j^\top X_j}n-\Sigma\right)}_{\max} \\
        &= \frac1k\sum_{j=1}^k n\matrixnorm{\left(\frac{X_j^\top X_j}n-\Sigma\right)(\btheta-\thetas)}_\infty^2 \\
        &\leq \frac1k\sum_{j=1}^k n\matrixnorm{\frac{X_j^\top X_j}n-\Sigma}_{\max}^2 \left\|\btheta-\thetas\right\|_1^2,
        \end{align*}
        where we use the triangle inequality and the fact that $\matrixnorm{aa^\top}_{\max}=\|a\|_\infty^2$ for any vector $a$, and $\|Aa\|_\infty\leq\matrixnorm{A}_{\max}\|a\|_1$ for any matrix $A$ and vector $a$ with compatible dimensions.  By \eqref{eqn:bern3}, we have that
        $$P\left(\matrixnorm{\frac{X_j^\top X_j}n-\Sigma}_{\max} > \matrixnorm{\Sigma}_{\max}\left(\frac{\log\frac{2kd^2}\delta}{cn}\vee\sqrt{\frac{\log\frac{2kd^2}\delta}{cn}}\right)\right)\leq \frac\delta k,$$
        and then, by the union bound,
        $$P\left(\max_j\matrixnorm{\frac{X_j^\top X_j}n-\Sigma}_{\max} > \matrixnorm{\Sigma}_{\max}\left(\frac{\log\frac{2kd^2}\delta}{cn}\vee\sqrt{\frac{\log\frac{2kd^2}\delta}{cn}}\right)\right)\leq \delta,$$
        which implies that
        $$\max_j\matrixnorm{\frac{X_j^\top X_j}n-\Sigma}_{\max} = O_P\left(\sqrt{\frac{\log(kd)}n}\right).$$
        Putting all the preceding bounds together, we obtain that
        $$U_{11}(\btheta) = O_P\left(\log(kd) r_{\btheta}^2\right),$$
        $$U_{12}(\btheta) = O_P\left(\left(1+\left(\frac{\log d}k\right)^{1/4} + \sqrt{\frac{\log^2(dk)\log d}k}\right)\sqrt{\log(kd)} r_{\btheta}\right),$$
        $$U_1(\btheta) = O_P\left(\left(1+\left(\frac{\log d}k\right)^{1/4} + \sqrt{\frac{\log^2(dk)\log d}k}\right)\sqrt{\log(kd)} r_{\btheta} + \log(kd) r_{\btheta}^2\right),$$
        and finally,
        \begin{align*}
	&\matrixnorm{\frac1k\sum_{j=1}^k n\left(\nabla\cL_j(\btheta)-\nabla\cL_N(\btheta)\right)\left(\nabla\cL_j(\btheta)-\nabla\cL_N(\btheta)\right)^\top-\Ee\left[\nabla\cL(\thetas;Z)\nabla\cL(\thetas;Z)^\top\right]}_2 \\
	&=O_P\Bigg(\sqrt{\frac{\log d}k} + \frac{\log^2(dk)\log d}k + \Bigg(1+\left(\frac{\log d}k\right)^{1/4} + \sqrt{\frac{\log^2(dk)\log d}k}\Bigg)\sqrt{\log(kd)} r_{\btheta} \\
	&\quad + \left(n + \sqrt{\frac{n\log d}k} + \log(kd)\right) r_{\btheta}^2\Bigg).
	\end{align*}
\end{proof}

\begin{lemma}\label{lem:nk1grad_decomp}
	For any $\theta$, we have that
	\begin{align*}
	&\Bigg|\!\Bigg|\!\Bigg|\frac1{n+k-1}\Bigg(\sum_{i=1}^n\left(\nabla\cL(\theta;Z_{i1})-\nabla\cL_N(\theta)\right)\left(\nabla\cL(\theta;Z_{i1})-\nabla\cL_N(\theta)\right)^\top \\
	&\quad+\sum_{j=2}^k n\left(\nabla\cL_j(\theta)-\nabla\cL_N(\theta)\right)\left(\nabla\cL_j(\theta)-\nabla\cL_N(\theta)\right)^\top\Bigg) -\Ee\left[\nabla\cL(\thetas;Z)\nabla\cL(\thetas;Z)^\top\right]\Bigg|\!\Bigg|\!\Bigg|_{\max} \\
	&\leq V_1(\theta) + V_1'(\theta) + V_2 + V_2' + V_3(\theta),
	\end{align*}
	where
	$$V_1(\theta)\defn\frac{k-1}{n+k-1}\matrixnorm{\frac1{k-1}\sum_{j=2}^k n\left(\nabla\cL_j(\theta)-\nabla\cLs(\theta)\right)\left(\nabla\cL_j(\theta)-\nabla\cLs(\theta)\right)^\top-n\nabla\cL_j(\thetas)\nabla\cL_j(\thetas)^\top}_{\max},$$
	$$V_1'(\theta)\defn\frac n{n+k-1}\matrixnorm{\frac1n\sum_{i=1}^n\left(\nabla\cL(\theta;Z_{i1})-\nabla\cLs(\theta)\right)\left(\nabla\cL(\theta;Z_{i1})-\nabla\cLs(\theta)\right)^\top-\nabla\cL(\thetas;Z_{i1})\nabla\cL(\thetas;Z_{i1})^\top}_{\max},$$
	$$V_2\defn\frac{k-1}{n+k-1}\matrixnorm{\frac1{k-1}\sum_{j=2}^k n\nabla\cL_j(\thetas)\nabla\cL_j(\thetas)^\top-\Ee\left[\nabla\cL(\thetas;Z)\nabla\cL(\thetas;Z)^\top\right]}_{\max},$$
	$$V_2'\defn\frac n{n+k-1}\matrixnorm{\frac1n\sum_{i=1}^n\nabla\cL(\thetas;Z_{i1})\nabla\cL(\thetas;Z_{i1})^\top-\Ee\left[\nabla\cL(\thetas;Z)\nabla\cL(\thetas;Z)^\top\right]}_{\max},$$
	and
	$$V_3(\theta)\defn \frac{nk}{n+k-1} \left\|\nabla\cL_N(\theta)-\nabla\cLs(\theta)\right\|_\infty^2.$$
\end{lemma}

Lemma~\ref{lem:nk1grad_decomp} is the same as Lemma~F.3 of \citet{ycc2020simultaneous}. We omit the proof.

\begin{lemma}\label{lem:vcov_reg}
	In sparse linear model, under Assumptions~\ref{as:design}~and~\ref{as:noise}, provided that $\left\|\btheta-\thetas\right\|_1=O_P(r_{\btheta})$, we have that
	\begin{align*}
	&\Bigg|\!\Bigg|\!\Bigg|\frac1{n+k-1}\Bigg(\sum_{i=1}^n\left(\nabla\cL(\btheta;Z_{i1})-\nabla\cL_N(\btheta)\right)\left(\nabla\cL(\btheta;Z_{i1})-\nabla\cL_N(\btheta)\right)^\top \\
	&\quad+\sum_{j=2}^k n\left(\nabla\cL_j(\btheta)-\nabla\cL_N(\btheta)\right)\left(\nabla\cL_j(\btheta)-\nabla\cL_N(\btheta)\right)^\top\Bigg) -\Ee\left[\nabla\cL(\thetas;Z)\nabla\cL(\thetas;Z)^\top\right]\Bigg|\!\Bigg|\!\Bigg|_{\max} \\
	&= O_P\Bigg(\sqrt{\frac{\log d}{n+k}} + \frac{\log^2(d(n+k))\log d}{n+k} + \left(\left(1+\sqrt{\frac{\log d}N}\right)\frac{nk}{n+k}+\log((n+k)d)\right)r_{\btheta}^2 \\
	&\quad + \Bigg(\sqrt{\log((n+k)d)} + \frac{\log^{1/4} d\sqrt{\log((n+k)d)}}{(n+k)^{1/4}} + \sqrt{\frac{\log^3(d(n+k))\log d}{n+k}}\Bigg) r_{\btheta}\Bigg).
	\end{align*}
\end{lemma}

\begin{proof}[Lemma \ref{lem:vcov_reg}]
	By Lemma~\ref{lem:nk1grad_decomp}, it suffices to bound $V_1(\btheta)$, $V_1'(\btheta)$, $V_2$, $V_2'$, and $V_3(\btheta)$.  By the proof of Lemma~\ref{lem:vcov0_reg}, we have that under Assumptions~\ref{as:design}~and~\ref{as:noise}, assuming that $\left\|\btheta-\thetas\right\|_1=O_P(r_{\btheta})$,
	\begin{align*}
	V_1(\btheta) &= \frac{k-1}{n+k-1} O_P\left(\left(1+\left(\frac{\log d}k\right)^{1/4} + \sqrt{\frac{\log^2(dk)\log d}k}\right)\sqrt{\log(kd)} r_{\btheta} + \log(kd) r_{\btheta}^2\right) \\
	&= O_P\left(\left(1+\left(\frac{\log d}k\right)^{1/4} + \sqrt{\frac{\log^2(dk)\log d}k}\right)\frac{k\sqrt{\log(kd)}}{n+k} r_{\btheta} + \frac{k\log(kd)}{n+k} r_{\btheta}^2\right),
	\end{align*}
	$$V_2 = \frac{k-1}{n+k-1} O_P\left(\sqrt{\frac{\log d}k} + \frac{\log^2(dk)\log d}k\right) = O_P\left(\frac{\sqrt{k\log d}}{n+k} + \frac{\log^2(dk)\log d}{n+k}\right),$$
	and
	$$V_3(\btheta) =  \frac{nk}{n+k-1} O_P\left(\left(1+\sqrt{\frac{\log d}N}\right) r_{\btheta}^2 + \frac{\log d}N\right) = O_P\left(\left(1+\sqrt{\frac{\log d}N}\right)\frac{nk}{n+k}r_{\btheta}^2 + \frac{\log d}{n+k}\right).$$
	It remains to bound $V_1'(\btheta)$ and $V_2'$.
	
	To bound $V_2'$, we have that in linear model, under Assumptions~\ref{as:design}~and~\ref{as:noise},
	\begin{align*}
	V_2'&=\frac n{n+k-1}\matrixnorm{\frac1n\sum_{i=1}^n\left(x_{i1}e_{i1}\right)\left(x_{i1}e_{i1}\right)^\top-\sigma^2\Sigma}_{\max}.
	\end{align*}
	Note that
	$$\Ee\left[\left(x_{i1}e_{i1}\right)_l^2\right]=\sigma^2\Sigma_{l,l}$$
	is bounded away from zero, and also, $\left(x_{i1}e_{i1}\right)_l$ is sub-exponential with O(1) $\psi_1$-norm for each $(i,l)$.  Then, by the proof of Corollary 3.1 of \citet{chernozhukov2013gaussian}, we have that
	\begin{align*}
	\Ee\left[\matrixnorm{\frac1n\sum_{i=1}^n\left(x_{i1}e_{i1}\right)\left(x_{i1}e_{i1}\right)^\top-\sigma^2\Sigma}_{\max}\right] \lesssim \sqrt{\frac{\log d}n} + \frac{\log^2(dn)\log d}n,
	\end{align*}
	and then, for any $\delta\in(0,1)$, with probability at least $1-\delta$,
	\begin{align*}
	\matrixnorm{\frac1n\sum_{i=1}^n\left(x_{i1}e_{i1}\right)\left(x_{i1}e_{i1}\right)^\top-\sigma^2\Sigma}_{\max}&\lesssim\frac1\delta\left(\sqrt{\frac{\log d}n} + \frac{\log^2(dn)\log d}n\right),
	\end{align*}
	by Markov's inequality, which implies that
	$$V_2' = \frac n{n+k-1} O_P\left(\sqrt{\frac{\log d}n} + \frac{\log^2(dn)\log d}n\right) = O_P\left(\frac{\sqrt{n\log d}}{n+k} + \frac{\log^2(dn)\log d}{n+k}\right).$$
	
	Lastly, we bound $V_1'(\btheta)$ using the same argument as in bounding $U_1(\btheta)$ in the proof of Lemma~\ref{lem:vcov0_reg}.  We write $\nabla\cL(\theta;Z_{i1})-\nabla\cLs(\theta)$ as $\left(\nabla\cL(\theta;Z_{i1})-\nabla\cLs(\theta)-\nabla\cL(\thetas;Z_{i1})\right)+\nabla\cL(\thetas;Z_{i1})$, and obtain by the triangle inequality that
        \begin{align*}
        \frac{n+k-1}n V_1'(\btheta) &\leq \matrixnorm{\frac1n\sum_{i=1}^n \left(\nabla\cL(\theta;Z_{i1})-\nabla\cLs(\theta)-\nabla\cL(\thetas;Z_{i1})\right)\left(\nabla\cL(\theta;Z_{i1})-\nabla\cLs(\theta)-\nabla\cL(\thetas;Z_{i1})\right)^\top}_{\max} \\
        &\quad+ \matrixnorm{\frac1n\sum_{i=1}^n \nabla\cL(\thetas;Z_{i1})\left(\nabla\cL(\theta;Z_{i1})-\nabla\cLs(\theta)-\nabla\cL(\thetas;Z_{i1})\right)^\top}_{\max} \\
        &\quad+ \matrixnorm{\frac1n\sum_{i=1}^n \left(\nabla\cL(\theta;Z_{i1})-\nabla\cLs(\theta)-\nabla\cL(\thetas;Z_{i1})\right)\nabla\cL(\thetas;Z_{i1})^\top}_{\max} \\
        &= \matrixnorm{\frac1n\sum_{i=1}^n \left(\nabla\cL(\theta;Z_{i1})-\nabla\cLs(\theta)-\nabla\cL(\thetas;Z_{i1})\right)\left(\nabla\cL(\theta;Z_{i1})-\nabla\cLs(\theta)-\nabla\cL(\thetas;Z_{i1})\right)^\top}_{\max} \\
        &\quad+ 2\matrixnorm{\frac1n\sum_{i=1}^n \nabla\cL(\thetas;Z_{i1})\left(\nabla\cL(\theta;Z_{i1})-\nabla\cLs(\theta)-\nabla\cL(\thetas;Z_{i1})\right)^\top}_{\max} \\
        &\defn V_{11}'(\btheta) + 2 V_{12}'(\btheta).
        \end{align*}
        Applying Cauchy-Schwarz inequality, we obtain that
        \begin{align*}
        V_{12}'(\btheta) &\leq \matrixnorm{\frac1n\sum_{i=1}^n \nabla\cL(\thetas;Z_{i1})\nabla\cL(\thetas;Z_{i1})^\top}_{\max}^{1/2} \\
        &\quad \cdot\matrixnorm{\frac1n\sum_{i=1}^n \left(\nabla\cL(\theta;Z_{i1})-\nabla\cLs(\theta)-\nabla\cL(\thetas;Z_{i1})\right)\left(\nabla\cL(\theta;Z_{i1})-\nabla\cLs(\theta)-\nabla\cL(\thetas;Z_{i1})\right)^\top}_{\max}^{1/2} \\
        &= \matrixnorm{\frac1n\sum_{i=1}^n \nabla\cL(\thetas;Z_{i1})\nabla\cL(\thetas;Z_{i1})^\top}_{\max}^{1/2}V_{11}'(\btheta)^{1/2}.
        \end{align*}
        By the triangle inequality, we have that
        \begin{align*}
        &\matrixnorm{\frac1n\sum_{i=1}^n \nabla\cL(\thetas;Z_{i1})\nabla\cL(\thetas;Z_{i1})^\top}_{\max} \\
        &\leq \matrixnorm{\frac1n\sum_{i=1}^n \nabla\cL(\thetas;Z_{i1})\nabla\cL(\thetas;Z_{i1})^\top-\Ee\left[\nabla\cL(\thetas;Z)\nabla\cL(\thetas;Z)^\top\right]}_{\max} + \matrixnorm{\Ee\left[\nabla\cL(\thetas;Z)\nabla\cL(\thetas;Z)^\top\right]}_{\max} \\
        &= \frac{n+k-1}n V_2' + \sigma^2\matrixnorm{\Sigma}_{\max} \\
        &= O_P\left(1+\sqrt{\frac{\log d}n} + \frac{\log^2(dn)\log d}n\right).
        \end{align*}
        It remains to bound $V_{11}'(\btheta)$.  Note that
        \begin{align*}
        \nabla\cL(\theta;Z_{i1})-\nabla\cLs(\theta)-\nabla\cL(\thetas;Z_{i1})
        &= x_{ij}(x_{ij}^\top\btheta-y_{ij}) - \Sigma(\btheta-\thetas) + x_{ij}(x_{ij}^\top\thetas-y_{ij}) \\
        &= \left(x_{ij} x_{ij}^\top-\Sigma\right)(\btheta-\thetas).
        \end{align*}
        Then, we have by the triangle inequality that
        \begin{align*}
        V_{11}'(\btheta) &= \matrixnorm{\frac1n\sum_{i=1}^n \left(x_{i1} x_{i1}^\top-\Sigma\right)(\btheta-\thetas)(\btheta-\thetas)^\top\left(x_{i1} x_{i1}^\top-\Sigma\right)}_{\max} \\
        &\leq \frac1n\sum_{i=1}^n \matrixnorm{\left(x_{i1} x_{i1}^\top-\Sigma\right)(\btheta-\thetas)(\btheta-\thetas)^\top\left(x_{i1} x_{i1}^\top-\Sigma\right)}_{\max} \\
        &= \frac1n\sum_{i=1}^n \matrixnorm{\left(x_{i1} x_{i1}^\top-\Sigma\right)(\btheta-\thetas)}_\infty^2 \\
        &\leq \frac1n\sum_{i=1}^n \matrixnorm{x_{i1} x_{i1}^\top-\Sigma}_{\max}^2 \left\|\btheta-\thetas\right\|_1^2.
        \end{align*}
        Similarly to obtaining \eqref{eqn:bern3}, we have that
        $$P\left(\matrixnorm{x_{i1} x_{i1}^\top-\Sigma}_{\max} > \matrixnorm{\Sigma}_{\max}\left(\frac{\log\frac{2nd^2}\delta}c\vee\sqrt{\frac{\log\frac{2nd^2}\delta}c}\right)\right)\leq \frac\delta n,$$
        and then, by the union bound,
        $$P\left(\max_i\matrixnorm{x_{i1} x_{i1}^\top-\Sigma}_{\max} > \matrixnorm{\Sigma}_{\max}\left(\frac{\log\frac{2nd^2}\delta}c\vee\sqrt{\frac{\log\frac{2nd^2}\delta}c}\right)\right)\leq \delta,$$
        which implies that
        $$\max_i\matrixnorm{x_{i1} x_{i1}^\top-\Sigma}_{\max} = O_P\left(\sqrt{\log(nd)}\right).$$
        Putting all the preceding bounds together, we obtain that
        $$V_{11}'(\btheta) = O_P\left(\log(nd) r_{\btheta}^2\right),$$
        $$V_{12}'(\btheta) = O_P\left(\left(1+\left(\frac{\log d}n\right)^{1/4} + \sqrt{\frac{\log^2(dn)\log d}n}\right)\sqrt{\log(nd)} r_{\btheta}\right),$$
        \begin{align*}
        V_1'(\btheta) &= \frac n{n+k-1} O_P\left(\left(1+\left(\frac{\log d}n\right)^{1/4} + \sqrt{\frac{\log^2(dn)\log d}n}\right)\sqrt{\log(nd)} r_{\btheta} + \log(nd) r_{\btheta}^2\right) \\
        &= O_P\left(\left(1+\left(\frac{\log d}n\right)^{1/4} + \sqrt{\frac{\log^2(dn)\log d}n}\right)\frac{n\sqrt{\log(nd)}}{n+k} r_{\btheta} + \frac{n\log(nd)}{n+k} r_{\btheta}^2\right),
        \end{align*}
        and finally,
        \begin{align*}
	&\Bigg|\!\Bigg|\!\Bigg|\frac1{n+k-1}\Bigg(\sum_{i=1}^n\left(\nabla\cL(\btheta;Z_{i1})-\nabla\cL_N(\btheta)\right)\left(\nabla\cL(\btheta;Z_{i1})-\nabla\cL_N(\btheta)\right)^\top \\
	&\quad+\sum_{j=2}^k n\left(\nabla\cL_j(\btheta)-\nabla\cL_N(\btheta)\right)\left(\nabla\cL_j(\btheta)-\nabla\cL_N(\btheta)\right)^\top\Bigg) -\Ee\left[\nabla\cL(\thetas;Z)\nabla\cL(\thetas;Z)^\top\right]\Bigg|\!\Bigg|\!\Bigg|_{\max} \\
	&=O_P\Bigg(\sqrt{\frac{\log d}{n+k}} + \frac{\log^2(d(n+k))\log d}{n+k} + \left(\left(1+\sqrt{\frac{\log d}N}\right)\frac{nk}{n+k}+\log((n+k)d)\right)r_{\btheta}^2 \\
	&\quad + \Bigg(\sqrt{\log((n+k)d)} + \frac{\log^{1/4} d\sqrt{\log((n+k)d)}}{(n+k)^{1/4}} + \sqrt{\frac{\log^3(d(n+k))\log d}{n+k}}\Bigg) r_{\btheta}\Bigg).
	\end{align*}
	
\end{proof}

\begin{lemma}\label{lem:vcov0_reg_glm}
	In sparse GLM, under Assumptions~\ref{as:smth_glm}--\ref{as:hes_glm}, provided that $\left\|\btheta-\thetas\right\|_1=O_P(r_{\btheta})$, we have that
	\begin{align*}
	&\matrixnorm{\frac1k\sum_{j=1}^k n\left(\nabla\cL_j(\btheta)-\nabla\cL_N(\btheta)\right)\left(\nabla\cL_j(\btheta)-\nabla\cL_N(\btheta)\right)^\top-\Ee\left[\nabla\cL(\thetas;Z)\nabla\cL(\thetas;Z)^\top\right]}_{\max} \\
	&= O_P\left(\sqrt{\frac{\log d}k} + \frac{\log d}k + \left(1+\left(\frac{\log d}k\right)^{1/4}\right)\left(\sqrt{\log d} + \sqrt{n}r_{\btheta}\right) r_{\btheta} + \left(n + \log d + nr_{\btheta}^2\right) r_{\btheta}^2\right).
	\end{align*}
\end{lemma}

\begin{proof}[Lemma \ref{lem:vcov0_reg_glm}]
	By Lemma~\ref{lem:kgrad_decomp}, it suffices to bound $U_1(\btheta)$, $U_2$, and $U_3(\btheta)$.  We begin by bounding $U_2$.  Using the argument for obtaining \eqref{eqn:hoef1}, we have that for any $t>0$,
	$$P\left(\left|\nabla\cL_j(\thetas)_l\right|>t\right)\leq2\exp\left(-\frac{nt^2}c\right),$$
	and then,
	$$P\left(\sqrt n\left|\nabla\cL_j(\thetas)_l\right|>t\right)\leq2\exp\left(-\frac{t^2}c\right),$$
	that is, $\sqrt n \nabla\cL_j(\thetas)_l$ is sub-Gaussian with $O(1)$ $\psi_2$-norm.  Therefore, $n \nabla\cL_j(\thetas)_l\nabla\cL_j(\thetas)_{l'}$ is sub-exponential with $O(1)$ $\psi_1$-norm.  Note that $\Ee[n \nabla\cL_j(\thetas)_l\nabla\cL_j(\thetas)_{l'}]=\Ee[\nabla\cL(\thetas;Z)_l\nabla\cL(\thetas;Z)_{l'}]$.  Then, we apply Bernstein's inequality and obtain that for any $t>0$,
	$$P\left(\left|\frac1k\sum_{j=1}^k n \nabla\cL_j(\thetas)_l\nabla\cL_j(\thetas)_{l'}-\Ee\left[\nabla\cL(\thetas;Z)_l\nabla\cL(\thetas;Z)_{l'}\right]\right| > t\right)\leq 2\exp\left(-ck\left(t^2\wedge t\right)\right),$$
	or, for any $\delta\in(0,1)$,
	$$P\left(\left|\frac1k\sum_{j=1}^k n \nabla\cL_j(\thetas)_l\nabla\cL_j(\thetas)_{l'}-\Ee\left[\nabla\cL(\thetas;Z)_l\nabla\cL(\thetas;Z)_{l'}\right]\right|>\sqrt{\frac{\log\frac{2d^2}\delta}{ck}}\vee\frac{\log\frac{2d^2}\delta}{ck}\right)\leq\frac\delta{d^2},$$
	and by the union bound, with probability at least $1-\delta$,
	$$U_2\leq\sqrt{\frac{\log\frac{2d^2}\delta}{ck}}\vee\frac{\log\frac{2d^2}\delta}{ck},$$
	which implies that
	$$U_2=O_P\left(\sqrt{\frac{\log d}k}\right).$$
	
	Next, we bound $U_3(\btheta)$.  By the triangle inequality, we have that
	\begin{align*}
	\left\|\nabla\cL_N(\btheta)-\nabla\cLs(\btheta)\right\|_\infty
	&\leq \left\|\nabla\cL_N(\btheta)-\nabla\cL_N(\thetas)\right\|_\infty + \left\|\nabla\cL_N(\thetas)\right\|_\infty + \left\|\nabla\cLs(\btheta)\right\|_\infty.
	\end{align*}
	By \eqref{eqn:graddiff}, we have that
	\begin{align*}
	\nabla\cL_N(\btheta)-\nabla\cL_N(\thetas)&=\int_0^1\nabla^2\cL_N(\thetas+t(\btheta-\thetas))dt(\btheta-\thetas) \\
	&=\int_0^1\frac1N\sum_{i=1}^n\sum_{j=1}^k g''(y_{ij},x_{ij}^\top(\thetas+t(\btheta-\thetas)))x_{ij}x_{ij}^\top dt(\btheta-\thetas),
	\end{align*}
	and then, under Assumptions~\ref{as:smth_glm}~and~\ref{as:design_glm},
	\begin{align*}
	\left\|\nabla\cL_N(\btheta)-\nabla\cL_N(\thetas)\right\|_\infty&=\int_0^1\frac1N\sum_{i=1}^n\sum_{j=1}^k \left|g''(y_{ij},x_{ij}^\top(\thetas+t(\btheta-\thetas)))\right|\left\|x_{ij}\right\|_\infty^2 dt \left\|\btheta-\thetas\right\|_\infty \\
	&\lesssim \left\|\btheta-\thetas\right\|_\infty.
	\end{align*}
	Note that for any $\theta$,
	\begin{align*}
	\left\|\nabla\cLs(\theta)\right\|_\infty&=\left\|\nabla\cLs(\theta)-\nabla\cLs(\thetas)\right\|_\infty \\
	&=\left\|\Ee\left[\left(g'(y,x^\top\theta)-g'(y,x^\top\thetas))\right)x\right]\right\|_\infty \\
	&=\left\|\Ee\left[\int_0^1 g''(y,x^\top(\thetas+t(\theta-\thetas)))dt xx^\top(\theta-\thetas)\right]\right\|_\infty \\
	&\leq\Ee\left[\int_0^1 \left|g''(y,x^\top(\thetas+t(\theta-\thetas)))\right|dt \left\|x\right\|_\infty^2 \left\|\theta-\thetas\right\|_\infty\right] \\
	&\lesssim \left\|\theta-\thetas\right\|_\infty.
	\end{align*}
	Therefore, $$\left\|\nabla\cLs(\btheta)\right\|_\infty\lesssim \left\|\btheta-\thetas\right\|_\infty.$$
	By \eqref{eqn:hoef2}, we have that
	$$\left\|\nabla\cL_N(\thetas)\right\|_\infty=O_P\left(\sqrt{\frac{\log d}N}\right).$$
	Then, assuming that $\left\|\btheta-\thetas\right\|_1=O_P(r_{\btheta})$, we have that
	\begin{align*}
	\left\|\nabla\cL_N(\btheta)-\nabla\cLs(\btheta)\right\|_\infty
	&=O_P\left(r_{\btheta}+\sqrt{\frac{\log d}N}\right),
	\end{align*}
	and then,
	$$U_3(\btheta) = O_P\left(n r_{\btheta}^2 + \frac{\log d}k\right).$$
	
	Lastly, we bound $U_1(\btheta)$.  As in the proof of Lemma~\ref{lem:vcov0_reg}, we have that
        \begin{align*}
        U_1(\btheta) &\leq \matrixnorm{\frac1k\sum_{j=1}^k n\left(\nabla\cL_j(\btheta)-\nabla\cLs(\btheta)-\nabla\cL_j(\thetas)\right)\left(\nabla\cL_j(\btheta)-\nabla\cLs(\btheta)-\nabla\cL_j(\thetas)\right)^\top}_{\max} \\
        &\quad+ 2\matrixnorm{\frac1k\sum_{j=1}^k n\nabla\cL_j(\thetas)\left(\nabla\cL_j(\btheta)-\nabla\cLs(\btheta)-\nabla\cL_j(\thetas)\right)^\top}_{\max} \\
        &\defn U_{11}(\btheta) + 2 U_{12}(\btheta),
        \end{align*}
        and
        \begin{align*}
        U_{12}(\btheta) &\leq \matrixnorm{\frac1k\sum_{j=1}^k n\nabla\cL_j(\thetas)\nabla\cL_j(\thetas)^\top}_{\max}^{1/2} U_{11}(\btheta)^{1/2}.
        \end{align*}
        Note that $\matrixnorm{\Ee\left[\nabla\cL(\thetas;Z)\nabla\cL(\thetas;Z)^\top\right]}_{\max} = O(1)$ under Assumption~\ref{as:hes_glm}.  Then, by the triangle inequality, we have that
        \begin{align*}
        &\matrixnorm{\frac1k\sum_{j=1}^k n\nabla\cL_j(\thetas)\nabla\cL_j(\thetas)^\top}_{\max} \\
        &\leq \matrixnorm{\frac1k\sum_{j=1}^k n\nabla\cL_j(\thetas)\nabla\cL_j(\thetas)^\top-\Ee\left[\nabla\cL(\thetas;Z)\nabla\cL(\thetas;Z)^\top\right]}_{\max} + \matrixnorm{\Ee\left[\nabla\cL(\thetas;Z)\nabla\cL(\thetas;Z)^\top\right]}_{\max} \\
        &=U_2 + \matrixnorm{\Ee\left[\nabla\cL(\thetas;Z)\nabla\cL(\thetas;Z)^\top\right]}_{\max} \\
        &= O_P\left(1+\sqrt{\frac{\log d}k}\right).
        \end{align*}
        It remains to bound $U_{11}(\btheta)$.  Note that
	\begin{align*}
	\nabla\cL_j(\btheta)-\nabla\cL_j(\thetas)&=\int_0^1\nabla^2\cL_j(\thetas+t(\btheta-\thetas))dt(\btheta-\thetas) \\
	&=\int_0^1\frac1n\sum_{i=1}^n g''(y_{ij},x_{ij}^\top(\thetas+t(\btheta-\thetas)))x_{ij}x_{ij}^\top dt(\btheta-\thetas),
	\end{align*}
	and
	\begin{align*}
	g''(y_{ij},x_{ij}^\top(\thetas+t(\btheta-\thetas)))&=g''(y_{ij},x_{ij}^\top\thetas)+\int_0^1 g'''(y_{ij},x_{ij}^\top(\thetas+st(\btheta-\thetas)))ds x_{ij}^\top(t(\btheta-\thetas)),
	\end{align*}
	and then
	\begin{align*}
	\nabla\cL_j(\btheta)-\nabla\cL_j(\thetas)&=\frac1n\sum_{i=1}^n g''(y_{ij},x_{ij}^\top\thetas)x_{ij}x_{ij}^\top (\btheta-\thetas) \\
	&\quad+ \int_0^1\int_0^1 \frac1n\sum_{i=1}^n g'''(y_{ij},x_{ij}^\top(\thetas+st(\btheta-\thetas))) x_{ij}^\top t(\btheta-\thetas) x_{ij}x_{ij}^\top dtds (\btheta-\thetas).
	\end{align*}
	In a similar way, we have that
	\begin{align*}
	\nabla\cLs(\btheta)&=\nabla\cLs(\btheta)-\nabla\cLs(\thetas) \\
	&= \Ee\left[g''(y,x^\top\thetas)xx^\top\right] (\btheta-\thetas) \\
	&\quad+ \int_0^1\int_0^1 \Ee_{x,y}\left[g'''(y,x^\top(\thetas+st(\btheta-\thetas))) x^\top t(\btheta-\thetas) xx^\top\right] dtds (\btheta-\thetas),
	\end{align*}
	and then,
	\begin{align*}
        \nabla\cL_j(\btheta)-\nabla\cLs(\btheta)-\nabla\cL_j(\thetas)
        &=\left(\frac1n\sum_{i=1}^n g''(y_{ij},x_{ij}^\top\thetas)x_{ij}x_{ij}^\top-\Ee\left[g''(y,x^\top\thetas)xx^\top\right]\right) (\btheta-\thetas) \\
	&\quad+ \bigg(\int_0^1\int_0^1 \frac1n\sum_{i=1}^n g'''(y_{ij},x_{ij}^\top(\thetas+st(\btheta-\thetas))) x_{ij}^\top t(\btheta-\thetas) x_{ij}x_{ij}^\top \\
	&\quad - \Ee_{x,y}\left[g'''(y,x^\top(\thetas+st(\btheta-\thetas))) x^\top t(\btheta-\thetas) xx^\top\right] dtds (\btheta-\thetas)\bigg) \\
	&\defn U_{111,j}(\btheta)+U_{112,j}(\btheta).
        \end{align*}
        Then, we have by the triangle inequality that
        	\begin{align*}
	U_{11}(\btheta)&=\matrixnorm{\frac1k\sum_{j=1}^k n\left(U_{111,j}(\btheta)+U_{112,j}(\btheta)\right)\left(U_{111,j}(\btheta)+U_{112,j}(\btheta)\right)^\top}_{\max} \\
	&\leq \frac1k\sum_{j=1}^k n \matrixnorm{\left(U_{111,j}(\btheta)+U_{112,j}(\btheta)\right)\left(U_{111,j}(\btheta)+U_{112,j}(\btheta)\right)^\top}_{\max} \\
	&= \frac1k\sum_{j=1}^k n \left\|U_{111,j}(\btheta)+U_{112,j}(\btheta)\right\|_\infty^2 \\
	&\leq \frac2k\sum_{j=1}^k n \left(\left\|U_{111,j}(\btheta)\right\|_\infty^2 + \left\|U_{112,j}(\btheta)\right\|_\infty^2\right)
	\end{align*}
	Using the argument for obtaining \eqref{eqn:hoef3}, we have that
	\begin{align*}
	\left\|U_{111,j}(\btheta)\right\|_\infty&=\left\|\left(\nabla^2\cL_j(\thetas)-\nabla^2\cLs(\thetas)\right)(\btheta-\thetas)\right\|_\infty \\
	&\leq\matrixnorm{\nabla^2\cL_j(\thetas)-\nabla^2\cLs(\thetas)}_{\max}\left\|\btheta-\thetas\right\|_1 \\
	&=O_P\left(\sqrt{\frac{\log d}n}\right) O_P\left(r_{\btheta}\right) \\
	&=O_P\left(\sqrt{\frac{\log d}n} r_{\btheta}\right).
	\end{align*}
	Under Assumptions~\ref{as:smth_glm}~and~\ref{as:design_glm}, we have that
	\begin{align*}
	\left\|U_{112,j}(\btheta)\right\|_\infty&\leq\int_0^1\int_0^1 \frac1n\sum_{i=1}^n \left|g'''(y_{ij},x_{ij}^\top(\thetas+st(\btheta-\thetas)))\right| \left\|x_{ij}\right\|_\infty t\left\|\btheta-\thetas\right\|_1 \left\|x_{ij}\right\|_\infty^2 \\
	&\quad + \Ee_{x,y}\left[\left|g'''(y,x^\top(\thetas+st(\btheta-\thetas)))\right| \left\|x\right\|_\infty t\left\|\btheta-\thetas\right\|_1 \left\|x\right\|_\infty^2\right] dtds \left\|\btheta-\thetas\right\|_1 \\
	&\lesssim \left\|\btheta-\thetas\right\|_1^2 \\
	&=O_P\left(r_{\btheta}^2\right).
	\end{align*}
	Hence, we have that
	$$U_{11}(\btheta) = n\left(O_P\left(\frac{\log d}n r_{\btheta}^2\right) + O_P\left(r_{\btheta}^4\right)\right) = O_P\left(\left(\log d + nr_{\btheta}^2\right) r_{\btheta}^2\right).$$
        Putting all the preceding bounds together, we obtain that
        $$U_{12}(\btheta) = O_P\left(\left(1+\left(\frac{\log d}k\right)^{1/4}\right)\left(\sqrt{\log d} + \sqrt{n}r_{\btheta}\right) r_{\btheta}\right),$$
        $$U_1(\btheta) = O_P\left(\left(1+\left(\frac{\log d}k\right)^{1/4}\right)\left(\sqrt{\log d} + \sqrt{n}r_{\btheta}\right) r_{\btheta} + \left(\log d + nr_{\btheta}^2\right) r_{\btheta}^2\right),$$
        and finally,
        \begin{align*}
	&\matrixnorm{\frac1k\sum_{j=1}^k n\left(\nabla\cL_j(\btheta)-\nabla\cL_N(\btheta)\right)\left(\nabla\cL_j(\btheta)-\nabla\cL_N(\btheta)\right)^\top-\Ee\left[\nabla\cL(\thetas;Z)\nabla\cL(\thetas;Z)^\top\right]}_2 \\
	&=O_P\left(\sqrt{\frac{\log d}k} + \frac{\log d}k + \left(1+\left(\frac{\log d}k\right)^{1/4}\right)\left(\sqrt{\log d} + \sqrt{n}r_{\btheta}\right) r_{\btheta} + \left(n + \log d + nr_{\btheta}^2\right) r_{\btheta}^2\right).
	\end{align*}
\end{proof}

\begin{lemma}\label{lem:vcov_reg_glm}
	In sparse GLM, under Assumptions~\ref{as:smth_glm}--\ref{as:hes_glm}, provided that $\left\|\btheta-\thetas\right\|_1=O_P(r_{\btheta})$, we have that
	\begin{align*}
	&\Bigg|\!\Bigg|\!\Bigg|\frac1{n+k-1}\Bigg(\sum_{i=1}^n\left(\nabla\cL(\btheta;Z_{i1})-\nabla\cL_N(\btheta)\right)\left(\nabla\cL(\btheta;Z_{i1})-\nabla\cL_N(\btheta)\right)^\top \\
	&\quad+\sum_{j=2}^k n\left(\nabla\cL_j(\btheta)-\nabla\cL_N(\btheta)\right)\left(\nabla\cL_j(\btheta)-\nabla\cL_N(\btheta)\right)^\top\Bigg) -\Ee\left[\nabla\cL(\thetas;Z)\nabla\cL(\thetas;Z)^\top\right]\Bigg|\!\Bigg|\!\Bigg|_{\max} \\
	&= O_P\Bigg(\sqrt{\frac{\log d}{n+k}} + \frac{\log d}{n+k} + \frac{nk}{n+k}r_{\btheta}^2 + \left(1+\left(\frac{\log d}n\right)^{1/4}\right)\frac{n}{n+k} \left(r_{\btheta} + r_{\btheta}^2\right) + \frac{n}{n+k} r_{\btheta}^4 \\
	&\quad + \left(1+\left(\frac{\log d}k\right)^{1/4}\right)\frac{k\sqrt{\log d}+k\sqrt n r_{\btheta}}{n+k} r_{\btheta} + \frac{k\log d+knr_{\btheta}^2}{n+k} r_{\btheta}^2\Bigg).
	\end{align*}
\end{lemma}

\begin{proof}[Lemma \ref{lem:vcov_reg_glm}]
	By Lemma~\ref{lem:nk1grad_decomp}, it suffices to bound $V_1(\btheta)$, $V_1'(\btheta)$, $V_2$, $V_2'$, and $V_3(\btheta)$.  By the proof of Lemma~\ref{lem:vcov0_reg_glm}, we have that under Assumptions~\ref{as:smth_glm}--\ref{as:hes_glm}, assuming that $\left\|\btheta-\thetas\right\|_1=O_P(r_{\btheta})$,
	\begin{align*}
	V_1(\btheta) &= \frac{k-1}{n+k-1} O_P\left(\left(1+\left(\frac{\log d}k\right)^{1/4}\right)\left(\sqrt{\log d} + \sqrt{n}r_{\btheta}\right) r_{\btheta} + \left(\log d + nr_{\btheta}^2\right) r_{\btheta}^2\right) \\
	&= O_P\left(\left(1+\left(\frac{\log d}k\right)^{1/4}\right)\frac{k\sqrt{\log d}+k\sqrt n r_{\btheta}}{n+k} r_{\btheta} + \frac{k\log d+knr_{\btheta}^2}{n+k} r_{\btheta}^2\right),
	\end{align*}
	$$V_2 = \frac{k-1}{n+k-1} O_P\left(\sqrt{\frac{\log d}k}\right) = O_P\left(\frac{\sqrt{k\log d}}{n+k}\right),$$
	and
	$$V_3(\btheta) =  \frac{nk}{n+k-1} O_P\left(r_{\btheta}^2+\frac{\log d}N\right) = O_P\left(\frac{nk}{n+k}r_{\btheta}^2 + \frac{\log d}{n+k}\right).$$
	It remains to bound $V_1'(\btheta)$ and $V_2'$.
	
	To bound $V_2'$, we note that each $\nabla\cL(\thetas;Z_{i1})_l\nabla\cL(\thetas;Z_{i1})_{l'}=g'(y_{i1},x_{i1}^\top\thetas)^2 x_{i1,l} x_{i1,l'}$ is bounded under Assumptions~\ref{as:smth_glm}~and~\ref{as:design_glm}.  Applying Hoeffding's inequality, we obtain that for any $t>0$
	$$P\left(\left|\frac1n\sum_{i=1}^n\nabla\cL(\thetas;Z_{i1})_l\nabla\cL(\thetas;Z_{i1})_{l'}-\Ee\left[\nabla\cL(\thetas;Z)_l\nabla\cL(\thetas;Z)_{l'}\right]\right| > t\right)\leq 2\exp\left(-\frac{nt^2}c\right),$$
	for some constant $c$, or, for any $\delta\in(0,1)$,
	$$P\left(\left|\frac1n\sum_{i=1}^n\nabla\cL(\thetas;Z_{i1})_l\nabla\cL(\thetas;Z_{i1})_{l'}-\Ee\left[\nabla\cL(\thetas;Z)_l\nabla\cL(\thetas;Z)_{l'}\right]\right| > \sqrt{\frac{c\log\frac{2d^2}\delta}n}\right)\leq\frac\delta{d^2},$$
	and by the union bound, with probability at least $1-\delta$,
	$$\matrixnorm{\frac1n\sum_{i=1}^n\nabla\cL(\thetas;Z_{i1})\nabla\cL(\thetas;Z_{i1})^\top-\Ee\left[\nabla\cL(\thetas;Z)\nabla\cL(\thetas;Z)^\top\right]}_{\max}\leq\sqrt{\frac{c\log\frac{2d^2}\delta}n},$$
	which implies that
	$$V_2'=\frac n{n+k-1} O_P\left(\sqrt{\frac{\log d}n}\right)=O_P\left(\sqrt{\frac{n\log d}{n+k}}\right).$$
	
	Lastly, we bound $V_1'(\btheta)$.  As in the proof of Lemma~\ref{lem:vcov_reg}, we have that
	\begin{align*}
        \frac{n+k-1}n V_1'(\btheta) &\leq \matrixnorm{\frac1n\sum_{i=1}^n \left(\nabla\cL(\theta;Z_{i1})-\nabla\cLs(\theta)-\nabla\cL(\thetas;Z_{i1})\right)\left(\nabla\cL(\theta;Z_{i1})-\nabla\cLs(\theta)-\nabla\cL(\thetas;Z_{i1})\right)^\top}_{\max} \\
        &\quad+ 2\matrixnorm{\frac1n\sum_{i=1}^n \nabla\cL(\thetas;Z_{i1})\left(\nabla\cL(\theta;Z_{i1})-\nabla\cLs(\theta)-\nabla\cL(\thetas;Z_{i1})\right)^\top}_{\max} \\
        &\defn V_{11}'(\btheta) + 2 V_{12}'(\btheta),
        \end{align*}
        and
        \begin{align*}
        V_{12}'(\btheta) &\leq \matrixnorm{\frac1n\sum_{i=1}^n \nabla\cL(\thetas;Z_{i1})\nabla\cL(\thetas;Z_{i1})^\top}_{\max}^{1/2}V_{11}'(\btheta)^{1/2}.
        \end{align*}
        Note that $\matrixnorm{\Ee\left[\nabla\cL(\thetas;Z)\nabla\cL(\thetas;Z)^\top\right]}_{\max} = O(1)$ under Assumption~\ref{as:hes_glm}.  Then, by the triangle inequality, we have that
        \begin{align*}
        &\matrixnorm{\frac1n\sum_{i=1}^n \nabla\cL(\thetas;Z_{i1})\nabla\cL(\thetas;Z_{i1})^\top}_{\max} \\
        &\leq \matrixnorm{\frac1n\sum_{i=1}^n \nabla\cL(\thetas;Z_{i1})\nabla\cL(\thetas;Z_{i1})^\top-\Ee\left[\nabla\cL(\thetas;Z)\nabla\cL(\thetas;Z)^\top\right]}_{\max} + \matrixnorm{\Ee\left[\nabla\cL(\thetas;Z)\nabla\cL(\thetas;Z)^\top\right]}_{\max} \\
        &= \frac{n+k-1}n V_2' + \matrixnorm{\Ee\left[\nabla\cL(\thetas;Z)\nabla\cL(\thetas;Z)^\top\right]}_{\max} \\
        &= O_P\left(1+\sqrt{\frac{\log d}n}\right).
        \end{align*}
        It remains to bound $V_{11}'(\btheta)$.  Using the same argument for analyzing $\nabla\cL_j(\btheta)-\nabla\cLs(\btheta)-\nabla\cL_j(\thetas)$ in the proof of Lemma~\ref{lem:vcov0_reg_glm}, we obtain that
        \begin{align*}
        \nabla\cL(\theta;Z_{i1})-\nabla\cLs(\theta)-\nabla\cL(\thetas;Z_{i1})
        &=\left(g''(y_{i1},x_{i1}^\top\thetas)x_{i1}x_{i1}^\top-\Ee\left[g''(y,x^\top\thetas)xx^\top\right]\right) (\btheta-\thetas) \\
	&\quad+ \bigg(\int_0^1\int_0^1 g'''(y_{i1},x_{i1}^\top(\thetas+st(\btheta-\thetas))) x_{i1}^\top t(\btheta-\thetas) x_{i1}x_{i1}^\top \\
	&\quad - \Ee_{x,y}\left[g'''(y,x^\top(\thetas+st(\btheta-\thetas))) x^\top t(\btheta-\thetas) xx^\top\right] dtds (\btheta-\thetas)\bigg) \\
	&\defn V_{111,i}'(\btheta)+V_{112,i}'(\btheta),
        \end{align*}
        and
        	\begin{align*}
	V_{11}'(\btheta)&=\matrixnorm{\frac1n\sum_{i=1}^n \left(V_{111,i}'(\btheta)+V_{112,i}'(\btheta)\right)\left(V_{111,i}'(\btheta)+V_{112,i}'(\btheta)\right)^\top}_{\max} \\
	&\leq \frac1n\sum_{i=1}^n \matrixnorm{\left(V_{111,i}'(\btheta)+V_{112,i}'(\btheta)\right)\left(V_{111,i}'(\btheta)+V_{112,i}'(\btheta)\right)^\top}_{\max} \\
	&= \frac1n\sum_{i=1}^n \left\|V_{111,i}'(\btheta)+V_{112,i}'(\btheta)\right\|_\infty^2 \\
	&\leq \frac2n\sum_{i=1}^n \left(\left\|V_{111,i}'(\btheta)\right\|_\infty^2 + \left\|V_{112,i}'(\btheta)\right\|_\infty^2\right).
	\end{align*}
	Moreover, under Assumptions~\ref{as:smth_glm}--\ref{as:hes_glm}, we have that
	\begin{align*}
	\left\|V_{111,i}'(\btheta)\right\|_\infty&=\left\|\left(\nabla^2\cL(\thetas;Z_{i1})-\nabla^2\cLs(\thetas)\right)(\btheta-\thetas)\right\|_\infty \\
	&\leq\matrixnorm{\nabla^2\cL(\thetas;Z_{i1})-\nabla^2\cLs(\thetas)}_{\max}\left\|\btheta-\thetas\right\|_1 \\
	&\leq\left(\left|g''(y_{i1},x_{i1}^\top\thetas)\right|\left\|x_{i1}\right\|_\infty^2+\matrixnorm{\nabla^2\cLs(\thetas)}_{\max}\right)\left\|\btheta-\thetas\right\|_1 \\
	&=O_P\left(r_{\btheta}\right),
	\end{align*}
	and
	\begin{align*}
	\left\|V_{112,i}'(\btheta)\right\|_\infty&\leq\int_0^1\int_0^1  \left|g'''(y_{i1},x_{i1}^\top(\thetas+st(\btheta-\thetas)))\right| \left\|x_{i1}\right\|_\infty t\left\|\btheta-\thetas\right\|_1 \left\|x_{i1}\right\|_\infty^2 \\
	&\quad + \Ee_{x,y}\left[\left|g'''(y,x^\top(\thetas+st(\btheta-\thetas)))\right| \left\|x\right\|_\infty t\left\|\btheta-\thetas\right\|_1 \left\|x\right\|_\infty^2\right] dtds \left\|\btheta-\thetas\right\|_1 \\
	&\lesssim \left\|\btheta-\thetas\right\|_1^2 \\
	&=O_P\left(r_{\btheta}^2\right),
	\end{align*}
	and hence,
	$$V_{11}'(\btheta) = O_P\left(r_{\btheta}^2 + r_{\btheta}^4\right).$$
	
	Putting all the preceding bounds together, we obtain that
        $$V_{12}'(\btheta) = O_P\left(\left(1+\left(\frac{\log d}n\right)^{1/4}\right)\left(r_{\btheta} + r_{\btheta}^2\right)\right),$$
        \begin{align*}
        V_1'(\btheta) &= \frac n{n+k-1} O_P\left(\left(1+\left(\frac{\log d}n\right)^{1/4}\right)\left(r_{\btheta} + r_{\btheta}^2\right) + r_{\btheta}^2 + r_{\btheta}^4\right) \\
        &= O_P\left(\left(1+\left(\frac{\log d}n\right)^{1/4}\right)\frac{n}{n+k} \left(r_{\btheta} + r_{\btheta}^2\right) + \frac{n}{n+k} r_{\btheta}^4\right),
        \end{align*}
        and finally,
        \begin{align*}
	&\Bigg|\!\Bigg|\!\Bigg|\frac1{n+k-1}\Bigg(\sum_{i=1}^n\left(\nabla\cL(\btheta;Z_{i1})-\nabla\cL_N(\btheta)\right)\left(\nabla\cL(\btheta;Z_{i1})-\nabla\cL_N(\btheta)\right)^\top \\
	&\quad+\sum_{j=2}^k n\left(\nabla\cL_j(\btheta)-\nabla\cL_N(\btheta)\right)\left(\nabla\cL_j(\btheta)-\nabla\cL_N(\btheta)\right)^\top\Bigg) -\Ee\left[\nabla\cL(\thetas;Z)\nabla\cL(\thetas;Z)^\top\right]\Bigg|\!\Bigg|\!\Bigg|_{\max} \\
	&=O_P\Bigg(\sqrt{\frac{\log d}{n+k}} + \frac{\log d}{n+k} + \frac{nk}{n+k}r_{\btheta}^2 + \left(1+\left(\frac{\log d}n\right)^{1/4}\right)\frac{n}{n+k} \left(r_{\btheta} + r_{\btheta}^2\right) + \frac{n}{n+k} r_{\btheta}^4 \\
	&\quad + \left(1+\left(\frac{\log d}k\right)^{1/4}\right)\frac{k\sqrt{\log d}+k\sqrt n r_{\btheta}}{n+k} r_{\btheta} + \frac{k\log d+knr_{\btheta}^2}{n+k} r_{\btheta}^2\Bigg).
	\end{align*}
	
\end{proof}

\begin{lemma}\label{lem:hes_hd}
	In high-dimensional linear model, under Assumption~\ref{as:design}, if $n\gg  {s^*} \log d$, we have that
	$$\matrixnorm{\tT}_\infty=O_P\left( \sqrt{{s^*}}\right),\quad\matrixnorm{\tT-\Theta}_\infty=O_P\left(  {s^*}\sqrt{\frac{\log d}n}\right),$$
	$$\matrixnorm{\tT\frac{X_1^\top X_1}n-I_d}_{\max}=O_P\left(\sqrt{\frac{\log d}n}\right),\quad\text{and}\quad\max_l\left\|\tT_l-\Theta_l\right\|_2=O_P\left( \sqrt{\frac{{s^*}\log d}n}\right).$$
\end{lemma}

\begin{proof}[Lemma \ref{lem:hes_hd}]
	In the high-dimensional setting, $\tT$ is constructed using nodewise Lasso.  We obtain the bounds in the lemma from the proof of Lemma 5.3 and Theorem 2.4 of \citet{van2014asymptotically}.
\end{proof}

\begin{lemma}\label{lem:hes_hd_glm}
	In high-dimensional GLM, under Assumptions~\ref{as:smth_glm}--\ref{as:hes_glm}, if $n\gg s_0^2\log^2 d+{s^*}^2\log d$, we have that
	$$\matrixnorm{\tT(\ttheta^{(0)})}_\infty=O_P\left( \sqrt{{s^*}}\right),\quad\matrixnorm{\tT(\ttheta^{(0)})-\Theta}_\infty=O_P\left(\left(s_0 +   {s^*}\right)\sqrt{\frac{\log d}n}\right),$$
	$$\matrixnorm{\tT(\ttheta^{(0)})\nabla^2\cL_1(\ttheta^{(0)})-I_d}_{\max}=O_P\left(\sqrt{\frac{\log d}n}\right),\quad\text{and}\quad$$
	$$\max_l\left\|\tT(\ttheta^{(0)})_l-\Theta_l\right\|_2=O_P\left(\sqrt{\frac{(s_0+{s^*})\log d}n}\right).$$
\end{lemma}

\begin{proof}[Lemma \ref{lem:hes_hd_glm}]
	In the high-dimensional setting, $\tT(\ttheta^{(0)})$ is constructed using nodewise Lasso.  We obtain the bounds in the lemma from Theorem~3.2 and the proof of Theorems~3.1~and~3.3 of \citet{van2014asymptotically}.
\end{proof}

\end{document}